\documentclass[letterpaper, 10 pt, conference]{ieeeconf}  

\IEEEoverridecommandlockouts                              

\usepackage{caption}
\usepackage{graphicx} 
\usepackage{subcaption}
\usepackage{url}
\usepackage{hyperref}
\usepackage{xcolor}
\usepackage{float}
\usepackage{algorithmic}
\usepackage[linesnumbered,ruled,vlined]{algorithm2e}
\SetCommentSty{emph}
\usepackage{amssymb,amsmath,comment,cite}
\usepackage{amsfonts}
\usepackage{xspace}
\usepackage{booktabs}
\usepackage{multicol}
\usepackage{multirow}
\usepackage{cleveref}

\usepackage[modulo,switch]{lineno} 
\renewcommand{\baselinestretch}{0.98}

\newcommand{\parameter}[1]{\textcolor{red!70!black}{#1}}

\begin{document}
\newtheorem{theorem}{Theorem}
\newtheorem{lemma}{Lemma}
\newtheorem{definition}{Definition}
\newtheorem{corollary}{Corollary}

\newcommand{\warehouseSmallR}{\texttt{warehouse-20-17}\xspace}
\newcommand{\warehouselargeW}{\texttt{warehouse-33-36}\xspace}
\newcommand{\warehouseXlarge}{\texttt{warehouse-10-20-10-2-1\xspace}}
\newcommand{\randomSmall}{\texttt{random-32-32-20}\xspace}
\newcommand{\roomLarge}{\texttt{room-64-64-8}\xspace}
\newcommand{\mazeSmall}{\texttt{maze-32-32-4}\xspace}
\newcommand{\emptyMid}{\texttt{paris-1-256}\xspace}
\newcommand{\randomLarge}{\texttt{random-64-64-20}\xspace}
\newcommand{\denSmall}{\texttt{den312d}\xspace}
\newcommand{\paris}{\texttt{Paris-1-256}\xspace}
\newcommand{\randomLargeLA}{\texttt{random-128-128-PMLA}\xspace}
\newcommand{\roomLargeLA}{\texttt{room-128-128-PMLA}\xspace}
\newcommand{\mazeLargeLA}{\texttt{maze-128-128-PMLA}\xspace}
\newcommand{\emptyLargeLA}{\texttt{empty-64-64-PMLA}\xspace}




\title{\LARGE \bf
Planning over MAPF Agent Dependencies via Multi-Dependency PIBT}


\author{
Zixiang Jiang$^{1*}$, Yulun Zhang$^{2*}$, Rishi Veerapaneni$^{2*}$, Jiaoyang Li$^{2}$
\thanks{$^*$These authors contributed equally.}
\thanks{$^{1}$Zixiang Jiang is with the University Of Melbourne. 
{\tt\small gmmichaeljiang@gmail.com}}%
\thanks{$^{2}$Yulun Zhang, Rishi Veerapaneni, and Jiaoyang Li are with the Robotics Institute, Carnegie Mellon University. 
{\tt\small \{yulunzhang,vrishi,jiaoyangli\}@cmu.edu}}%
}
\maketitle
\thispagestyle{empty}
\pagestyle{empty}

\begin{abstract}
Modern Multi-Agent Path Finding (MAPF) algorithms must plan for hundreds to thousands of agents in congested environments within a second, requiring highly efficient algorithms. Priority Inheritance with Backtracking (PIBT) is a popular algorithm capable of effectively planning in such situations. 
However, PIBT, and its variants like Enhanced PIBT (EPIBT), is constrained by its rule-based planning procedure and lacks generality because it restricts its search to paths that collide with at most one other agent. In this paper, we describe a new perspective on solving MAPF by planning over \textit{agent dependencies}. Taking inspiration from PIBT's priority inheritance logic, we define the concept of agent dependencies and propose Multi-Dependency PIBT (MD-PIBT) that searches over agent dependencies. MD-PIBT is a general framework where specific parameterizations can reproduce PIBT and EPIBT. At the same time, alternative configurations generalize PIBT and EPIBT to multi-step planning capable of reasoning paths that collide with more than one other agent.
Our experiments demonstrate that MD-PIBT effectively plans for as many as 10,000 homogeneous agents under various kinodynamic constraints, including pebble motion, rotation motion, and differential drive robots with speed and acceleration limits. 
We perform thorough evaluations on different variants of MAPF and find that MD-PIBT is particularly effective in MAPF with large agents. Our code is available at \url{https://github.com/lunjohnzhang/MD-PIBT}.
\end{abstract}

\section{Introduction}


Multi-Agent Path Finding (MAPF)~\cite{Stern2019benchmark} aims to move agents from their corresponding start to goal locations without collisions. Applications of MAPF such as autonomous warehouses~\cite{li2021lifelong} and robotic sorting systems~\cite{zhang2024tmo} could involve up to 4,000 ground robots moving in a shared environment~\cite{Brown2023amazonrobot}. These systems require efficient algorithms that return collision-free paths in less than a second. 

The state-of-the-art methods for fast and scalable planning leverage Priority Inheritance with Backtracking (PIBT)~\cite{okumura2022priority} due to its extreme speed and ability to find collision-free one-step solutions in congestion~\cite{Jiang2024Competition,jiang2025deploying10k,okumura2023lacam,yukhnevich2025epibt}. PIBT is a priority-based one-step planner capable of handling \emph{dependencies} between high-priority and low-priority agents. In particular, if a high-priority agent proposes a move that bumps into a low-priority agent, the low-priority agent inherits the high priority and forces other agents to make space for it. 

However, PIBT is fundamentally limited as it only reasons about one dependency at a time between a high-priority and low-priority agent. This means that PIBT cannot directly plan multi-step action sequences where one high-priority agent's path could collide with multiple low-priority agents. Existing work generalizing PIBT to multi-step planning~\cite{okumura2019winpibt,yukhnevich2025epibt} retains this limitation and cannot handle multiple dependencies between agents. 

This contrasts with Conflict-Based Search (CBS) \cite{sharon2015conflict}. CBS searches over space-time constraints, making it a general framework that can be applied to a variety of different dynamic models of MAPF agents~\cite{veerapaneni2025cbsprotocol}. Ideally, we would like to generalize PIBT to be a similar general framework. 

To that end, our contribution is designing Multi-Dependency PIBT (MD-PIBT), which formally introduces the idea of searching over \textit{agent dependencies} (as opposed to space-time constraints). This perspective enables planning over multi-step paths that can collide with multiple other agents and can directly work with different MAPF models. MD-PIBT is a general algorithm with many hyper-parameters that can be optimized and can replicate PIBT and EPIBT with specific instantiations. We evaluate MD-PIBT on several MAPF variants and show that MD-PIBT significantly outperforms PIBT and EPIBT when planning with large agents. Interestingly, however, on other MAPF instances of one-shot, lifelong MAPF with/without rotations, both PIBT and EPIBT are surprisingly strong, and MD-PIBT is able to replicate but not outperform their performance.


\section{Background}

\subsection{Multi-Agent Path Finding} \label{sec:background-mapf}

Multi-Agent Path Finding (MAPF) requires finding collision-free paths for a set of $N$ agents, denoted as $A = \{a_1,...,a_N\}$, where each agent must travel from its start location $s_i$ to its goal location $g_i$. In the standard 2D MAPF setup, agents move on a 4-connected grid graph discretized into evenly spaced cells. A collision-free solution consists of a set of paths $\Pi = \{ a_1.\pi, ..., a_N.\pi \}$ satisfying $a_i.\pi^0 = s_i$, $a_i.\pi^{T} = g_i$, where $T$ is the maximum timestep of all agents' paths. A collision-free solution must avoid vertex collisions (when two agents occupy the same cell at the same timestep) and edge collisions (when two agents swap positions between consecutive timesteps). 


In MAPF, agents can move according to different dynamic models. We consider three models: (1) Pebble Motion (PM)~\cite{Stern2019benchmark}, where agents move omni-directionally, (2) Pebble Motion with Large Agents (PMLA)~\cite{LiAAAI19a}, where agents have different sizes and move identically to PM, and (3) Rotation Motion (RM)~\cite{Jiang2024Competition}, where agents rotate in place or move forward. Additionally, planned paths can be executed using the Differential Drive Robots (DDR)~\cite{YanAndZhang2026LSMART} model, where agents rotate or move forward with speed and acceleration limits. In all models, agents can wait at their current cell. 

Our experiments evaluate on one-shot MAPF and lifelong MAPF. In one-shot MAPF the objective is to find a solution $\Pi$ that minimizes the total cost $|\Pi^{0:T}| = \sum_{i=1}^N |a_i.\pi^{0:T}| = \sum_{i=1}^N \sum_{t=0}^{T-1} c(a_i.\pi^t,a_i.\pi^{t+1})$. In this work, we assume that every action has unit cost $c(a_i.\pi^t,a_i.\pi^{t+1})=1$. When an agent remains at its goal, the action cost is 0.
Lifelong MAPF is a variant of MAPF where agents are reassigned a new goal each time they reach their current assigned goals. The objective is to maximize \emph{throughput}, the average number of goals reached per timestep. Lifelong MAPF is typically solved by decomposing it into a sequence of MAPF problems and solving them by searching for windowed paths~\cite{li2021lifelong}.

\subsection{MAPF Algorithms}
Search-based methods for MAPF are powerful and popular. 
CBS~\cite{sharon2015conflict} decomposes the MAPF instance into repeated single-agent calls and iteratively resolves collisions by searching over space-time constraints. CBS can be applied in many different domains and also has several algorithmic extensions~\cite{barer2014suboptimal,li2021eecbs}.
Another popular paradigm is priority-based algorithms, such as Prioritized Planning (PP)~\cite{erdmann1987multiple}, which plan each agent in a pre-defined priority but can easily lead to deadlocks. Priority Based Search (PBS)~\cite{ma2019searching} combines CBS and PP to search for priority orders of agents. 
In addition, learning-based methods are becoming increasingly popular. Methods leverage imitation learning~\cite{jiang2025deploying10k} or reinforcement learning to train a shared policy~\cite{learntofollow2024,guillaume2018primal}.

However, for certain real-world situations, having a planning time of less than one minute using search-based or priority-based algorithms is not sufficient, and learning-based methods struggle to generalize to unseen scenarios. For example, the Amazon-sponsored League of Robot Runners competition requires planning 10,000 agents in one second \cite{chan2024league}. 
In such scenarios, most state-of-the-art MAPF methods leverage rule-based methods~\cite{WangB11}, which move agents by following a set of pre-defined rules.
The state-of-the-art rule-based method is Priority Inheritance with Backtracking (PIBT)~\cite{okumura2022priority} due to its extreme speed. 

\subsection{PIBT and Its Variants}
PIBT is an extremely fast one-step priority-based planning algorithm capable of returning collision-free solutions for hundreds of agents in less than 200 \textit{milliseconds}~\cite{Jiang2024Competition}, making it the key component within several modern state-of-the-art MAPF methods~\cite{Jiang2024Competition,jiang2025deploying10k,okumura2023lacam,yukhnevich2025epibt}.
PIBT is a one-step rule-based planning algorithm. In each timestep, agents are assigned a priority. PIBT sequentially plans each agent, starting from the highest priority agent. When planning, an agent is always required to avoid higher-priority agents' paths. If a high-priority agent $a_i$ plans a one-step path that bumps into (i.e., collides with) a lower-priority agent $a_j$'s current location, PIBT lets $a_j$ temporarily \emph{inherit} the priority of $a_i$ and attempt to plan an action. If $a_j$ is able to find a collision-free path avoiding planned paths of higher-priority agents, then $a_i$ can use its one-step path. If not, $a_j$ is required to wait at its current location, and $a_i$ is required to try a different action.

Since PIBT plans for only one step, agents might perform myopically. For example, when two agents enter the same corridor in opposite directions, the lower priority agent only knows to move backward once it meets the high-priority agent. winPIBT~\cite{okumura2019winpibt} and EPIBT~\cite{yukhnevich2025epibt} attempt to extend PIBT to plan for windowed $w$-step ($w \geq 1$) paths, but winPIBT performs worse than PIBT. 
EPIBT outperforms PIBT by enumerating all $w$-step paths and running PIBT where agents choose between these paths instead of single-step actions.

Both PIBT and EPIBT assume that the agents are identical in size. Heterogeneous PIBT (HetPIBT)~\cite{Chakravarty2026HetPIBT} extends PIBT to work with heterogeneous agents that vary in both size and speed. It assumes that an agent's velocity is proportional to its size and plans agents sequentially as PIBT does. It starts by finding a spatial path for each agent such that the lower priority agents will vacate the space required by the higher priority agents' paths. It then backtracks to fill up the temporal information of the spatial paths.



A core problem with EPIBT and PIBT is that the $w$-step paths with $w > 1$ could result in a high-priority agent bumping into multiple low-priority agents. Even for HetPIBT, while it increases the length of each single-step path based on agents' velocities, it cannot consider multi-step paths that collide with multiple agents.
This poses a series of difficult questions. For example, how does priority inheritance work if we allow each agent to collide with multiple agents? Which low-priority agent gets planned next if there are multiple? What happens if one agent succeeds while the other fails?
PIBT does not need to reason about this instance as it only deals with one-step planning where an agent can only bump into at most one other agent. EPIBT also explicitly avoids these questions by only considering paths that collide with at most one other agent. 
Our objective is to remove the restriction of only considering paths with $\leq$ 1 collision by introducing Multi-Dependency PIBT.

\section{Agent Dependency Perspective of PIBT}
Our main theoretical insight is reinterpreting PIBT as a method that reasons over agent dependencies. 


We define each agent $a_i$ to have a \emph{tentative path} $a_i.\pi$ that we are trying to compute. $a_i$ also has a \emph{safe path} $a_i.\tau$, which has the property that all agents' safe paths $a_{1:N}.\tau$ are collision-free. Note that, at every timestep $t$, all agents have a trivial safe path of waiting at their current locations, i.e., $a_i.\pi^{t+1} \gets a_i.\pi^{t}, \forall a_i \in A$.



We now describe PIBT as reasoning over agent dependencies between tentative paths and safe paths. PIBT initializes the safe paths with wait action and plans tentative paths for each agent one by one. Each PIBT call attempts to find tentative paths that are collision-free with each other and with the safe paths of the agents that do not have tentative paths yet.
This requires PIBT to reason agent dependencies when a high-priority agent $a_i$'s tentative path collides with a low-priority agent $a_j$'s safe path. 
When this occurs, PIBT requires $a_j$ to find a tentative path that does not collide with any high-priority agents. If $a_j$ cannot, then it chooses its safe path, which requires $a_i$ to try an alternative path.
We now formally define dependencies.


\begin{definition}[Agent Dependencies]
    \label{def:agent-dependency}
    $a_i$ has an agent dependency on $a_j$ if $a_i$'s tentative path collides with $a_j$'s safe path, i.e., $a_i.\pi$ collides with $a_j.\tau$.
\end{definition}

Conceptually, an agent dependency from $a_i$ to $a_j$ denotes that, for $a_i$ to use $a_i.\pi$, $a_j$ must find a collision-free path $a_j.\pi$ that is not the safe path.
If $a_j$ cannot find such a path, it has to use its safe path, and $a_i$ needs to find a new path.

We now define two types of dependencies, namely hard and soft dependencies, depending on if $a_j$ is planned or not.

\begin{definition}[\emph{Hard} Dependencies]
    \label{def:hard-agent-dependency}
    $a_i$ has a hard dependency on $a_j$ if $a_i$'s path collides with $a_j$'s safe path and \textit{$a_j$ is not planned}, i.e., $a_i.\pi$ collides with $a_j.\tau$ and $a_j.\pi = \perp$.
\end{definition}

\begin{definition}[\emph{Soft} Dependencies]
    \label{def:soft-agent-dependency}
    $a_i$ has a soft dependency on $a_j$ if $a_i$'s path collides with $a_j$'s safe path and \textit{$a_j$ is planned}, i.e., $a_i.\pi$ collides with $a_j.\tau$ and $a_j.\pi \neq \perp$.
\end{definition}






Using this definition of agent dependencies, we can reinterpret PIBT as searching over an Agent Dependency Graph. 
\begin{definition}[Agent Dependency Graph (AgDG)]
    \label{def:agdg}
    An agent dependency graph is a directed graph where each agent is a node, and each edge denotes either a hard dependency ($a_i \rightarrow a_j$) or a soft dependency ($a_i \dashrightarrow a_j$).
\end{definition}

PIBT only plans with hard dependencies (we will revisit soft dependencies later). 
In PIBT, each $a_i.\pi$ is a one-step path, so each agent can depend on at most one other agent.
Thus, PIBT's AgDG is a linked list.
PIBT first plans the highest priority agent and then plans through all dependent agents until all of them have collision-free paths. Then, it proceeds to the next-highest-priority unplanned agent and repeats this process until all agents are planned.

\begin{algorithm}[!t]
\caption{Agent class}\label{alg:agent-class-main}
\textbf{Agent $a_k$:} \\
\Indp
$\pi \gets \perp$ \tcp{Tentative path}
$\tau \gets \tau_k$  \tcp{Safe path}
$p \gets \phi(a_k)$ \tcp{Priority}
$\mathcal{P} \gets \rho(s_k,g_k,\parameter{w})$ \tcp{All possible $w$-step paths of agent $a_k$, sorted by distance to goal}
$d \gets 0$ \tcp{Path index}
$r \gets 0$ \tcp{Number of planning attempts}
\Indm
\end{algorithm}

\begin{algorithm}[!t]
\caption{MD-PIBT: Outer Loop}\label{alg:main-loop-main}\LinesNumbered
\SetKwInput{KwInput}{Input}
\SetKwInput{KwOutput}{Output}
\SetKwProg{Fn}{Function}{:}{}
\SetKwFunction{MDPIBT}{MDPIBT}
\SetKwFunction{Sorted}{sorted}
\SetKwFunction{updateSafePaths}{updateSafePaths}
\SetKwFunction{findNextBestPath}{findBestPath}
\SetKwFunction{getParent}{getHardParent}
\SetKwFunction{removeDependency}{removeDepend}
\SetKwFunction{addDependency}{addDepend}
\SetKwFunction{ShouldAgentChooseSafePath}{fallToSafePath}
\SetKwFunction{chooseParent}{chooseParent}
\SetKwFunction{ShouldPlan}{shouldPlan}

\KwInput{
    Agents $A = \{a_1,...,a_n\}$ with initialized values according to the agent class;
}




    Initialize $\mathcal{G}(\mathcal{V}=A,\mathcal{E}=\emptyset)$,
    $A_{plan} \gets \emptyset$\\
    \For{$a_i \in \Sorted{A}$}{
        \If{$a_i.r = 0$} {
            \MDPIBT{$a_i$, $A$, $\mathcal{G}$, $A_{plan}$} \\
            \For{$a_k\in A: a_k.\pi \neq \perp$}{
                $a_k.{\tau} \gets a_k.{\pi}$, $a_k.{\pi} \gets \perp$\\
            }
        }
    }
\end{algorithm}

\begin{algorithm}[!t]
\caption{MD-PIBT Function}\label{alg:md-pibt-main}\LinesNumbered
\SetKwProg{Fn}{Function}{:}{}

\Fn{\MDPIBT{$a_i$, $A$, $\mathcal{G}$, $A_{plan}$}}{
    $q \gets$ \parameter{Queue}(\{$a_i$\}) \\
    \While{$ q \neq \emptyset$}{
        $a_k \gets q.pop()$ \\
        $a_k.r \gets a_k.r + 1$ \\
        $\pi' \gets$\parameter{\findNextBestPath}($a_k, A_{plan} \parameter{~|~m, C}$) \\
        \uIf(\tcp*[f]{\Cref{sec:mdpibt-add-dep}}){$\pi' \neq \perp$}{
            $a_k.{\pi} \gets \pi'$,
            $A_{plan} \gets A_{plan} \cup \{a_k\}$\\
            \addDependency{$\mathcal{G}$, $a_k$, $\pi'$, $q$, $A$}\\
        }
        \Else(\tcp*[f]{\Cref{sec:md-pibt-remove-dep}}){
            \uIf{\parameter{\ShouldAgentChooseSafePath}($a_k\parameter{~|~R}$)}{
                    \For{$a_p:(a_p \to a_k) \in \mathcal{E}$}{
                        \removeDependency{$\mathcal{G}$, $a_p$, $q$, $A_{plan}$}\\
                    }
                $a_k.\pi \gets a_k.\tau$,
                $A_{plan} \gets A_{plan} \cup \{a_k\}$\\
                
            }
            \Else{
                $a_p \gets \parameter{\chooseParent}(\delta^{\text{HARD}-}(a_k))$\\
                \removeDependency{$\mathcal{G}$, $a_p$, $q$, $A_{plan}$}\\
                $q.push(a_p)$\\
                
            }
        }
    }
}
\end{algorithm}

\section{Multi-Dependence PIBT}

We describe MD-PIBT in \Cref{alg:agent-class-main,alg:main-loop-main,alg:md-pibt-main,alg:agdg-helper-main}, where red refers to hyper-parameters. We first give an overview of MD-PIBT in \Cref{sec:mdpibt-overview} and then discuss key components in \Cref{sec:mdpibt-add-dep,sec:md-pibt-remove-dep}, with a summary of hyper-parameters in \Cref{sec:md-pibt-param}.
With certain choices of hyper-parameters, MD-PIBT can be reduced to standard PIBT or EPIBT. 
We give a running example of MD-PIBT in \Cref{fig:main-md-pibt}.

\subsection{Overview} \label{sec:mdpibt-overview}

Unlike PIBT, which generates collision-free one-step paths, MD-PIBT plans collision-free $w$-step paths ($w \geq 1$).
\Cref{alg:agent-class-main} shows the attributes for each agent. Each agent $a_k$ is initialized with the set of all $w$-step paths $\mathcal{P}$ starting at the start location $s_k$, sorted by the distance from the end of the path to the goal location $g_k$.
The agent try each path by iterating through $\mathcal{P}$ via a path index $d$. $r$ is the number of planning attempts (more on it later).

\Cref{alg:main-loop-main} shows the outer loop. 
MD-PIBT initializes the agents with an AgDG $\mathcal{G}$ and a set of agents $A_{plan}$ to remember all agents that have a tentative path. It then calls the \texttt{MDPIBT} function on each agent $a_i$ in the priority order given by $p$ (computed by a priority function $\phi$) in \Cref{alg:agent-class-main} if $a_i$ has never been planned (line 3).
After planning for $a_i$ and its potential descendants (line 4), we update the safe paths of all agents by replacing their current safe paths with the newly planned tentative paths, if there is one, and reset all tentative paths to $\perp$ (line 5-6).

\Cref{alg:md-pibt-main} presents the core logic of MD-PIBT. We maintain a priority queue $q$, initialized with the starting agent $a_i$ (line 2), which stores all agents whose dependencies require resolution and thus must be planned.
Because $\mathcal{G}$ is a graph rather than a linked list, multiple agents may simultaneously become eligible for planning, and the order in which they should be processed is not uniquely determined. To accommodate this flexibility, we employ a priority queue that allows agents to be selected according to a configurable ordering policy. At each step, MD-PIBT selects the next agent from $q$, plans for that agent, and introduces additional dependencies when necessary.

To manipulate $\mathcal{G}$, we define $\delta^{\text{HARD}+}$, $\delta^{\text{HARD}-}$, $\delta^{\text{SOFT}+}$, and $\delta^{\text{SOFT}-}$ to retrieve the hard and soft outgoing edges ($+$) and incoming edges ($-$) of an agent in $\mathcal{G}$. In addition, we define $\delta^+$ and $\delta^-$ to denote the sets of all the incoming and outgoing edges of an agent, respectively.

In each iteration, MD-PIBT pops the next agent $a_k$ from $q$ and increments its planning attempt counter $r$ (lines 3–5). MD-PIBT then leverages the \texttt{findBestPath} function (line 6), which uses a path index $d$ and a sorted list of all $w$-step paths $\mathcal{P}$ in \Cref{alg:agent-class-main} to find its next-best \emph{valid} path. We define the valid path in \Cref{sec:md-pibt-param}.
If a valid path $\pi'$ exists (line 7), it adds dependency edges based on whose safe paths $\pi'$ collides with (line 8), updates $a_k$'s tentative path, and adds $a_k$ to $A_{plan}$ (line 9). 
More details are explained in \Cref{sec:mdpibt-add-dep}.
If $a_k$ cannot find a valid path $\pi'$, $a_k$ needs to decide between taking its safe path $a_k.\tau$ (lines 12-14) or allowing itself to replan by first requesting its parent to replan (lines 16-18). 
More details are explained in \Cref{sec:md-pibt-remove-dep}.
This process is repeated until the queue becomes empty.

\begin{figure*}[t!]
    \centering
    \includegraphics[width=0.9\linewidth]{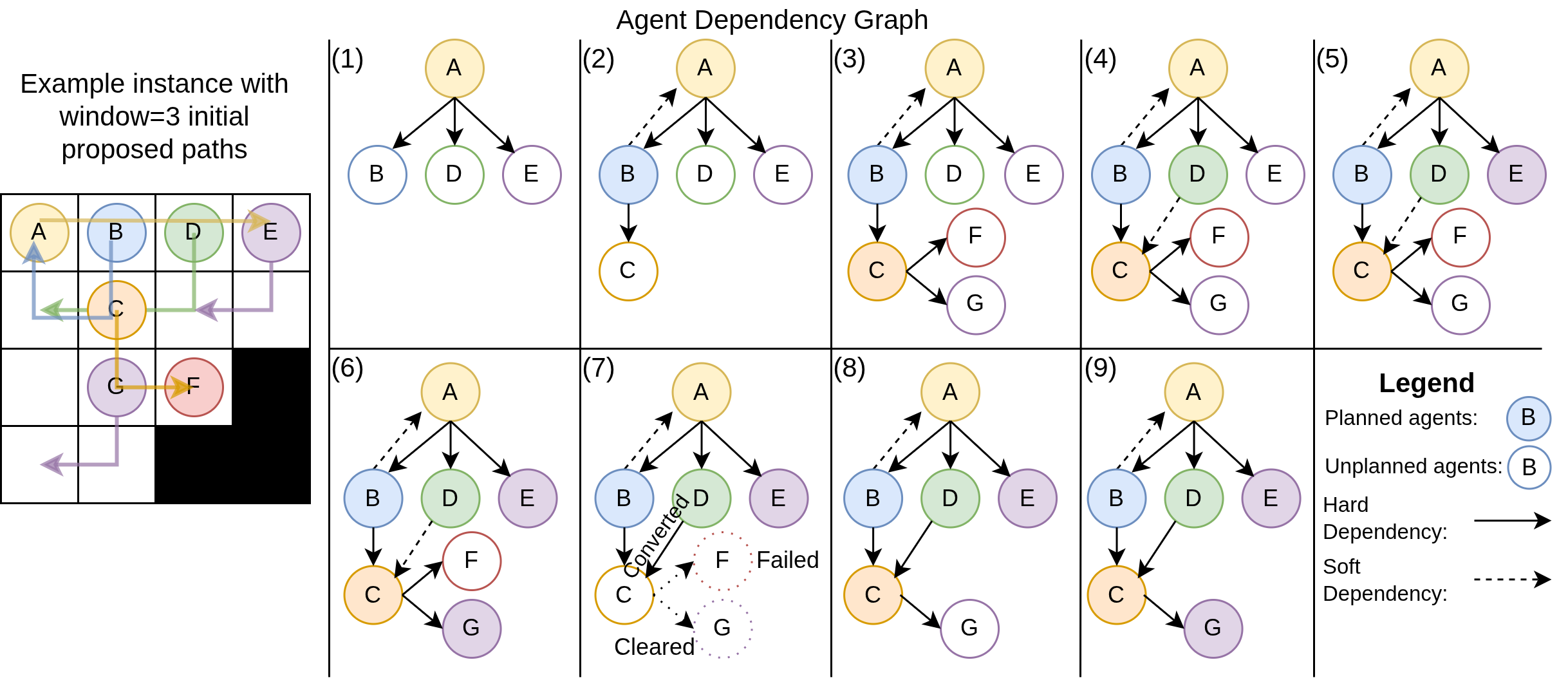}
    \vspace{-0.5em}
    \caption{Multi-Dependence PIBT (MD-PIBT) builds and searches over an \emph{Agent Dependency Graph}. Left shows a scenario for planning with a window size of 3, with initial path preferences drawn. Assume all agent's safe paths are waiting at their current locations. (1) Let MD-PIBT start planning with A. A's path conflicts with B, D, and E's safe path, causing A to have hard dependencies on them (\Cref{def:hard-agent-dependency}, they must find non-safe paths for A's path to be valid). Thus, $B, D, E$ need to be planned next. Given multiple agents, we plan in alphabetical order. (2) When B plans, B's path collides with C's and A's safe paths. Since A is already planned, we record a soft dependency between B and A. (3-6) This logic continues until planning F. (7) F fails to find a collision-free path. When this occurs, F requires a parent (in this case C) to \emph{replan}. The replan request unplans C which includes removing downstream dependencies and converting soft dependencies to C to hard dependencies. (8) Suppose that C replans by moving down, which does not intersect with F's safe path. Then F is not included in the AgDG. (9) After planning all agents in the AgDG, we can move on to plan other agents not in the AgDG (not depicted).}
    \label{fig:main-md-pibt}
    \vspace{-0.5em}
\end{figure*}


\begin{algorithm}[!t]
\caption{Agent Dependency Graph Functions}\label{alg:agdg-helper-main}\LinesNumbered
\SetKwProg{Fn}{Function}{:}{}
\SetKwFunction{hasConflict}{hasCollision}
\SetKwFunction{removeEdges}{removeEdges}
\SetKwFunction{convertEdges}{convertEdges}
\SetKwFunction{hasHardParents}{hasHardParents}

\Fn{\addDependency{$\mathcal{G}$, $a_k$, $\pi'$, $q$, $A$}}{
    \For{$a_{i \neq k} \in A :$ \hasConflict{$a_i.{\tau}, \pi'$}}{
        \uIf{$a_i.{\pi} = \perp$}{
            $\mathcal{E} \gets \mathcal{E} \cup \{(a_k \rightarrow a_i\})$\\
            $q.push(a_i)$\\
        }
        \Else{
            $\mathcal{E} \gets \mathcal{E} \cup \{(a_k \dashrightarrow a_i)\}$\\
        }
    }
}



\Fn{\removeDependency{$\mathcal{G}$, $a_p$, $q$, $A_{plan}$}}{

    $a_{p}.\pi \gets \perp$,
    $A_{plan} \gets A_{plan} \setminus \{a_{p}\}$\\
    \texttt{convertSoftEdgesToHard}($\delta^{\text{SOFT}-}$) \\
    \uIf{$|\delta^{\text{HARD}-}(a_{p})| > 0$ \textbf{and} $a_{p} \notin q$}{
        $q.push(a_{p})$
    }
    \ElseIf{$|\delta^{\text{HARD}-}(a_{p})| = 0$ \textbf{and} $a_{p} \in q$}{
        $q.erase(a_{p})$
    }
    $\mathcal{E} \gets \mathcal{E} \setminus \delta^{\text{SOFT}+}(a_p)$ \tcp{Remove out soft edges}
    \For{$a_c : (a_p \to a_c) \in \mathcal{E}$}{
    $\mathcal{E} \gets \mathcal{E} \setminus \{(a_p \to a_c)\}$ \\
    $a_{c}.d \gets 0$\\
    \removeDependency{$\mathcal{G}$, $a_c$, $q$, $A_{plan}$}\\
    }
}
\end{algorithm}

\subsection{Accepting the Tentative Path} \label{sec:mdpibt-add-dep}

When a valid path $\pi'$ exists for $a_k$ (line 7), we accept it as $a_k$'s tentative path and add $a_k$ to $A_{plan}$ (line 8). Since this tentative path may collide with the safe paths of some agents, we call the \texttt{addDepend} function to update the dependencies (line 9). As shown in \Cref{alg:agdg-helper-main}, the \texttt{addDepend} function identifies all agents $a_i$ whose safe paths collide with $\pi'$ and adds hard or soft dependencies accordingly.


\subsubsection{Add Hard Dependencies}
Lines 3-5 add \emph{hard} dependencies $a_k \rightarrow a_i$ to $\mathcal{G}$ for all \emph{unplanned} agents $a_i$ (i.e., $a_i$ does not have a tentative path $a_i.\pi$) whose safe paths collide with $\pi'$. 
For example, in \Cref{fig:main-md-pibt}(1), agent $a_A$ bumps into safe paths of unplanned agents $a_B, a_D,$ and $a_E$, so we add hard dependencies from $a_A$ to all of them.

\subsubsection{Add Soft Dependencies}
Unlike PIBT, MD-PIBT requires \emph{soft} dependencies.
Lines 6-7 add soft dependencies $a_k \dashrightarrow a_i$ for all \emph{planned} agents $a_i$ (i.e., $a_i$ has a tentative path $a_i.\pi$) whose safe paths collide with $\pi'$.
For example, in \Cref{fig:main-md-pibt}(4), $a_D$ is planning and $a_D.\pi$ collides with $a_C.\tau$ \textit{after} $a_C$ has a tentative path $a_C.\pi$. 
Later, if $a_C$'s hard dependency descendants fail, $a_C$ might need to replan.
Now, when $a_C$ replans, $a_D.\pi$ collides with $a_C.\tau$, resulting in a hard dependency $a_D \rightarrow a_C$, according to \Cref{def:hard-agent-dependency}.
To account for this possibility, we need soft dependency edges in \Cref{def:soft-agent-dependency} and convert them to hard ones when necessary.


\subsection{Backtracking} \label{sec:md-pibt-remove-dep}

When no valid path exists for $a_k$ (\Cref{alg:md-pibt-main} line 10), some already-planned agents must be replanned. This generalizes the backtracking process of PIBT. PIBT forces $a_k$ to take its safe path $a_k.\tau$ and asks the agent whose tentative path collides with $a_k.\tau$ to replan. We generalize this idea as the \emph{falling-to-safe path} option in \Cref{subsubsec:fall-safe}. 
Alternatively, EPIBT proposes a different option. If the agent whose tentative path collides with $a_k.\tau$ replans, $a_k$ may subsequently obtain a valid path that is better than its safe path. Therefore, instead of immediately committing to its safe path, EPIBT allows $a_k$ to be replanned. We generalize this idea as the \emph{attempting-to-replan} option in \Cref{subsubsec:attempt-replan}. To decide between these two options, we formulate the choice as a hyper-parameterized function \ShouldAgentChooseSafePath (line 11), which will be introduced in \Cref{sec:md-pibt-param}.

\subsubsection{Falling to Safe Paths} \label{subsubsec:fall-safe}
If $a_k$ is forced to its safe path (\Cref{alg:md-pibt-main} lines 12-14), then all $a_p$ that have hard dependencies on $a_k$ (i.e., $a_p \to a_k$) need to be replanned.
We use the \texttt{removeDepend} function to update each such $a_p$. 
As shown in \Cref{alg:agdg-helper-main}, \texttt{removeDepend} first empties the tentative path of $a_p$ and removes it from the planned agent set $A_{plan}$ (line 9). Since $a_p$ now no longer has a tentative path, by definition, all incoming soft edges become hard edges (line 10).  
Since any agent with hard dependencies (i.e., has incoming hard edges) must be replanned, we add $a_p$ to $q$ if it is not there (lines 11-12). On the other hand, since any agent with no hard dependencies does not need to be replanned, we remove $a_p$ from $q$ if it is there (lines 13-14). 
Then, all outgoing edges of $a_p$ should be removed because $a_p$ no longer has a tentative path (lines 15 and 17). Moreover, replanning $a_p$ may change the replanning necessity of its children. For example, for a child $a_{j}$ of $a_p$ with $a_p \to a_j$, if $a_p$'s new tentative path no longer collides with $a_j.\tau$, then $a_j$ may no longer need to be replanned.
Even if $a_p$'s new path collides with $a_j.\tau$, the new path would change the possible collision-free paths that $a_j$ could choose.
Therefore, MD-PIBT recursively unplans all child agents of $a_p$ (line 19). It also resets the path index $d$ (line 18) of each child agent $a_c$ since previously invalid paths for $a_c$ may become valid due to the potential change of $a_p.\pi$. Resetting the path index allows the agents to retry those paths.

\subsubsection{Attempting to Replan} \label{subsubsec:attempt-replan}
The second option for dealing with the case where no valid paths exist for $a_k$ is to ask one of $a_k$'s parent to replan and then retry $a_k$ (\Cref{alg:md-pibt-main} lines 16-18).
Thus, instead of forcing $a_k$ to pick the safe path, we choose some parent $a_p$ and ask it to replan by reusing the same \texttt{removeDepend} function. We then add $a_p$ to $q$. 
\Cref{fig:main-md-pibt}(6)-(8) shows an example of attempting to replan.

With replanning, agents can be planned more than once. In the worst case, each agent can be bumped into by a parent an exponential number of times, resulting in intractable runtime. Thus, the \texttt{fallToSafePath} function is necessary for reasonable runtimes. 

\begin{table*}[!t]
    \centering
    \small
    \resizebox{\linewidth}{!}{
    \begin{tabular}{cccccccccc}
    \toprule
    MAPF & Map & $|V|$ & Agent Model & $w$ & $h$ & $R$ & $C$ & $m$ & $p$ \\
    \midrule
    \multirow{6}{*}{One-shot} & \randomSmall       & 819 & PM & \multirow{2}{*}{$\{1,2,3,4\}$} & \multirow{2}{*}{$\{1,w\}$} & \multirow{2}{*}{100} & \multirow{2}{*}{$\{1,\infty\}$} & \multirow{2}{*}{$\{\text{PIBT},\text{EPIBT}\}$} & \multirow{2}{*}{LET}              \\
                              & \warehouseXlarge   & 5,699 & PM & & & &                     \\    
    \cmidrule(lrr){2-10}
                              & \randomLargeLA         & 15,580 / 10,862  (S/L) & PMLA  & \multirow{4}{*}{$\{3,4\}$} & \multirow{4}{*}{1} & \multirow{4}{*}{100} & \multirow{4}{*}{$\{1,2,4,8,\infty\}$} & \multirow{4}{*}{EPIBT} & \multirow{4}{*}{\{LET,SD\}}                   \\
                              & \emptyLargeLA         & 4,096 / 3,600  (S/L) & PMLA  & & & &                   \\
                              & \mazeLargeLA         & 14,818 / 8,574 (S/L) & PMLA  & & & &                   \\
                              & \roomLargeLA         & 14,716 / 11,843  (S/L) & PMLA    & & & &                 \\
    \midrule
    \multirow{2}{*}{Lifelong} & \randomSmall       & 819 & PM, RM & \multirow{2}{4cm}{PM:$\{1,2,3\}$\ RM:$\{3,4,5\}$\newline DDR:$\{3,6\}$}  & \multirow{2}{*}{$\{1,w\}$} & \multirow{2}{*}{1,000} & \multirow{2}{*}{$\{1,\infty\}$} & \multirow{2}{*}{$\{\text{PIBT},\text{EPIBT}\}$} & \multirow{2}{2cm}{PM/RM:\{LET,SD\}\ DDR: SD}   \\
                              & \warehouseXlarge   & 5,699 & PM, RM  & & & &                         \\
    \bottomrule
    \end{tabular}
    }
    \caption{
    Summary of MAPF model, maps, and agent models used in the experiments. $|V|$ is the number of vertices in the maps. 
    We report the number of vertices for small and large agents with the format of ``$x$ / $y$ (S/L)'', where $x$ and $y$ are the numbers of vertices for small and large agents, respectively. For agent priority strategies ($p$), we consider Longer Elapsed Time (LTE) and Shorter Dist (SD).
    }
    \label{tab:exp-setup}
    \par\vspace{-\abovecaptionskip}
\end{table*}

\subsection{MD-PIBT Hyper-parameters} \label{sec:md-pibt-param}

MD-PIBT has many hyper-parameters (highlighted in red).

\noindent \textbf{Queue:} The queue may have multiple unplanned agents. Thus, it is a hyper-parameter to decide which agent to plan next. We choose to use a stack, yielding DFS ordering.

\noindent \textbf{findBestPath($a_k, A_{plan} \parameter{~|~m, C}$):} 
An agent has multiple options when trying to pick the next path.
Such paths are defined differently for PIBT and EPIBT. For PIBT~\cite{okumura2022priority}, a tentative path is valid if it does not collide with the tentative paths of the other planned agents (i.e., agents in $A_{plan}$). For EPIBT, a tentative path is valid if it does not collide with the tentative paths of higher priority agents that obtained a new tentative path in the current call.

In our experiments, to capture this parameterization, we vary a hyper-parameter \emph{Find Path Mode} $m \in \{\text{PIBT}, \text{EPIBT}\}$. We also have a parameter $C$ that controls the maximum number of agents each agent can bump into. For PIBT and EPIBT, $C = 1$, while for MD-PIBT, $C \in \mathbb{Z}^+$. A valid path shall only collide with at most $C$ other agents.
We highlight that these are just hyper-parameters we choose and that other versions of \texttt{findBestPath} are possible. Appendix~\ref{appen:add-algo} shows our implementation of \texttt{findBestPath}.


\noindent \textbf{fallToSafePath($a_k \parameter{~|~R}$):} Agents can choose when to fall to their safe paths. We use a maximum replan limit $R$ and have each $a_k$ attempt at most $R$ times before choosing its safe path. Without $R$, it is possible to bump into an agent an exponential number of times.

\noindent \textbf{chooseParent($a_k$):} When $a_k$ cannot find a path, it needs to choose a parent to replan. Note that the picked parent's descendant agents will be cleared when it replans. Thus, we pick the most recently added parent as that should have fewer descendants than other parents.

We conduct theoretical analysis in Appendix~\ref{appen:theory}.

\section{Experimental Evaluation}

\begin{figure*}[!t]
    \centering
    \includegraphics[width=1\textwidth]{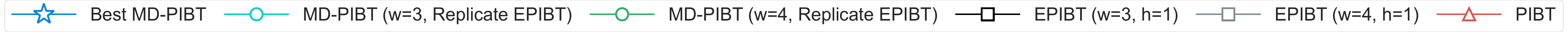}\par\medskip
    \vspace{-0.5em}
    \begin{subfigure}{0.49\textwidth}
        \centering
        \includegraphics[width=0.5\textwidth]{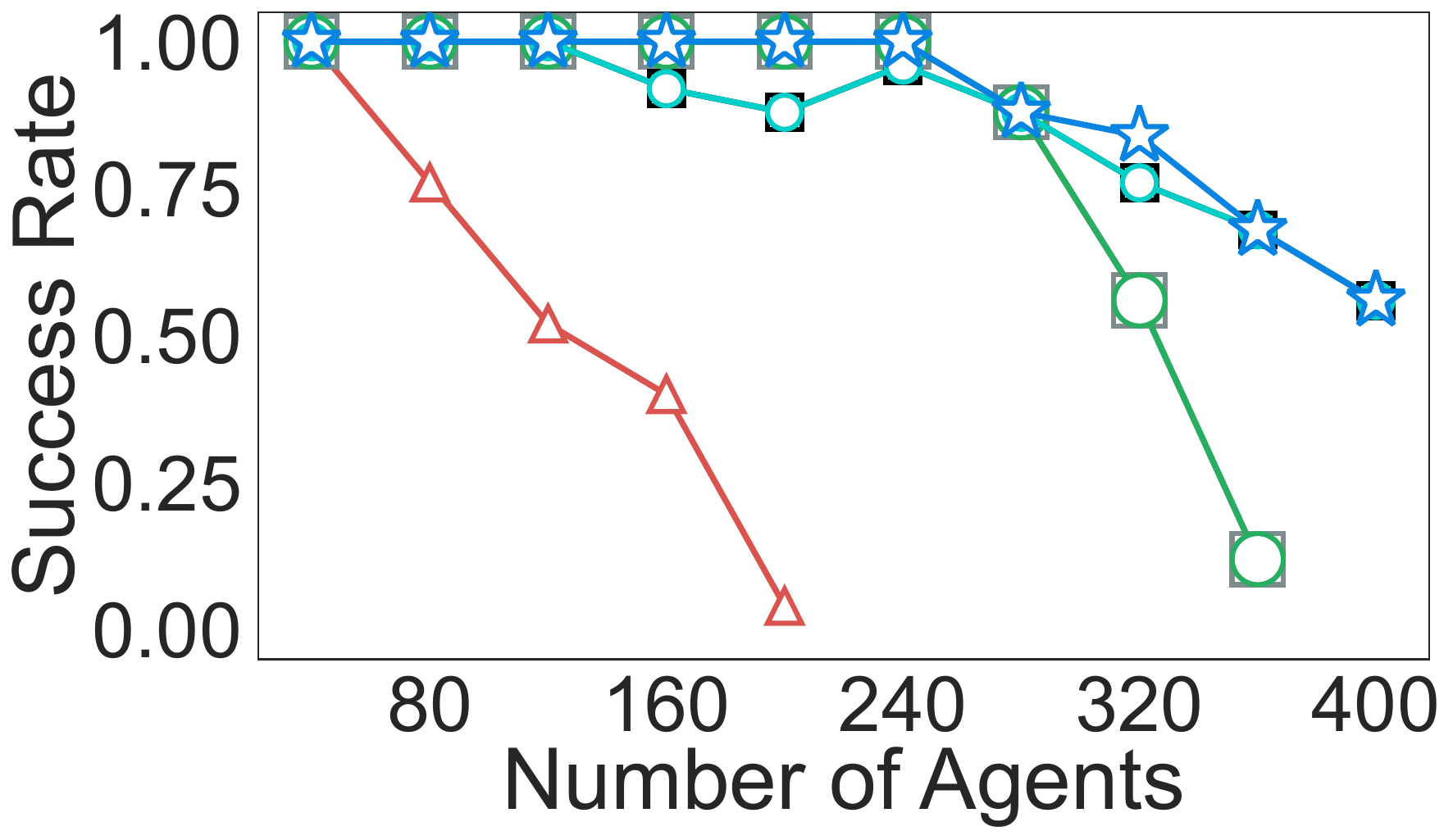}%
        \includegraphics[width=0.5\textwidth]{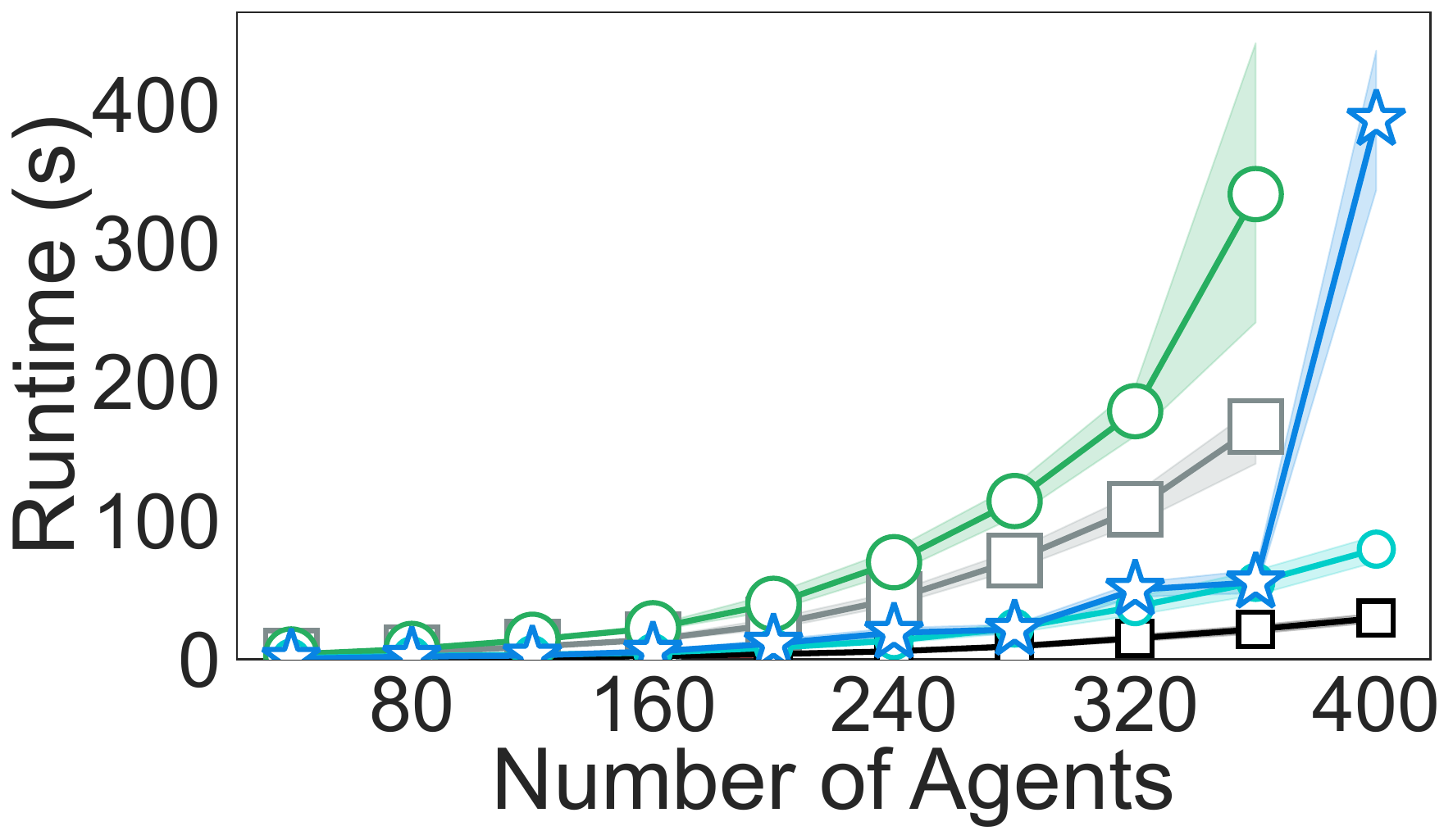}
        \vspace{-1.7em}
        \caption{\randomSmall}
        \label{fig:oneshotmapf-small-agents-random_32_32_20}
    \end{subfigure}%
    \hfill
    \begin{subfigure}{0.49\textwidth}
        \centering
        \includegraphics[width=0.5\textwidth]{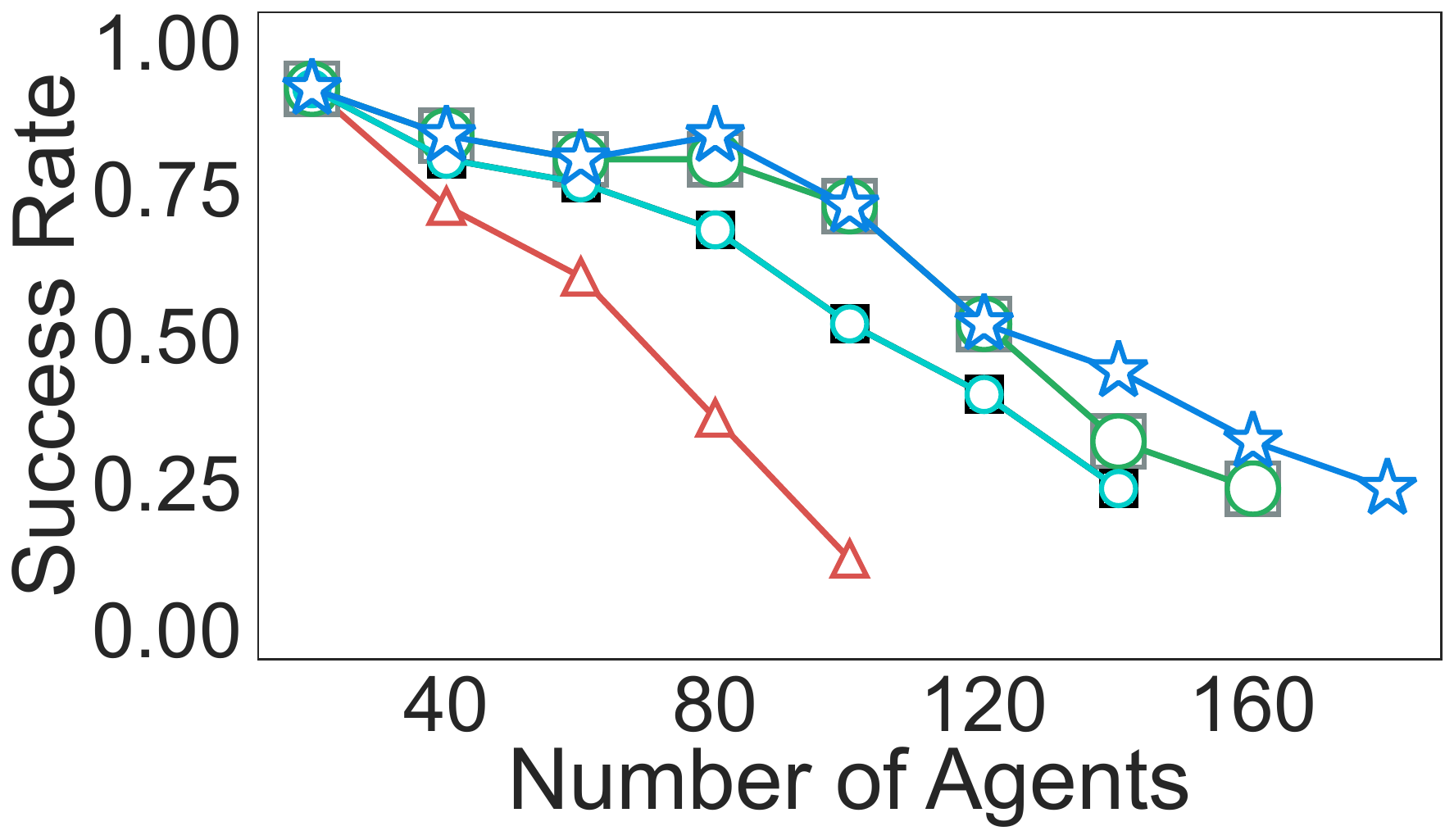}%
        \includegraphics[width=0.5\textwidth]{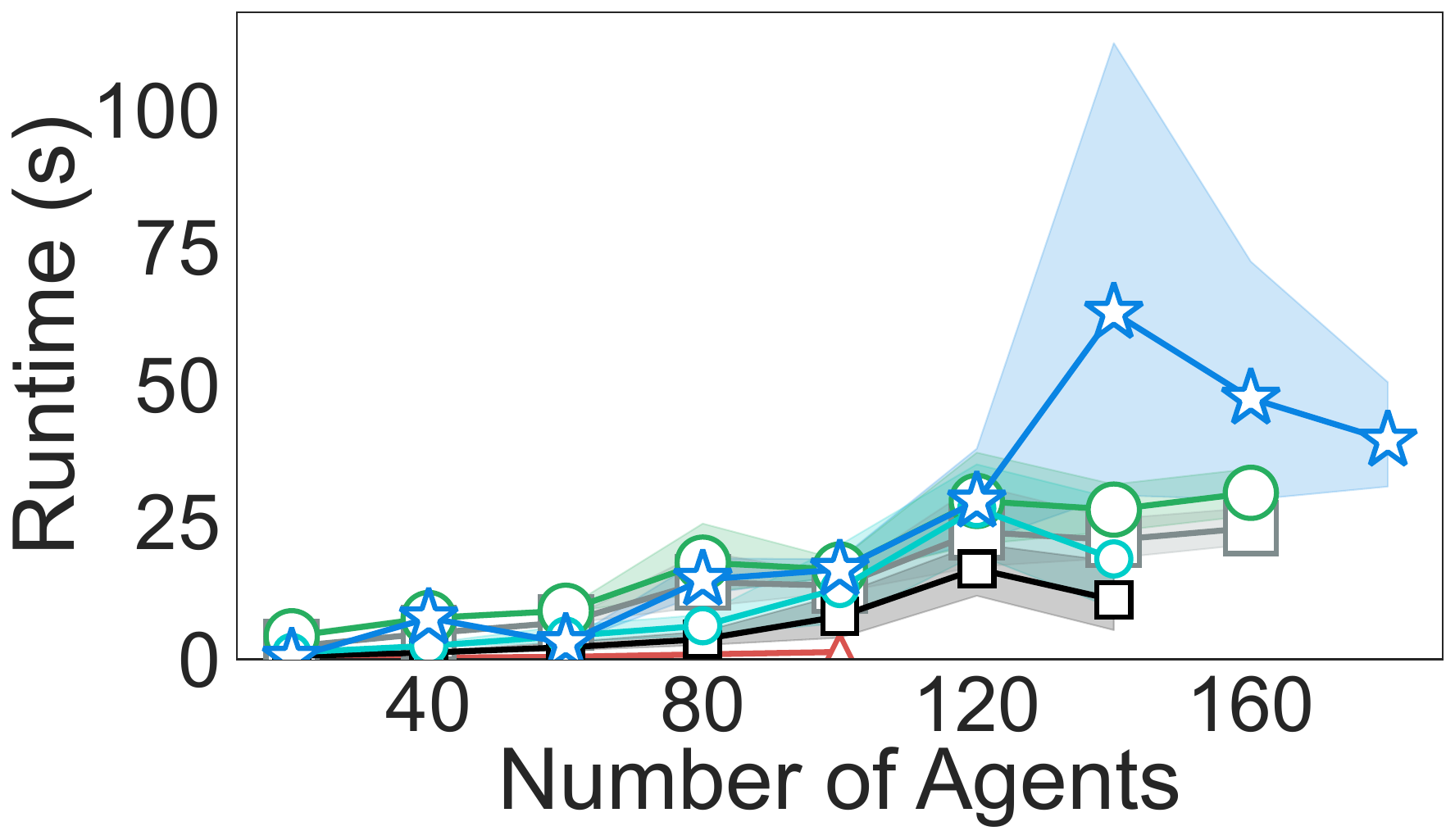}
        \vspace{-1.7em}
        \caption{\warehouseXlarge}
        \label{fig:oneshotmapf-small_agents-warehouse_10_20_10_2_1}
    \end{subfigure}
    \hfill

    \caption{Success rate and runtime for different numbers of agents for one-shot MAPF with PM agents.}
    \label{fig:major-result-oneshotmapf-smallagents}
    \par\vspace{-\abovecaptionskip}
\end{figure*}

\begin{figure*}[!t]
    \centering
    \includegraphics[width=0.8\linewidth]{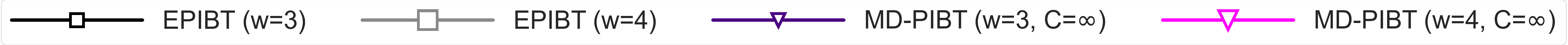}\par\medskip
    \vspace{-0.5em}
    \begin{subfigure}[t]{0.24\textwidth}
        \centering
        \includegraphics[width=\linewidth]{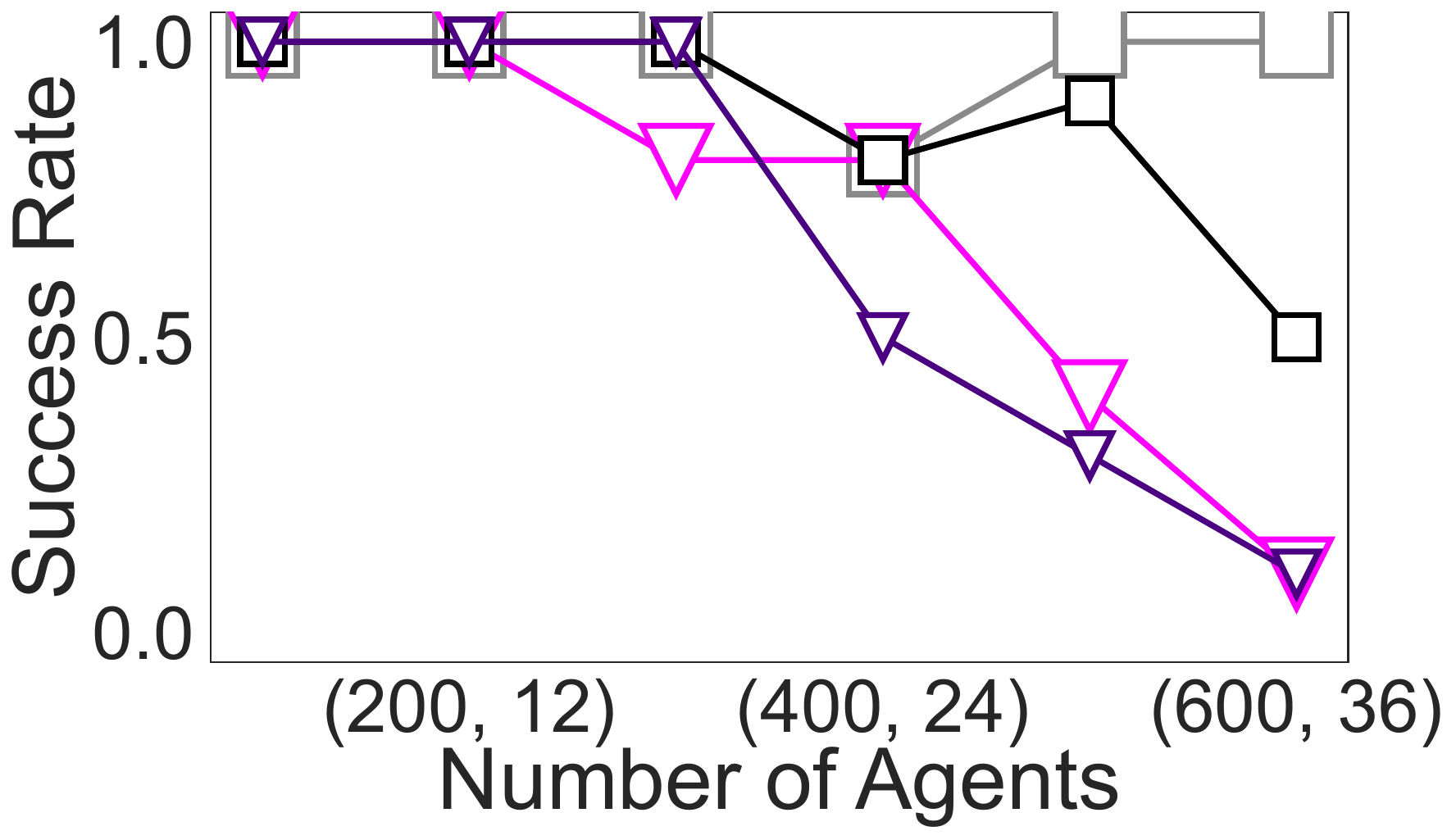}
    \end{subfigure}\hfill
    \begin{subfigure}[t]{0.24\textwidth}
        \centering
        \includegraphics[width=\linewidth]{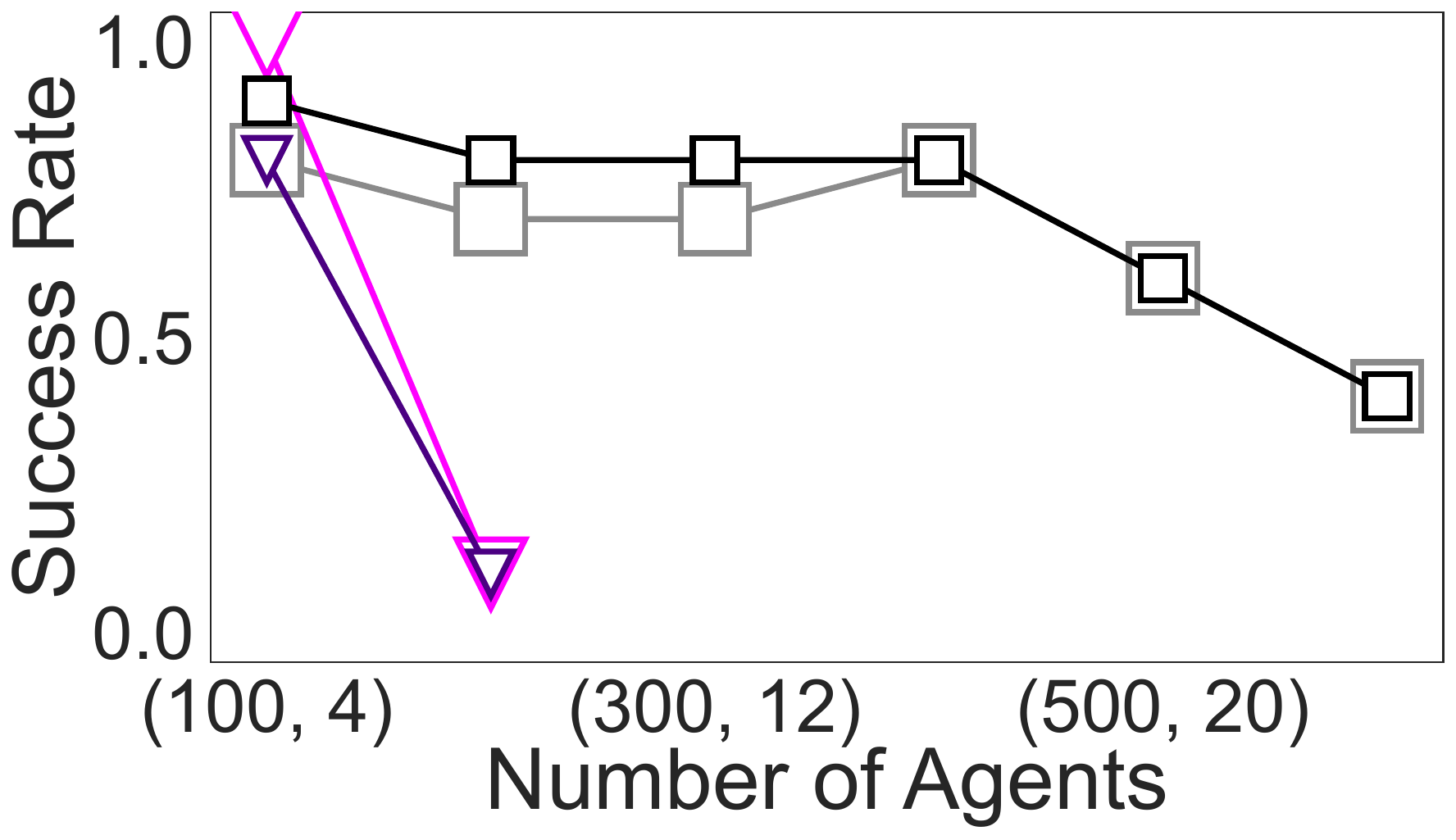}
    \end{subfigure}\hfill
    \begin{subfigure}[t]{0.24\textwidth}
        \centering
        \includegraphics[width=\linewidth]{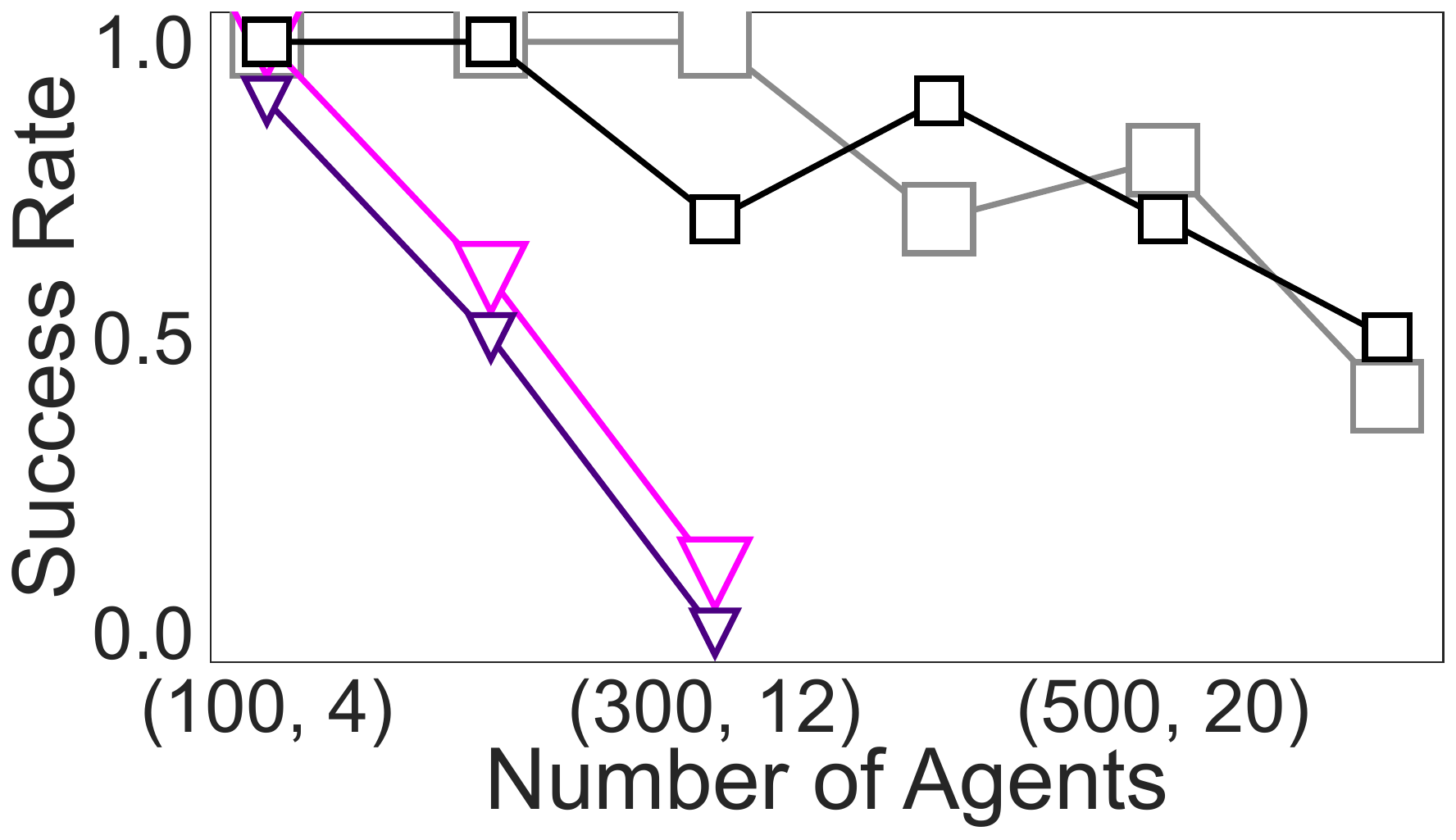}
    \end{subfigure}\hfill
    \begin{subfigure}[t]{0.24\textwidth}
        \centering
        \includegraphics[width=\linewidth]{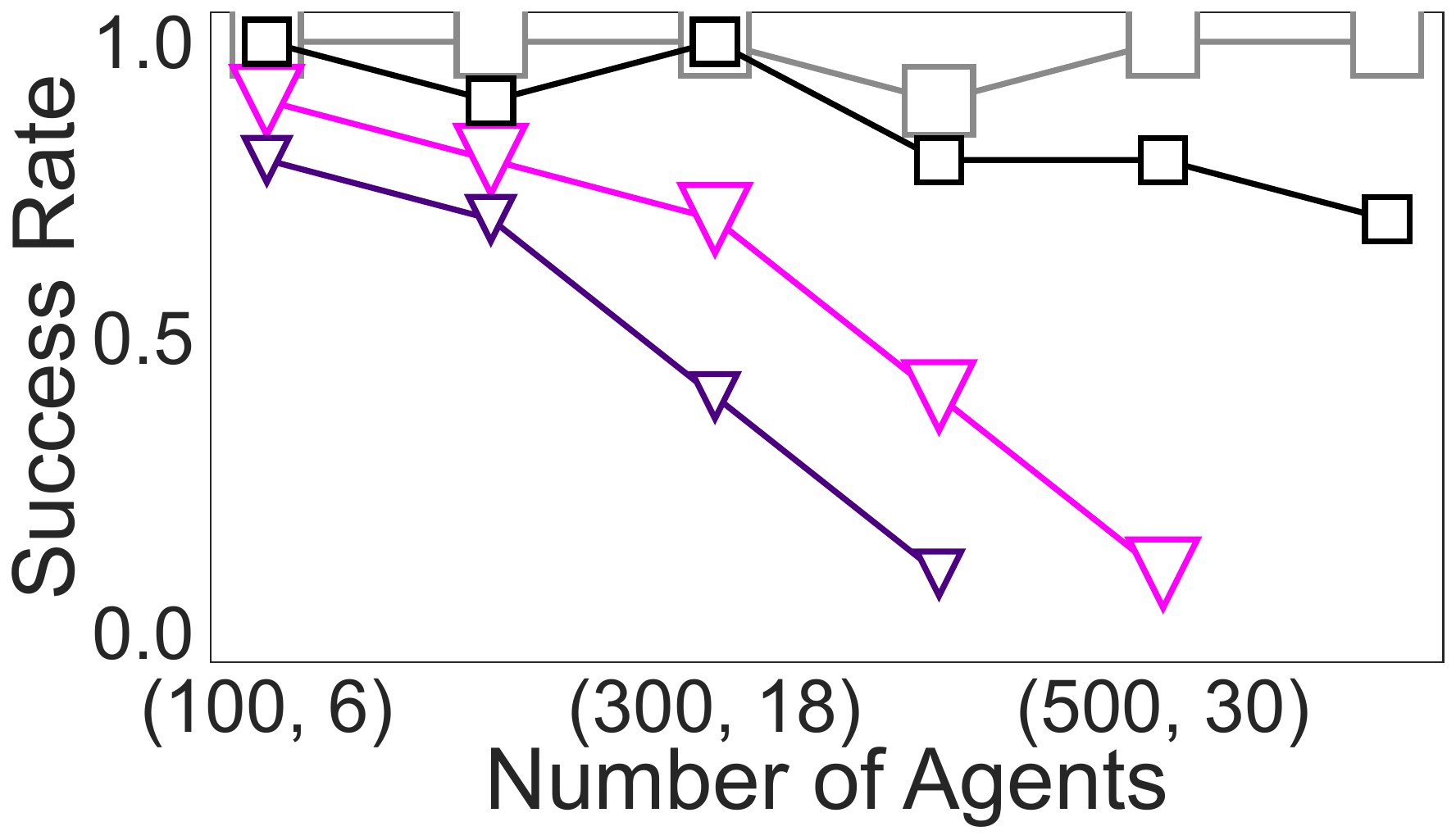}
    \end{subfigure}

    \begin{subfigure}[t]{0.24\textwidth}
        \centering
        \includegraphics[width=\linewidth]{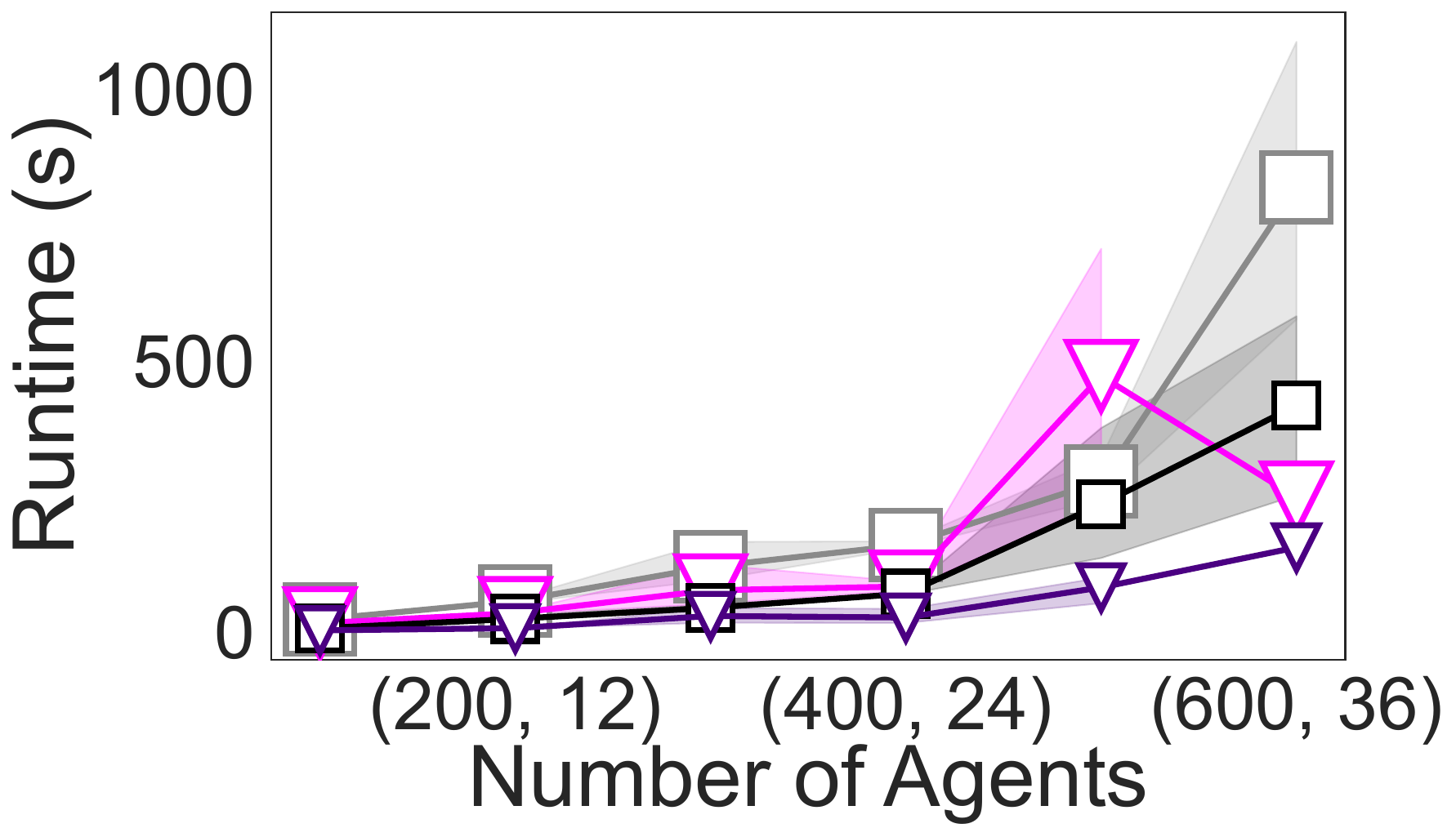}
    \end{subfigure}\hfill
    \begin{subfigure}[t]{0.24\textwidth}
        \centering
        \includegraphics[width=\linewidth]{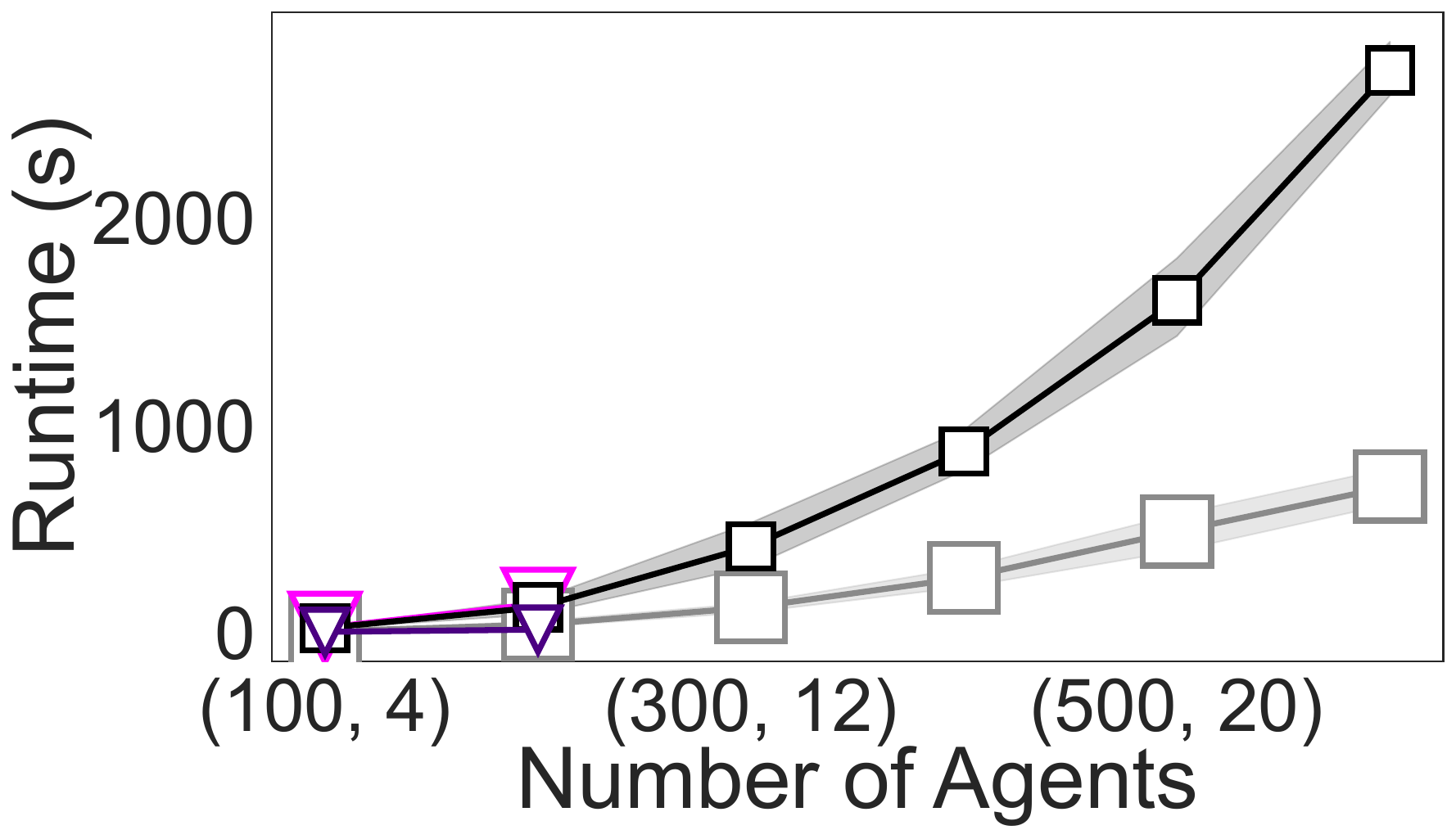}
    \end{subfigure}\hfill
    \begin{subfigure}[t]{0.24\textwidth}
        \centering
        \includegraphics[width=\linewidth]{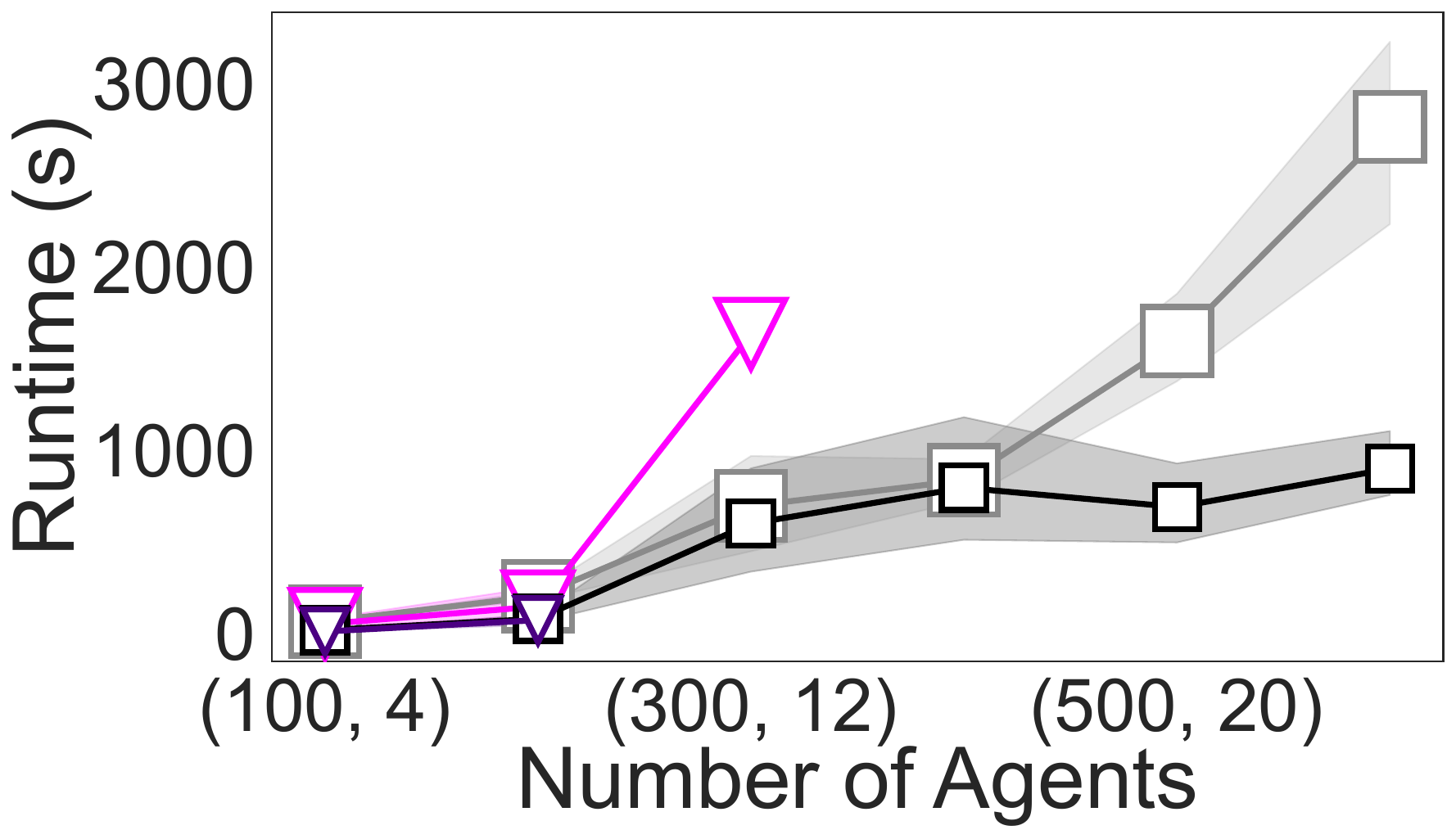}
    \end{subfigure}\hfill
    \begin{subfigure}[t]{0.24\textwidth}
        \centering
        \includegraphics[width=\linewidth]{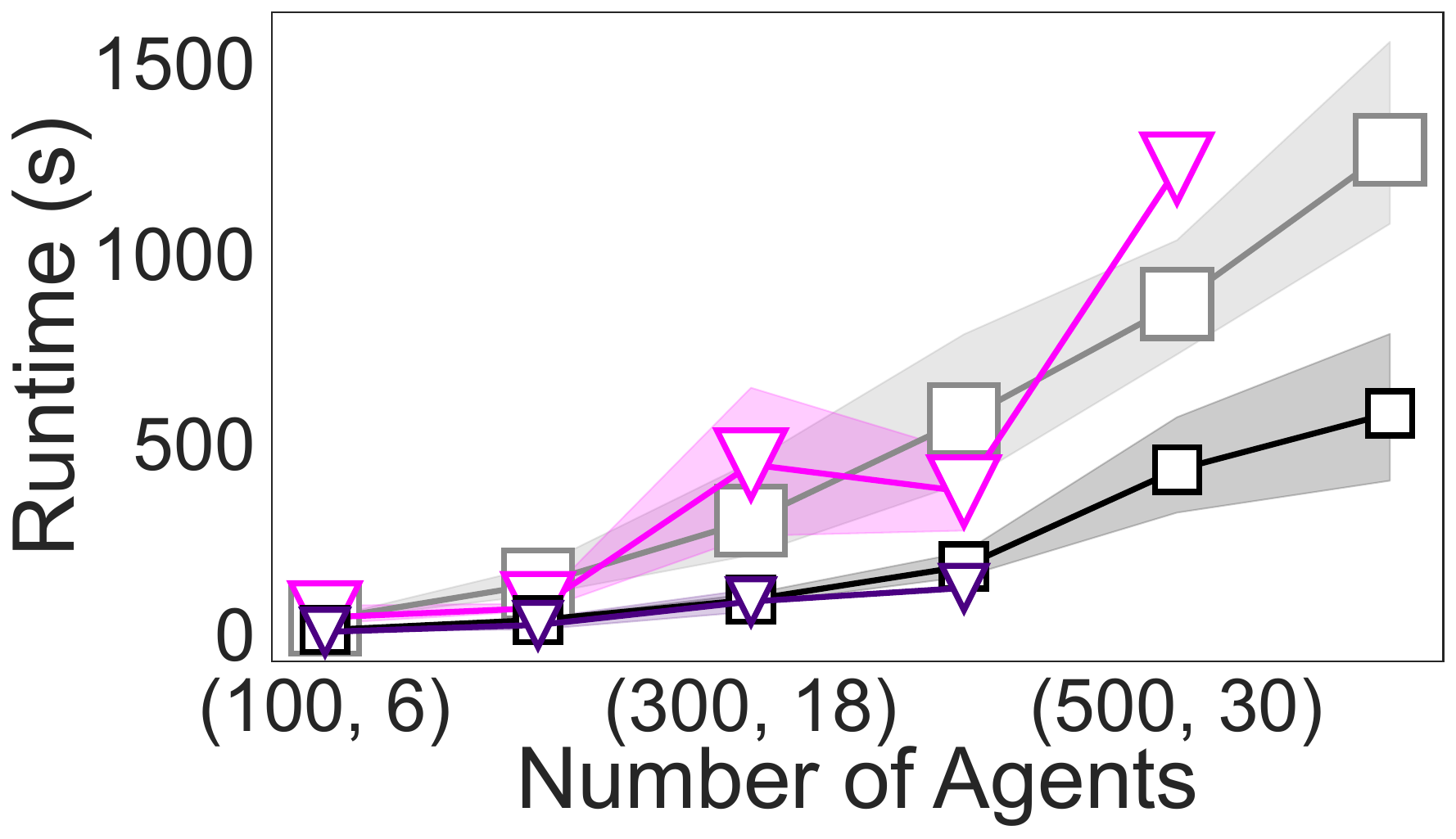}
    \end{subfigure}
    \par\vspace{-\abovecaptionskip}
    \captionsetup[subfigure]{skip=-0.6em}
    \subcaptionbox{\randomLargeLA\label{fig:oneshotmapf-large_agents-random_128_128}}
        {\phantom{\rule{0.24\textwidth}{0pt}}}\hfill
    \subcaptionbox{\emptyLargeLA\label{fig:oneshotmapf-large_agents-empty_64_64}}
        {\phantom{\rule{0.24\textwidth}{0pt}}}\hfill
    \subcaptionbox{\mazeLargeLA\label{fig:oneshotmapf-large_agents-maze_128_128}}
        {\phantom{\rule{0.24\textwidth}{0pt}}}\hfill
    \subcaptionbox{\roomLargeLA\label{fig:oneshotmapf-large_agents-room_128_128}}
        {\phantom{\rule{0.24\textwidth}{0pt}}}\hfill
    
    \caption{Success rate and runtime for different numbers of agents for one-shot MAPF with PMLA agents.}
    \label{fig:major-result-oneshotmapf-largeagents}
    \par\vspace{-\abovecaptionskip}
\end{figure*}

\begin{figure*}[!t]
    \centering
    \includegraphics[width=1\linewidth]{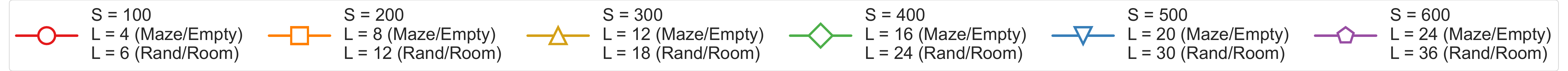}\par\medskip
    \vspace{-0.6em}
    \begin{subfigure}[t]{0.24\textwidth}
        \centering
        \includegraphics[width=\linewidth]{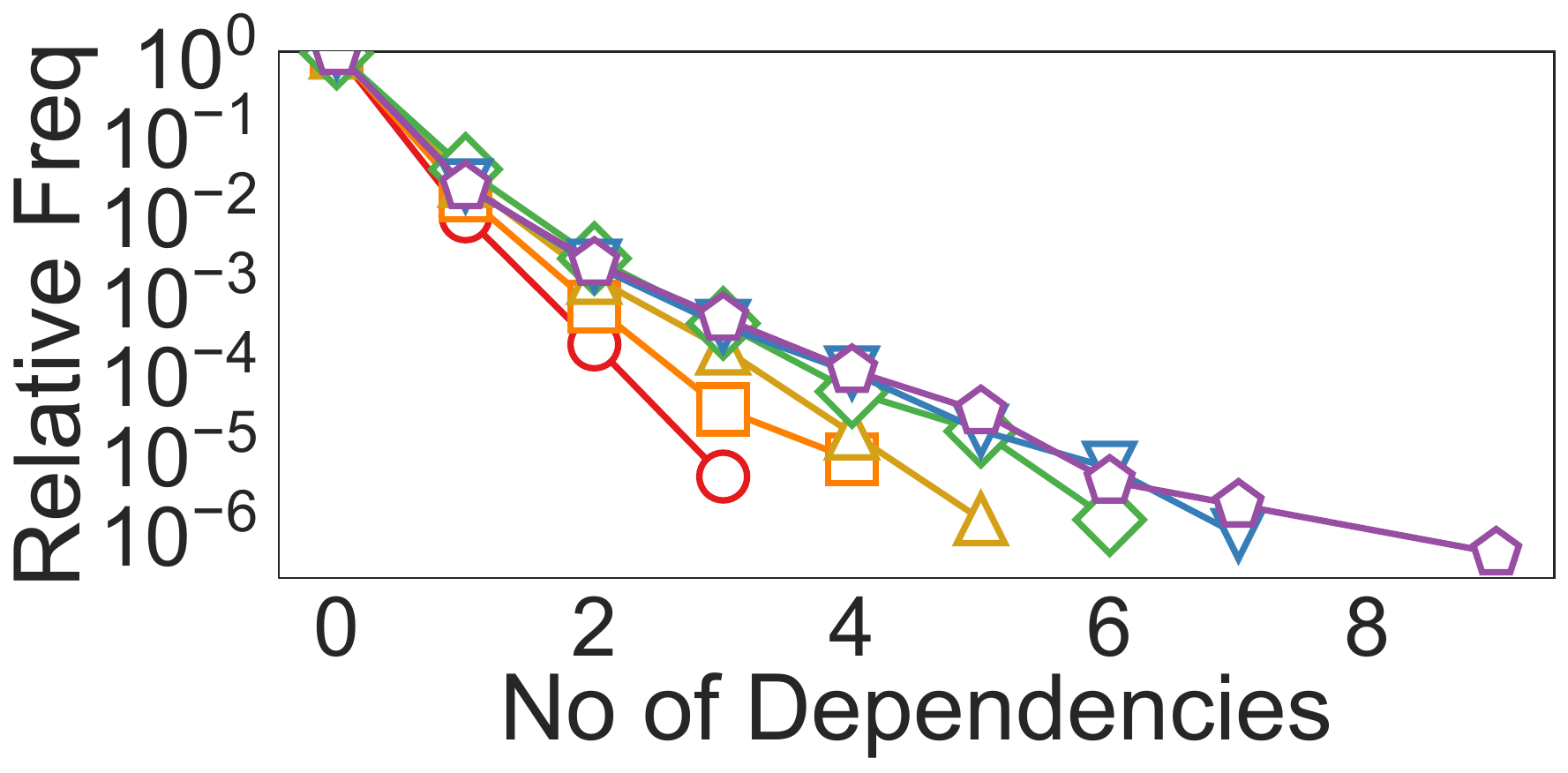}
        \vspace{-1.7em}
        \caption{\randomLargeLA\label{fig:oneshotmapf-large_agents-random_128_128-bump}}
    \end{subfigure}\hfill
    \begin{subfigure}[t]{0.24\textwidth}
        \centering
        \includegraphics[width=\linewidth]{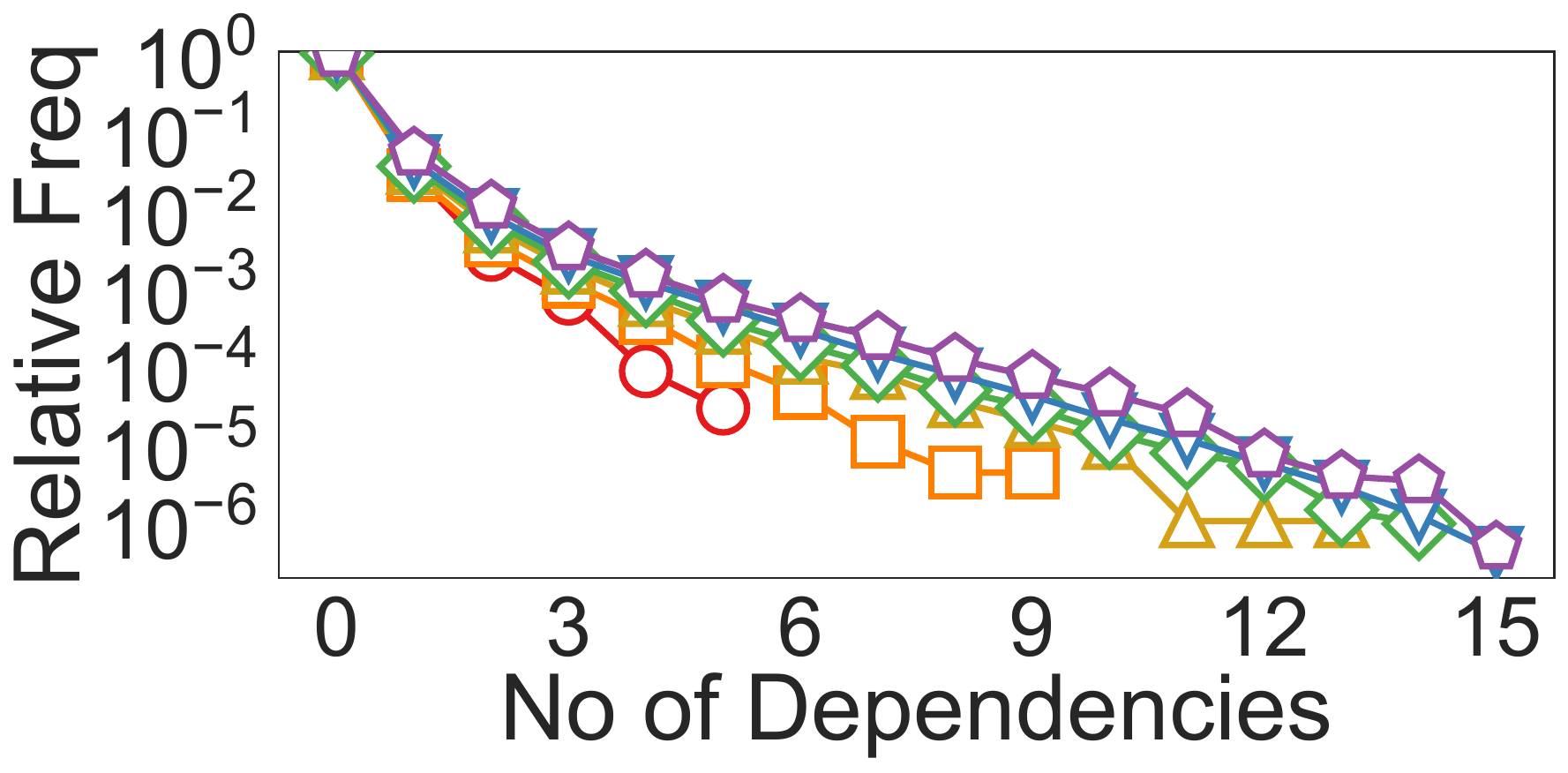}
        \vspace{-1.7em}
        \caption{\emptyLargeLA\label{fig:oneshotmapf-large_agents-empty_64_64-bump}}
    \end{subfigure}\hfill
    \begin{subfigure}[t]{0.24\textwidth}
        \centering
        \includegraphics[width=\linewidth]{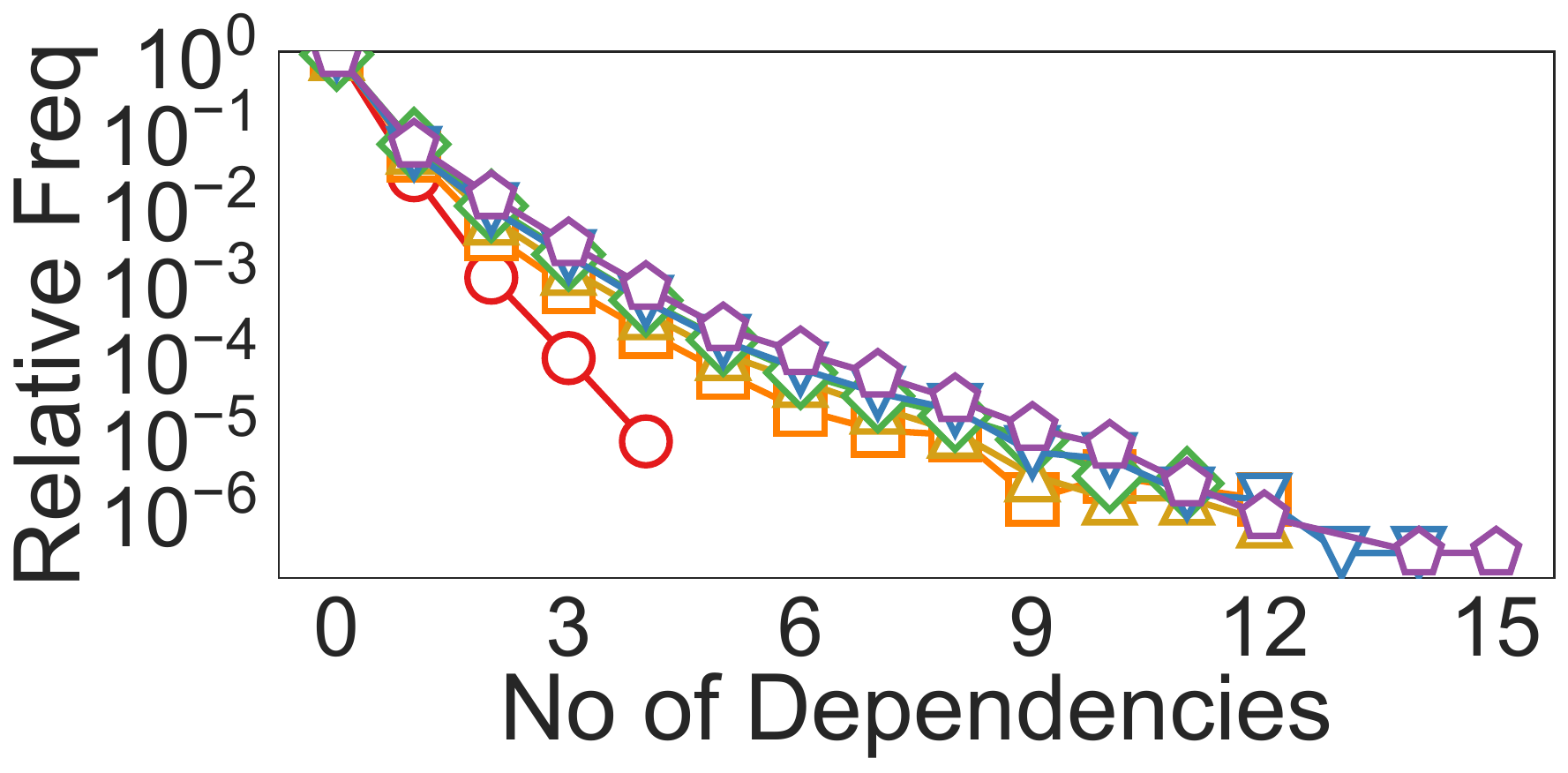}
        \vspace{-1.7em}
        \caption{\mazeLargeLA\label{fig:oneshotmapf-large_agents-maze_128_128-bump}}
    \end{subfigure}\hfill
    \begin{subfigure}[t]{0.24\textwidth}
        \centering
        \includegraphics[width=\linewidth]{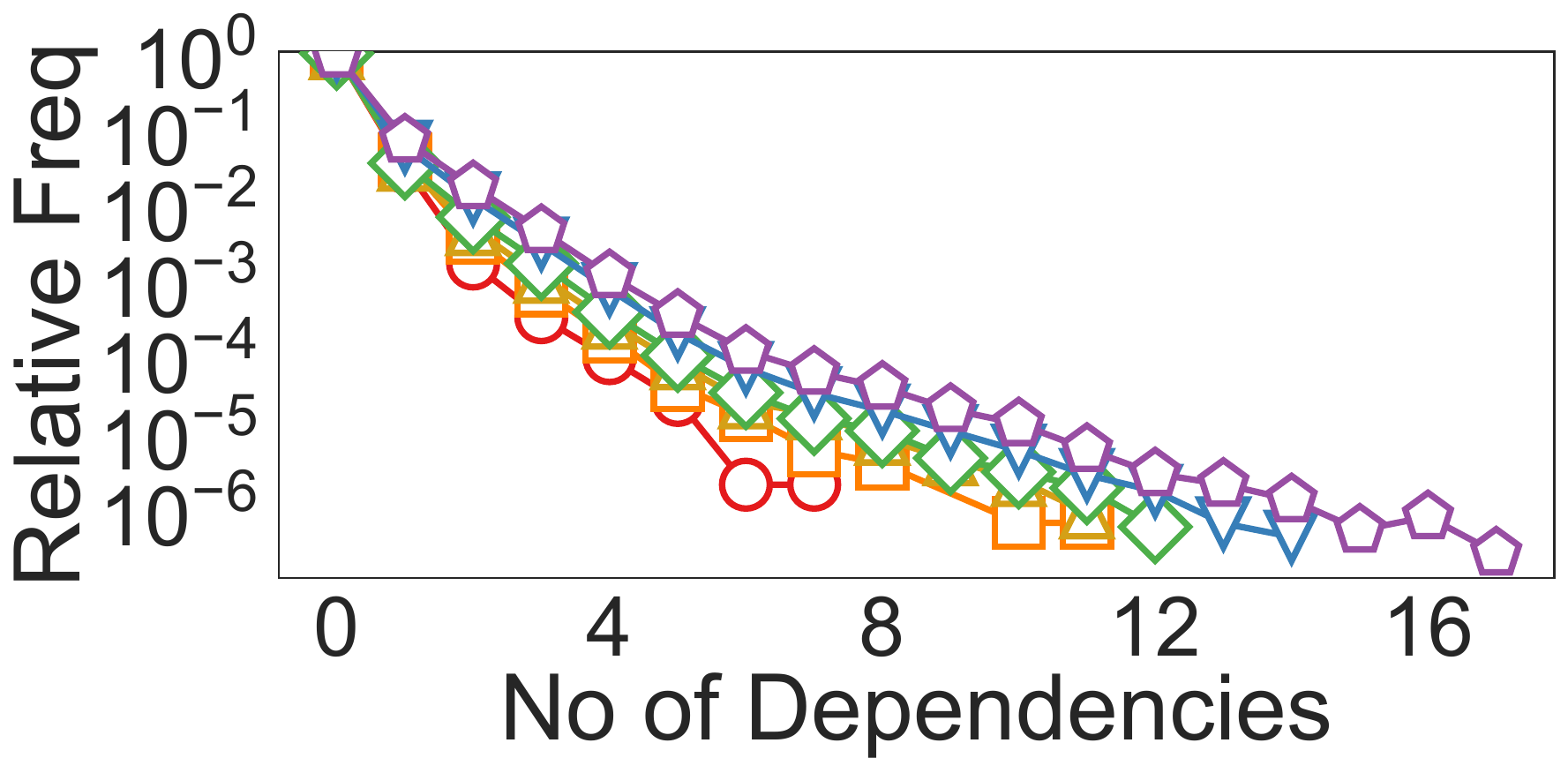}
        \vspace{-1.7em}
        \caption{\roomLargeLA\label{fig:oneshotmapf-large_agents-room_128_128-bump}}
    \end{subfigure}

    \caption{Distribution of the number of agent dependencies generated per agent per successful planning attempt, under different agent densities, with $w = 3$. $S$ and $L$ refers to the number of small and large agents, respectively.}
    \label{fig:major-result-oneshotmapf-largeagents-bump}
    \par\vspace{-\abovecaptionskip}
\end{figure*}

We compare MD-PIBT with PIBT~\cite{okumura2022priority} and EPIBT~\cite{yukhnevich2025epibt}. We do not compare with winPIBT~\cite{okumura2019winpibt}, as it has been shown to perform worse than PIBT or EPIBT.

\subsection{Experiment Setup}
\noindent \textbf{MAPF, Maps, and Agent Models.}
\Cref{tab:exp-setup} columns 1-4 summarize the MAPF problems, maps, and agent models used in the experiments.
We conduct experiments with both one-shot and lifelong MAPF and various agent models, including PM, PMLA, RM, and DDR, defined in \Cref{sec:background-mapf}. For maps, we select two maps from the MAPF benchmark~\cite{Stern2019benchmark} for PM, RM, and DDR, while creating maps by emulating patterns from the MAPF benchmark for PMLA. Each map is designed so that the largest agent can traverse the narrowest corridor. We show the agent models and maps in Appendix~\ref{appen:agent-model-map} and compute resources in Appendix~\ref{appen:compute}.

\begin{figure*}[!t]
    \centering
    \includegraphics[width=1\linewidth]{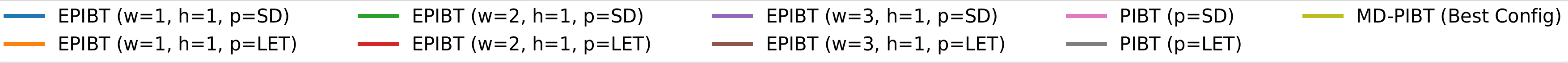}\par\medskip
    \vspace{-0.5em}
    \begin{subfigure}{0.49\textwidth}
        \centering
        \includegraphics[width=0.5\textwidth]{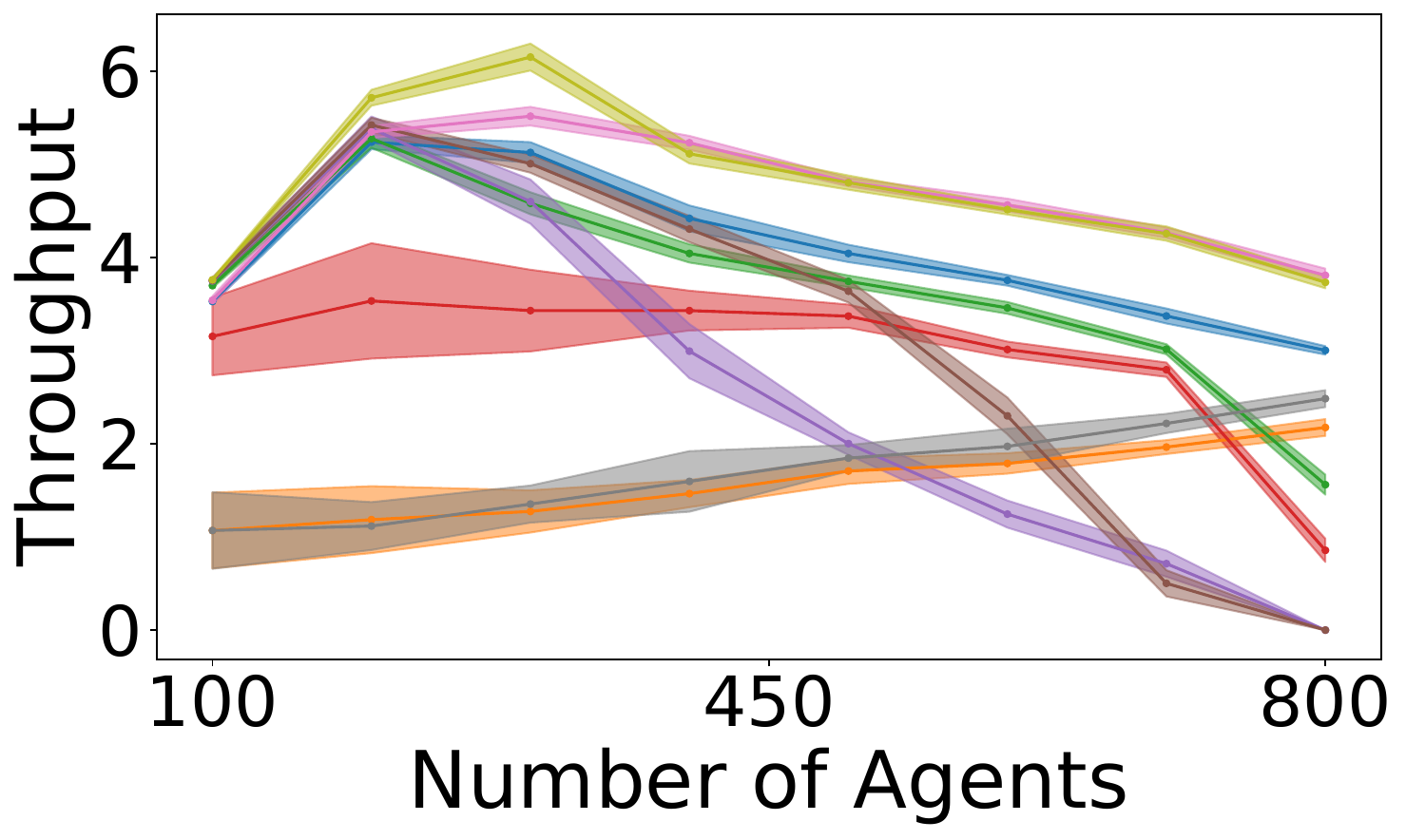}%
        \includegraphics[width=0.5\textwidth]{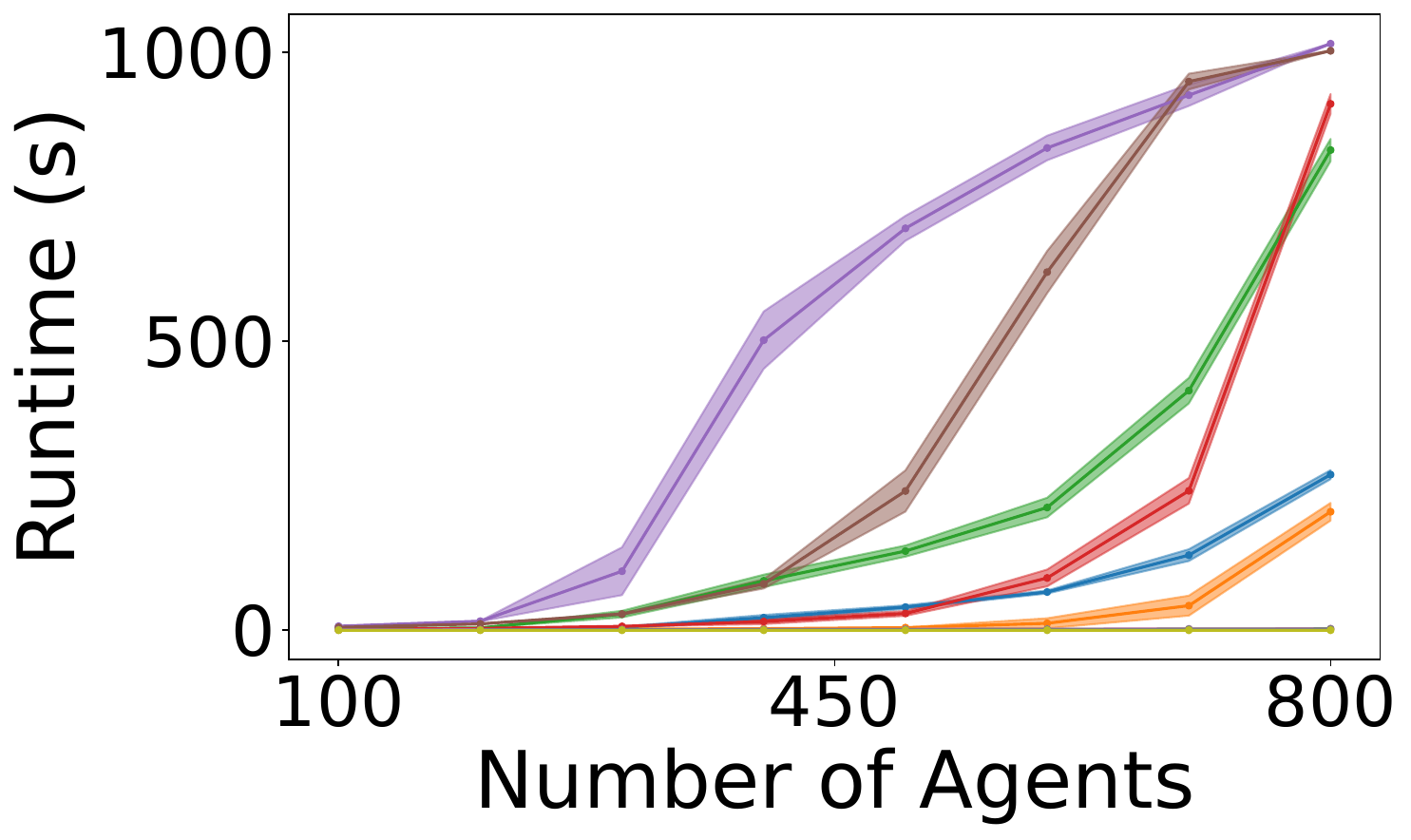}
        \vspace{-1.7em}
        \caption{\randomSmall}
        \label{fig:lmapf-pm-random-32-32-20}
    \end{subfigure}%
    \hfill
    \begin{subfigure}{0.49\textwidth}
        \centering
        \includegraphics[width=0.5\textwidth]{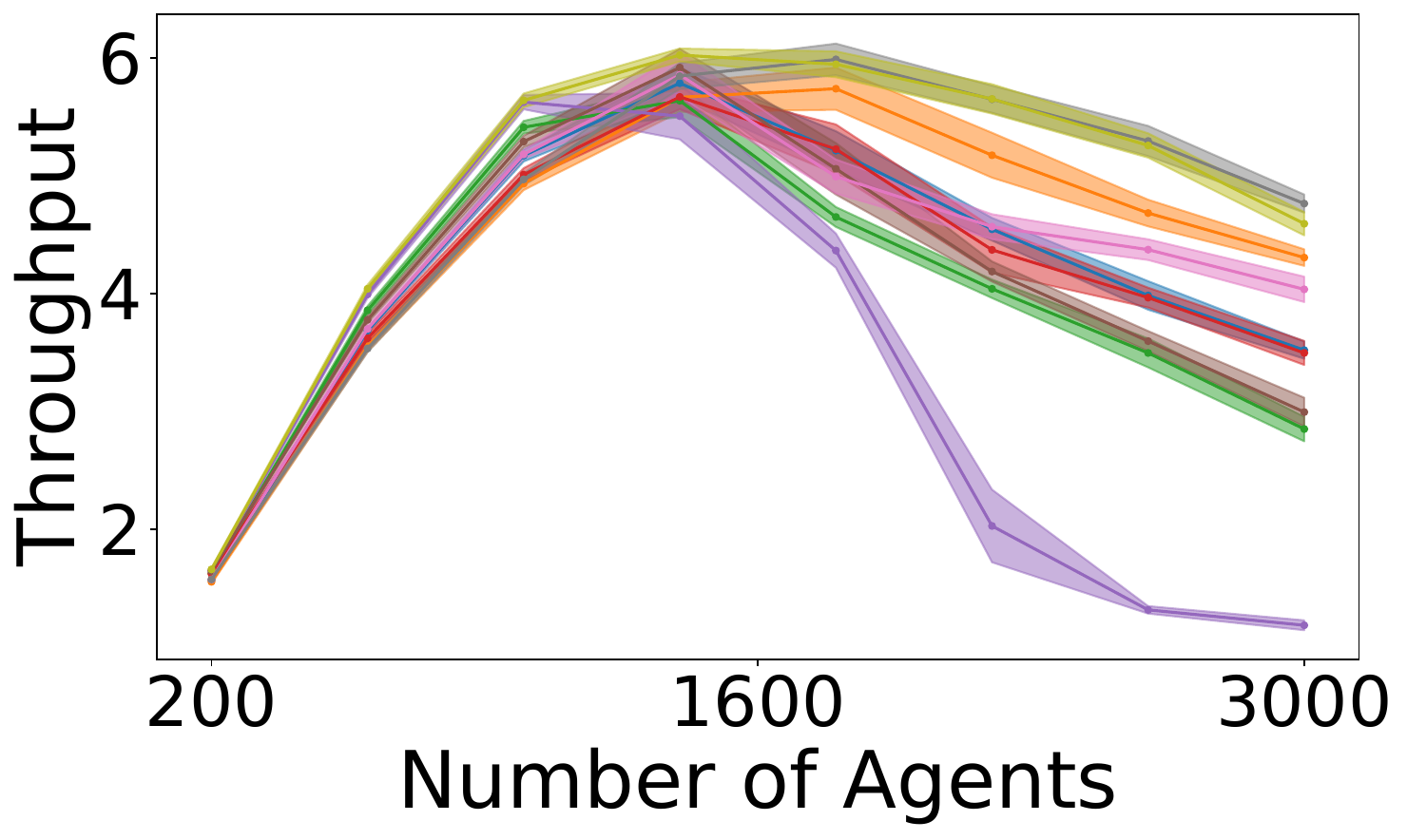}%
        \includegraphics[width=0.5\textwidth]{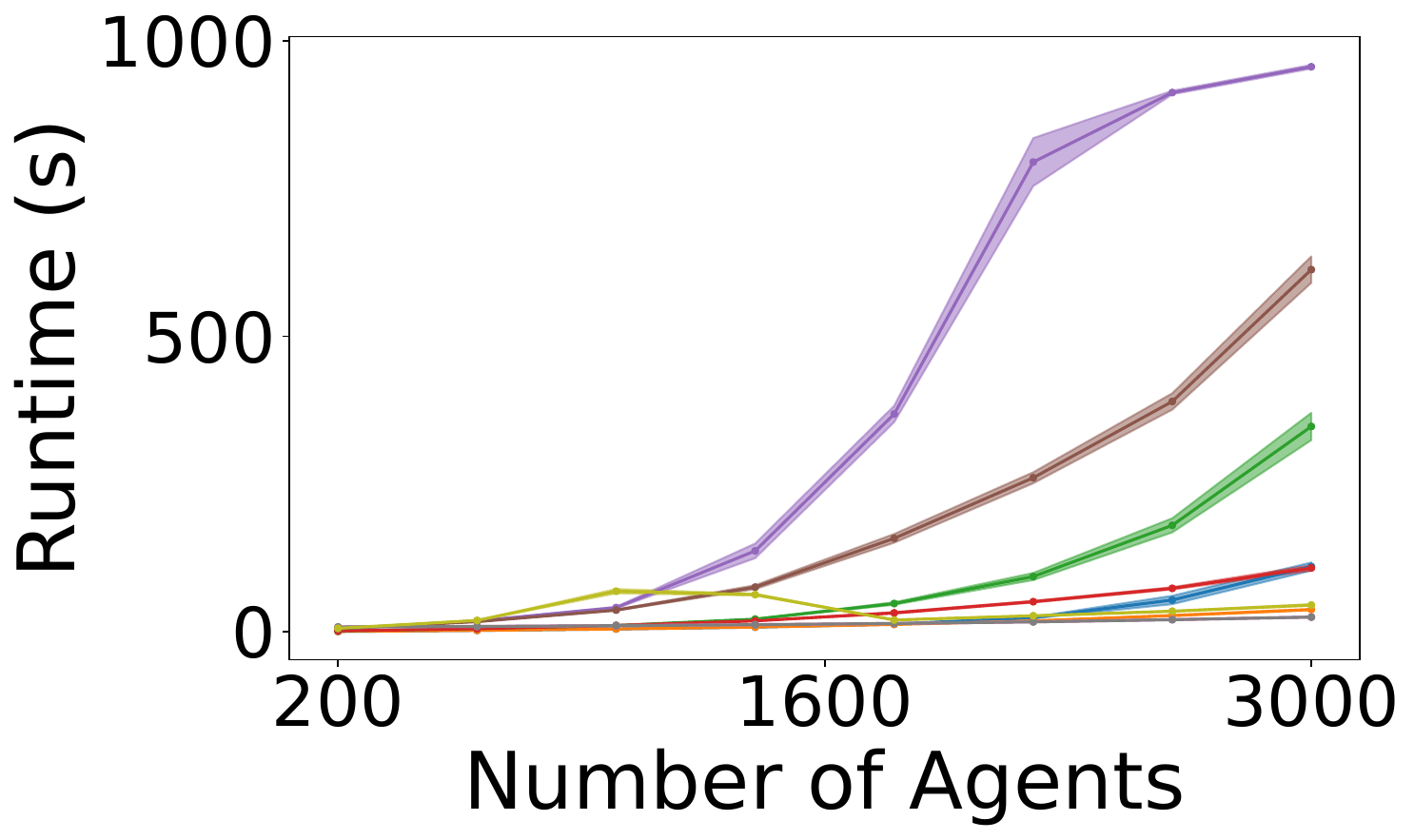}
        \vspace{-1.7em}
        \caption{\warehouseXlarge}
        \label{fig:lmapf-pm-warehouse-10-20-10-2-1}
    \end{subfigure}
    \hfill
    \caption{Throughput and runtime with different numbers of agents for LMAPF with PM agents. }
    \label{fig:major-result-lmapf-pm}
    \par\vspace{-\abovecaptionskip}
\end{figure*}

\begin{figure*}[!t]
    \centering
    \includegraphics[width=1\linewidth]{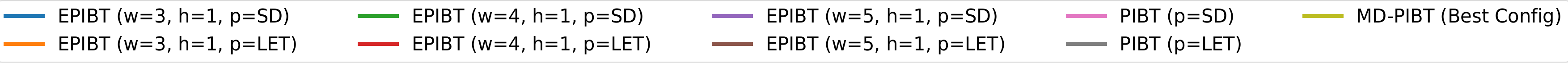}\par\medskip
    \vspace{-0.5em}
    \begin{subfigure}{0.49\textwidth}
        \centering
        \includegraphics[width=0.5\textwidth]{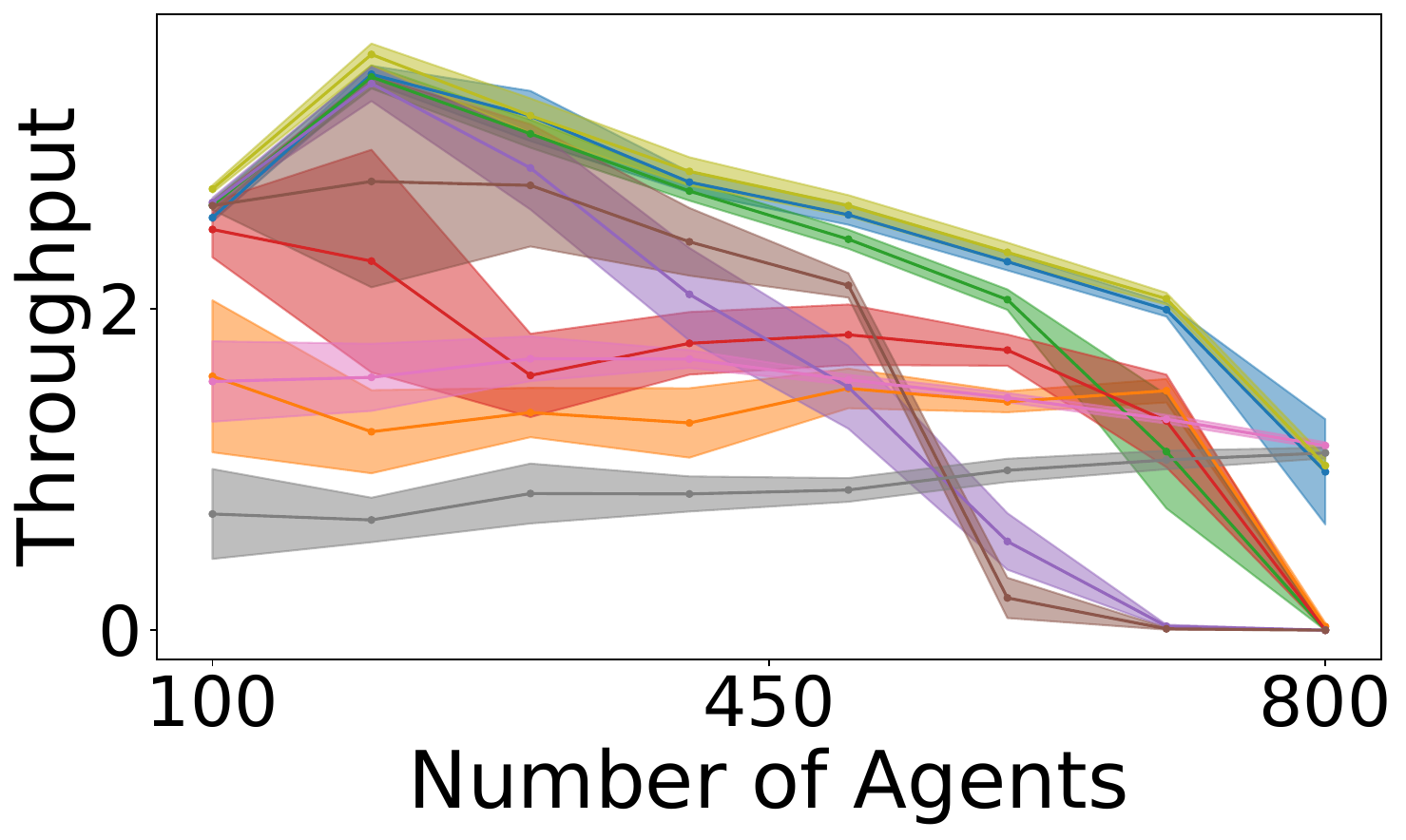}%
        \includegraphics[width=0.5\textwidth]{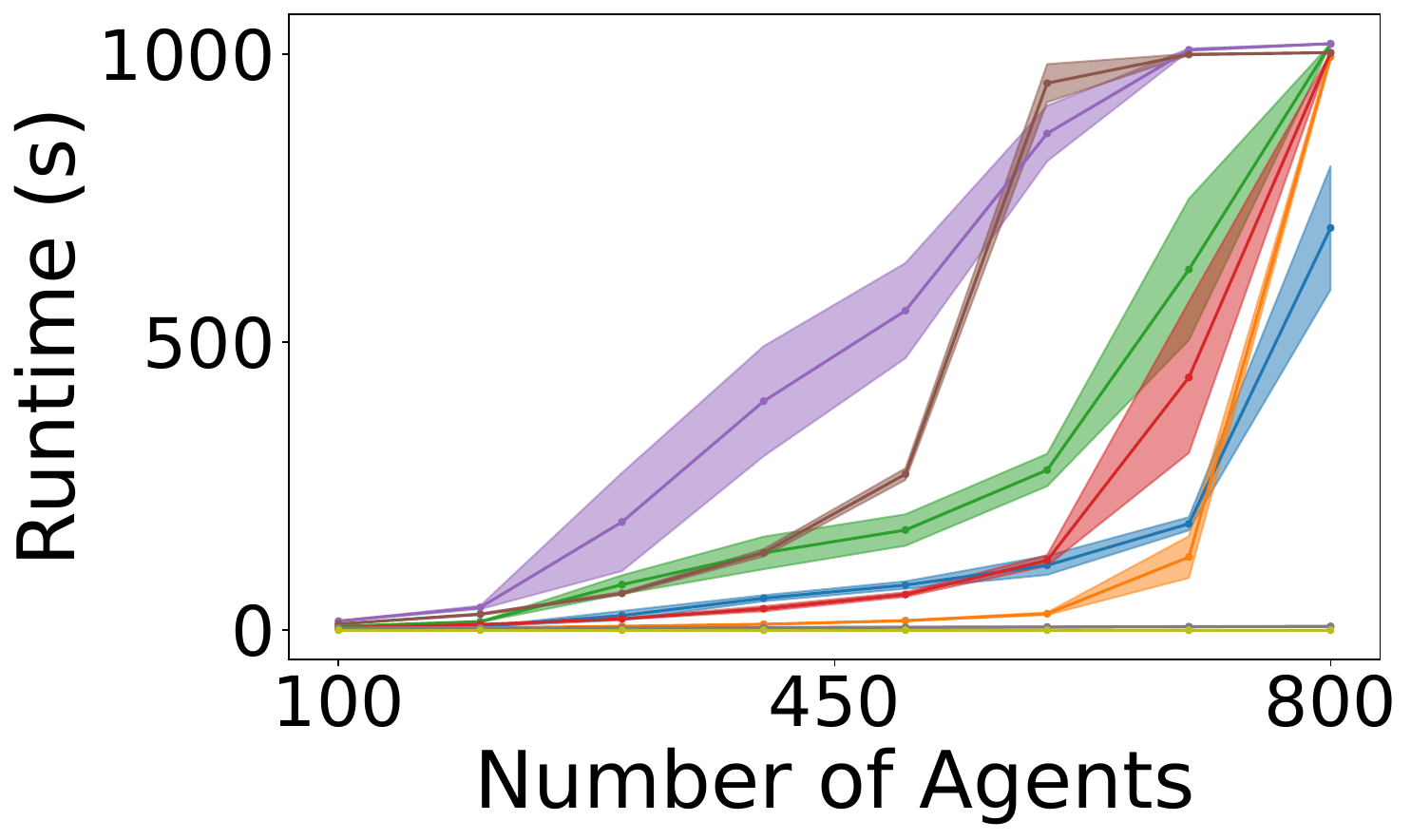}
        \vspace{-1.7em}
        \caption{\randomSmall}
        \label{fig:lmapf-rot-random-32-32-20}
    \end{subfigure}%
    \hfill
    \begin{subfigure}{0.49\textwidth}
        \centering
        \includegraphics[width=0.5\textwidth]{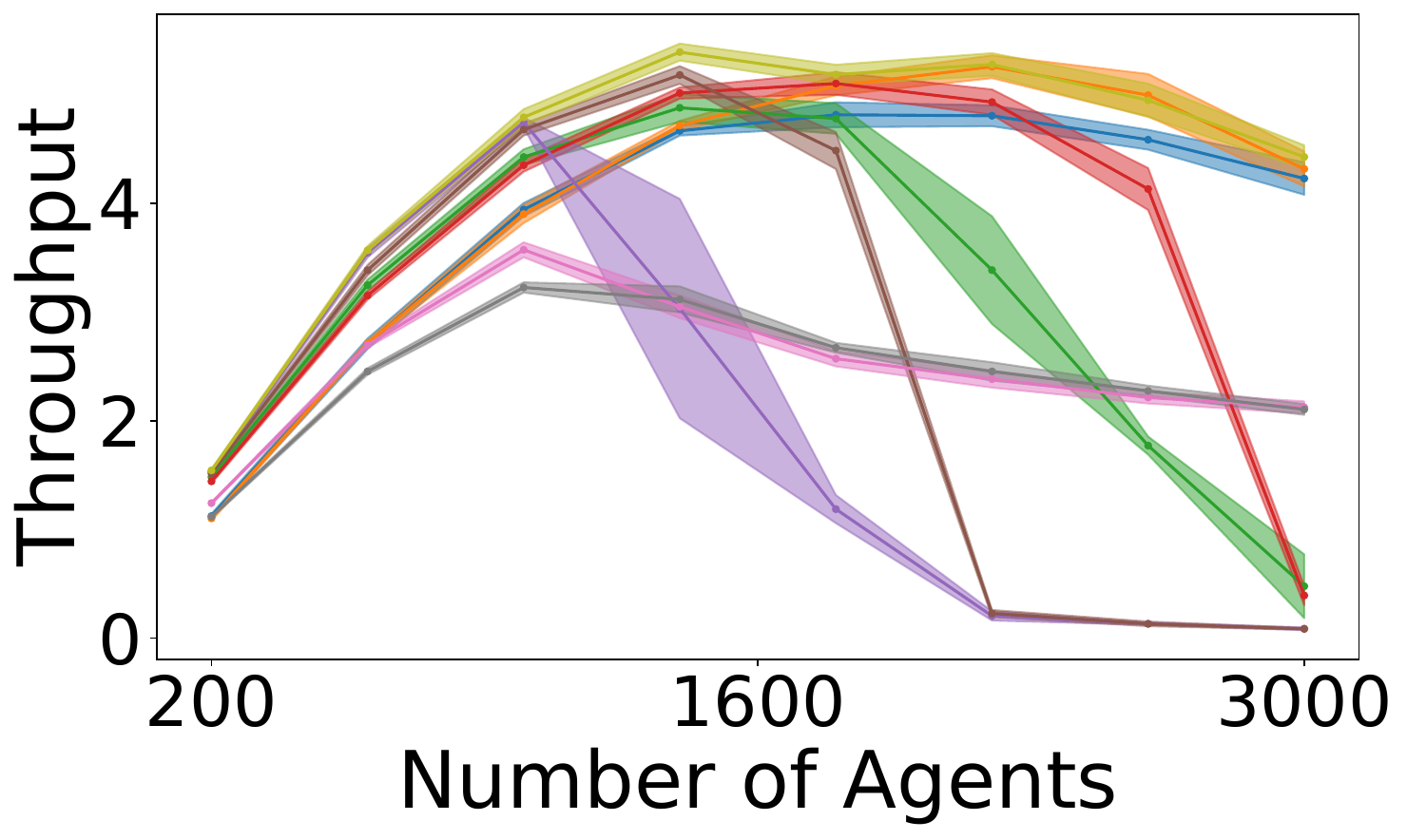}%
        \includegraphics[width=0.5\textwidth]{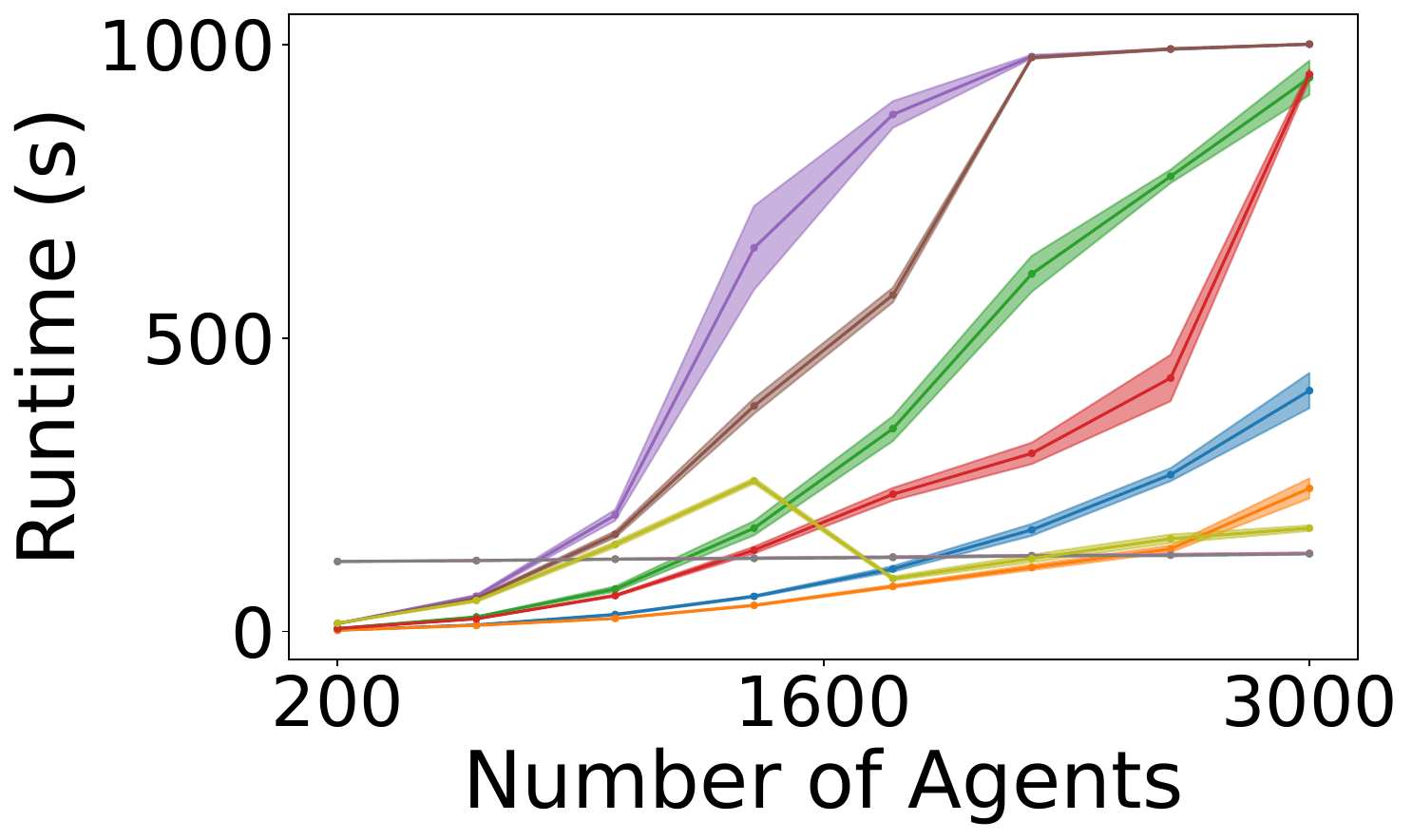}
        \vspace{-1.7em}
        \caption{\warehouseXlarge}
        \label{fig:lmapf-rot-warehouse-10-20-10-2-1}
    \end{subfigure}
    \hfill
    \caption{Throughput and runtime with different numbers of agents for LMAPF with RM agents. }
    \label{fig:major-result-lmapf-rot}
    \par\vspace{-\abovecaptionskip}
\end{figure*}

\begin{figure*}[!t]
    \centering
    \includegraphics[width=1\linewidth]{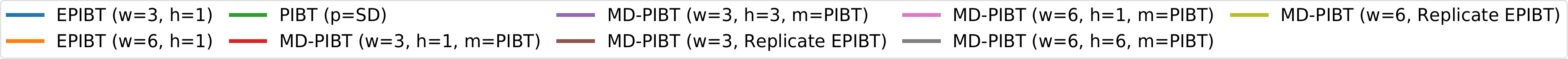}\par\medskip
    \vspace{-0.5em}
    \begin{subfigure}{0.49\textwidth}
        \centering
        \includegraphics[width=0.5\textwidth]{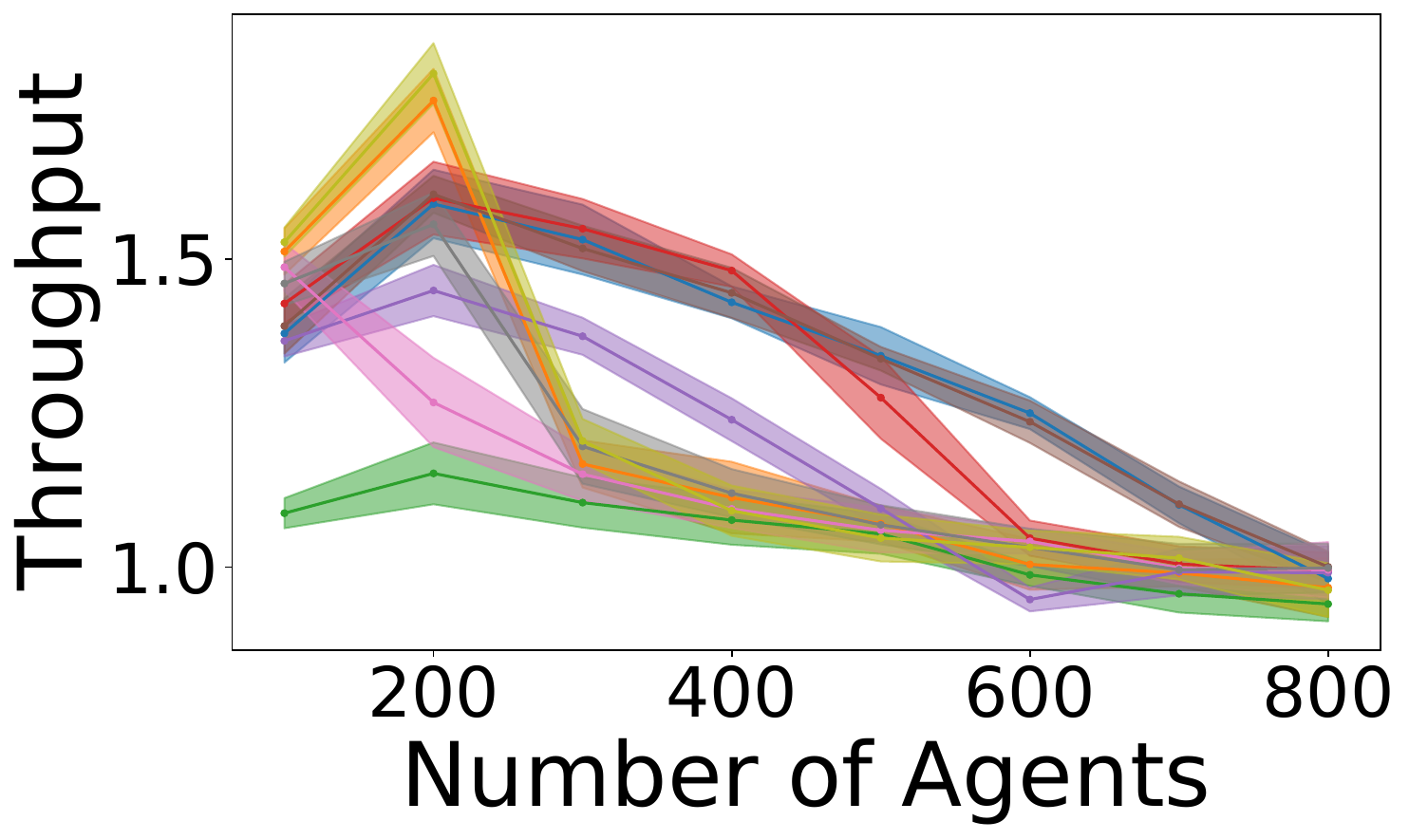}%
        \includegraphics[width=0.5\textwidth]{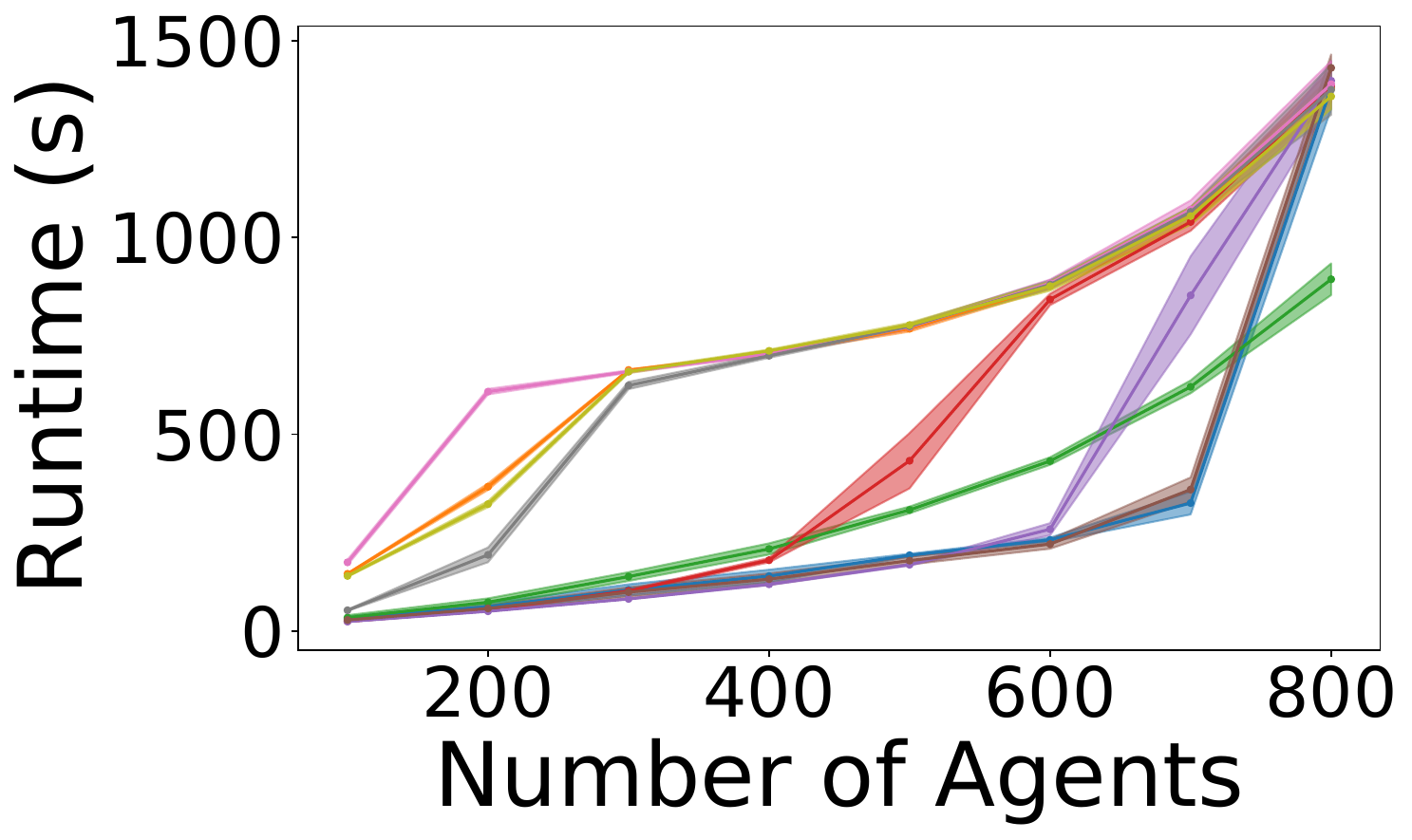}
        \vspace{-1.7em}
        \caption{\randomSmall}
        \label{fig:lsmart-random-32-32-20}
    \end{subfigure}%
    \hfill
    \begin{subfigure}{0.49\textwidth}
        \centering
        \includegraphics[width=0.5\textwidth]{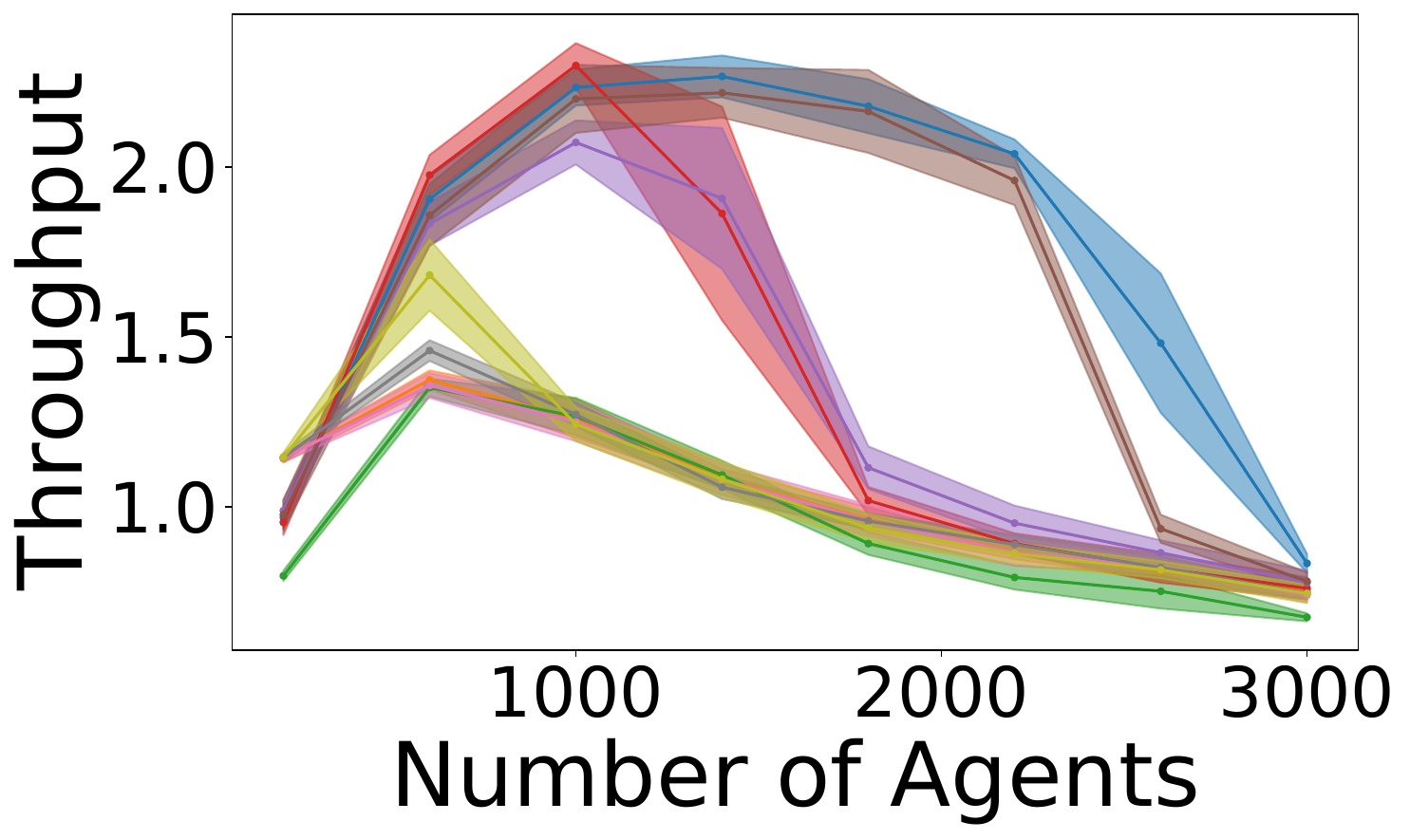}%
        \includegraphics[width=0.5\textwidth]{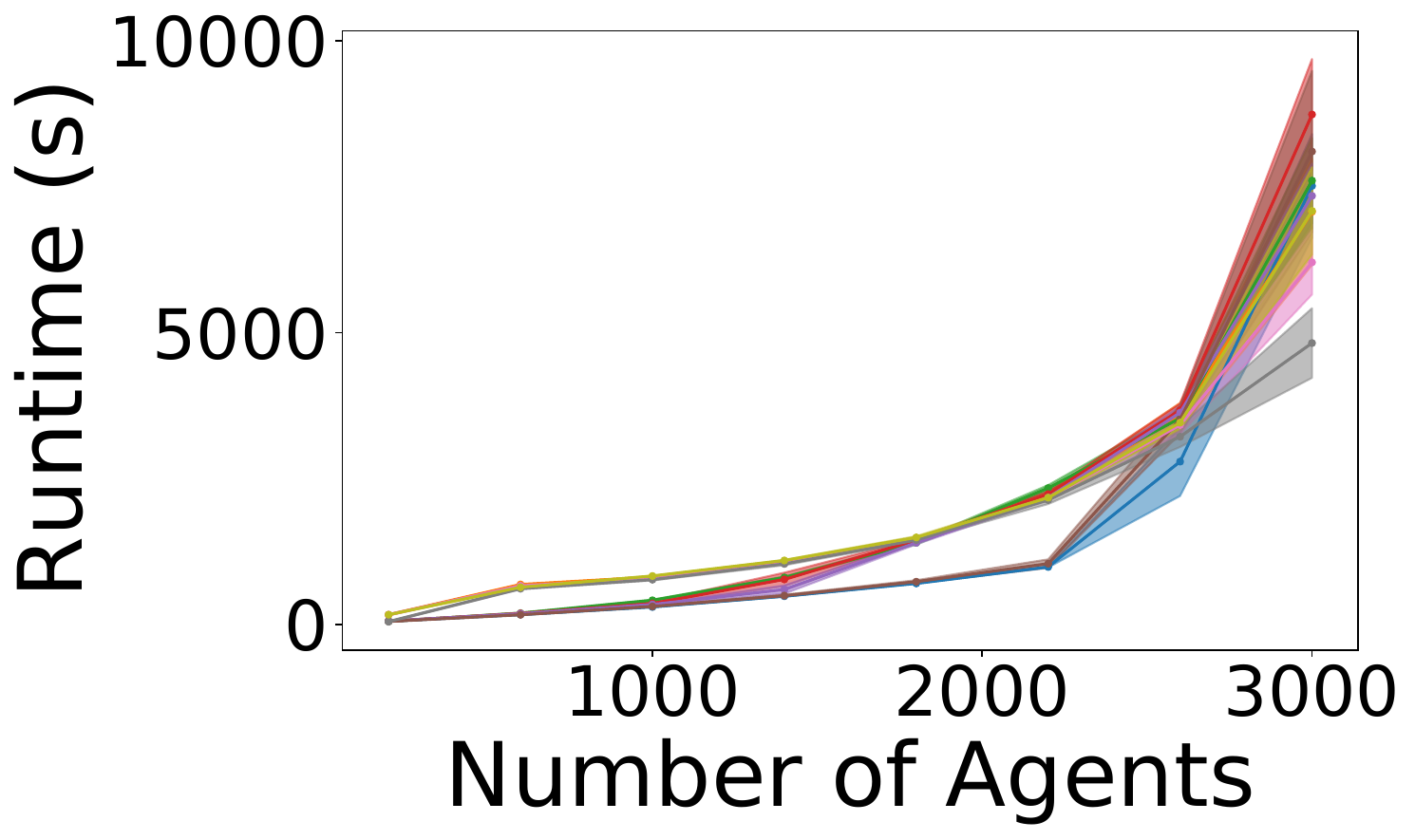}
        \vspace{-1.7em}
        \caption{\warehouseXlarge}
        \label{fig:lsmart-warehouse-10-20-10-2-1}
    \end{subfigure}
    \hfill
    \caption{Throughput and runtime with different numbers of agents for LMAPF with DDR agents. }
    \label{fig:major-result-lsmart}
    \par\vspace{-\abovecaptionskip}
    \vspace{-0.5em}
\end{figure*}

\noindent \textbf{MD-PIBT Hyper-Parameters.} 
We perform an exhaustive hyper-parameter search on MD-PIBT; \Cref{tab:exp-setup} columns 5-10 summarize the hyper-parameters. A prior work~\cite{Zhang2026mggo} shows that priority ordering in PIBT is important for resolving deadlocks in lifelong MAPF, so in lifelong experiments we also vary it: (1) Longer Elapsed Time (LET)~\cite{okumura2022priority}, sorting agents by elapsed timesteps since reaching the last goal, and (2) Shorter Distance (SD) to the goal~\cite{yukhnevich2025epibt}, sorting by distance to the next goal. For one-shot experiments, we always use LET because SD might let agents block others by staying at their goals.

All algorithms plan a $w$-step path and execute the first $h$ steps ($w \geq h \geq 1$)~\cite{yukhnevich2025epibt,li2021lifelong}. For EPIBT and MD-PIBT, if $w > h$, we initialize the safe paths to be the unexecuted paths of the previous run. For example, at timestep $t$, the algorithms plan collision-free paths $a_{1:N}.\pi^{t:t+w}$ and execute $a_{1:N}.\pi^{t:t+h}$. Then the remaining paths (i.e., $a_{1:N}.\pi^{t+h:t+w}$) plus $h$ wait actions is the initial safe path at timestep $t+h$. For PIBT, we always set $w=h=1$.

\subsection{One-Shot MAPF} \label{sec:exp-oneshot}
For one-shot experiments, we report the success rate (ratio of agents reach their goals) and runtime. We set a time limit of 1,800 seconds for PM and 3,600 seconds for PMLA. For each number of agents, we run 25 instances and plot the average (solid lines) and the 95\% confidence interval (shaded areas). We conduct all one-shot experiments in Python.

\subsubsection{Pebble Motion Agents}
In \Cref{fig:major-result-oneshotmapf-smallagents}, we show (1) a line that, for each number of agents, demonstrates the MD-PIBT variant with the highest success rate, (2) the configuration that replicates EPIBT in MD-PIBT, and (3) the baseline PIBT and EPIBT. The result shows that (1) the best MD-PIBT configuration always achieves a better or similar success rate with EPIBT, and (2) specific configurations of MD-PIBT are able to replicate EPIBT, despite longer runtimes due to the overhead of manipulating the AgDG. In general, $w=3$ or $4$ with $h = 1$ and $m =$ EPIBT perform better.

\subsubsection{Pebble Motion with Large Agents}

MD-PIBT allows an agent to bump into multiple agents. For homogeneous agents, multiple bumps occur only when $w > 1$. With PMLA, larger agents may simultaneously bump into multiple smaller agents in a single step.
In \emptyLargeLA and \mazeLargeLA, agents have sizes $1 \times 1$ and $5 \times 5$, while in \randomLargeLA and \roomLargeLA, agent sizes are $1 \times 1$ and $4 \times 4$.

\Cref{fig:major-result-oneshotmapf-largeagents} shows the result. MD-PIBT with $C>1$ significantly outperforms EPIBT in terms of success rate. Because EPIBT is limited to $C = 1$, large agents can be trapped by livelocks or deadlocks when surrounded by multiple small agents. However, when $C > 1$, large agents can push away the surrounding small agents. In addition, increasing $w$ also increases the success rate. 

\Cref{fig:major-result-oneshotmapf-largeagents-bump} shows the average number of dependencies in AgDG that each agent has in a successful run. Each line represents a scenario with different numbers of small and large agents. As the scenarios become more congested, the number of dependencies increases. Thus, it is crucial to be able to resolve multi-dependencies in scenarios with different agent sizes, especially in dense scenarios. 

\subsection{Lifelong MAPF}

For lifelong experiments with PM and RM, we set a time limit of 1 second for each timestep and run 1000 timesteps for each simulation.
In case of timeout, we run a rule-based backup planner with negligible runtime~\cite{li2021lifelong}. For DDRs, we use LSMART~\cite{YanAndZhang2026LSMART}, a tool to evaluate MAPF algorithms with DDRs in continuous time. We use planners to plan paths every 1 second and execute paths using ADG~\cite{honig2019warehouse}. Due to limits in compute, we do not perform a hyper-parameter search for DDR but rather plotted representative configurations of MD-PIBT. We report throughput and runtime. For each number of agents, we run 10 simulations and plot the average as solid lines and the 95\% confidence interval as shaded areas. We conduct all lifelong experiments in C++.

\Cref{fig:major-result-lmapf-pm,fig:major-result-lmapf-rot} shows the results of lifelong MAPF with PM and RM. In \randomSmall, with PM agents, planning with a longer window of $w=3$ is helpful to achieve the maximum throughput with 300 agents. However, with more agents, PIBT with SD priority surprisingly outperforms other methods, followed by the MD-PIBT replication of it.
Similarly, in \warehouseXlarge, PIBT with LET priority achieves the best performance.
With RM agents, EPIBT remains the best, followed by the MD-PIBT replication of it.
\Cref{fig:major-result-lsmart} shows the results of lifelong MAPF with DDR. The trend is similar to RM.

\subsection{Additional Experiments}
We show additional experiment results in Appendix~\ref{appen:add-result-result}, including 
(1) additional one-shot and lifelong MAPF results with PM, RM and PMLA agents, and
(2) a summary of the best MD-PIBT hyper-parameterizations in all settings.

\section{Conclusion}
We propose MD-PIBT as a general framework that searches over \emph{agent dependencies}. MD-PIBT removes the restrictions of PIBT/EPIBT, which can only reason about paths that collide with at most one other agent. We show that MD-PIBT performs similarly to PIBT/EPIBT in one-shot and lifelong MAPF settings while significantly outperforming them in MAPF with large agents.

There are many future works.
First, MD-PIBT could be used as a collision-shield for future learned multi-step MAPF policies~\cite{veerapaneni2024improving_mapf_policies_with_search}. 
Second, MD-PIBT is agnostic to the low-level search, similar to CBS~\cite{sharon2015conflict}. Thus, MD-PIBT could be developed as a future multi-robot protocol~\cite{veerapaneni2025cbsprotocol}.



\section*{Acknowledgments}
This work was partially supported by the National Science Foundation (NSF) under grant numbers \#$2328671$ and \#$2441629$. 
This work used Bridge-$2$ at Pittsburgh Supercomputing Center (PSC)~\cite{PSCBridgeTwo2021} through allocation CIS$220115$ from the Advanced Cyberinfrastructure Coordination Ecosystem: Services \& Support (ACCESS) program, which is supported by NSF under grant numbers \#$2138259$, \#$2138286$, \#$2138307$, \#$2137603$, and \#$2138296$.

\bibliographystyle{IEEEtran}
\bibliography{reference}

\clearpage


\section*{Appendix}
\setcounter{figure}{0}
\renewcommand{\thefigure}{A\arabic{figure}}
\setcounter{table}{0}
\renewcommand{\thetable}{A\arabic{table}}
\setcounter{section}{0}
\renewcommand{\thesection}{A\arabic{section}}

\section{Additional Algorithms} \label{appen:add-algo}

\begin{algorithm}[!t]
\caption{Path-Related Helper Functions}\label{alg:path-helper}\LinesNumbered
\SetKwProg{Fn}{Function}{:}{}
\SetKwFunction{getCollidingAgents}{getCollidingAgents}


\Fn{\parameter{\findNextBestPath}($a_k, A_{plan} \parameter{~|~m, C}$)} {

    $A_{curr} \gets \{a_i \in A | a_i.{\pi} \neq \perp \}$\\
    $\Lambda \gets \bigcup\limits_{a_i \in A \setminus A_{curr}} a_i.\tau \cup \bigcup\limits_{a_i\in A_{curr}} a_i.{\pi} $\\
    \While{$a_k.d < |a_k.\mathcal{P}|$}{
        $tp \gets a_k.\mathcal{P}[a_k.d]$\\
        $a_k.d \gets a_k.d + 1$\\
        $A_{collide} \gets$ \getCollidingAgents{$tp$, $\Lambda$}\\
        \If{$|A_{collide}| > C$}{
            \textbf{continue}
        }
        \uIf{$m = \text{PIBT}$ \textbf{and} $A_{collide} \cap A_{plan} = \emptyset$}{
            \Return $tp$\\
        }
        \ElseIf{$m = \text{EPIBT}$ \textbf{and} $A_{collide} \cap A_{curr} = \emptyset$ \textbf{and} $\forall a_c\in A_{collide},a_c.p > p, a_c.r < R$}{
            \Return $tp$\\
        }
    }
    \Return $\perp$
}

\end{algorithm}

\Cref{alg:path-helper} shows the pseudocode of our implementation of the \texttt{findBestPath} function, used in \Cref{alg:md-pibt-main}. Depending on the find path mode $m$, the function finds a valid path for agent $a_k$ that does not collide with a different set of agents. 

We start by finding all agents $A_{curr}$ that currently have a tentative path (line 2). We then form a reservation table $\Lambda$ that consists of \emph{the most recent path} of each agent (line 3). 
Specifically, if an agent $a_i \in A_{curr}$, we add its tentative path $a_i.\pi$ to $\Lambda$, otherwise, we add its safe path $a_i.\tau$ to $\Lambda$. 
We then iterate through the candidate paths $a_k.\mathcal{P}$ starting from the current path index $a_k.d$ (lines 4-5). After incrementing the path index (line 6), we obtain the set of agents $A_{collide}$ whose most recent path collides with the current candidate path $tp$ (line 7). Depending on the find path mode $m$ and the maximum number of colliding agents $C$, we judge whether $tp$ is valid. 
For all $m$, if $tp$ collides with more than $C$ agents, it is not valid, and we proceed to the next candidate path (lines 8-9).
If $m=\text{PIBT}$, then $tp$ should not collide with the most recent paths of the agents in $A_{plan}$ (line 10). If $m=\text{EPIBT}$, $a_k$ should not collide with (1) agents that have a tentative path in the current \texttt{MDPIBT} function call, (2) agents with higher priorities, and (3) agents that have no replan attempts left (line 12). If $tp$ is valid, we return it (lines 11 and 13). If no valid paths are found, we return the failure (line 14).

\section{Theoretical Analysis} \label{appen:theory}

\subsection{Runtime Complexity}

To compute the worst-case runtime of MD-PIBT, we consider a hyper-parameterization of MD-PIBT with $R=\infty$ and arbitrary $w$, $h$, $m$, $p$, and $C$ ($w \geq h \geq 1$, $C \geq 1$). In addition, suppose that each agent $a_i$ ($i \in \{1,\ldots,N\}$) has a set of all possible paths of length $w$, denoted by $a_i.\mathcal{P}$, with size $|a_i.\mathcal{P}| = K^w$, where $K$ is the number of discrete actions per timestep. Let
\begin{equation}
\mathbb{P} = a_1.\mathcal{P} \times a_2.\mathcal{P} \times \cdots \times a_N.\mathcal{P}
\end{equation}
be the set of all possible combinations of paths for the $N$ agents.

\begin{theorem} \label{th1}
With $R=\infty$, during one invocation of the \texttt{MDPIBT} function in \Cref{alg:md-pibt-main}, each path combination $\rho \in \mathbb{P}$ is tried at most once. 
\end{theorem}

\begin{proof}
According to \Cref{alg:agent-class-main}, we associate each agent $a_i$ with a path index $a_i.d \in \{0,1,\ldots, |a_i.\mathcal{P}|\}$. When $a_i.d = 0$, $a_i$ has no tentative path. Otherwise, it is trying the $a_i.d$-th path in $a_i.\mathcal{P}$. To capture the case where $a_i.d=0$, let $a_i.\mathcal{P'} = a_i.\mathcal{P} \cup \{\perp\}$ for all agents. We then let

\begin{equation}
\mathbb{P'} = a_1.\mathcal{P'} \times a_2.\mathcal{P'} \times \cdots \times a_N.\mathcal{P'}.
\end{equation}

A path combination $\rho' \in \mathbb{P'}$ can then be represented by the path-index vector
\begin{equation}
\mathbf{d} = (a_1.d, a_2.d, \ldots, a_N.d),
\end{equation}
where $a_i.d$ specifies the current path selected for agent $a_i$ (or $a_i.d=0$ denoting that no path has yet been selected). Thus, proving that no path combination is tried twice is equivalent to proving that the same path-index vector $\mathbf{d}$ cannot appear twice during a single \texttt{MDPIBT} function call.

With $R=\infty$, according to \Cref{alg:md-pibt-main}, when an agent $a_i$ fails to find a valid path, it always asks one of its parent agents $a_j$ to replan. That parent agent then advances to another candidate path, i.e., its path index $a_j.d$ is increased to a previously untried value. 

If the failure does not backtrack to the root agent, this will advance $\mathbf{d}$ to a untried state where some ancestor agent $a_p$ of $a_i$ incremented path index.
After $a_p.d$ is incremented, any future $\mathbf{d}$ must contain this new value of $a_p.d$. Hence, the previous $\mathbf{d}$ cannot be reconstructed. Descendant agents may subsequently choose new paths, but the changed parent index prevents the entire vector from becoming identical to the failed one.

If the failure eventually backtracks to the root agent $a_{root}$, it is guaranteed that $a_{root}$ can select its safe path, because all other agents are also taking their safe paths. Therefore, the search either returns a valid path combination before exhausting $\mathbb{P'}$, or, in the worst case, it enumerates all combinations in $\mathbb{P'}$ exactly once before terminating. Therefore, each path combination in $\mathbb{P}$ is tried at most once.
\end{proof}

Now we can compute the worst-case runtime of MD-PIBT with $R=\infty$.

\begin{theorem}
    With $R=\infty$, the worst-case runtime of MD-PIBT is $O(K^{wN})$.
\end{theorem}

\begin{proof}
According to \Cref{th1}, the worst-case runtime of MD-PIBT is upper bounded by the number of path combinations we have tried, which is upper bounded by
\begin{equation}
    O(|\mathbb{P}|)
    =
    O(\prod_{i=1}^{N} |a_i.\mathcal{P}|)
    =
    O((K^w)^N)
    =
    O(K^{wN}).
\end{equation}
\end{proof}

This bound is both useless and useful. The bound is useless in that $O(K^{wN})$ is the size of enumerating all possible paths, which is the naive bound. However, the bound is insightful as it reveals how MD-PIBT is searching through the entire search space when encountering failures. Finally, we can show that this bound is somewhat tight in that we can construct an instance (where all agents lined up going into a dead-end corridor) which takes $O(K^{wN})$ planning calls.

Now consider a finite $R$ ($R \ll \infty$).

\begin{theorem}
    With a finite positive $R$, the worst-case runtime of MD-PIBT is $O(NRK^w)$.
\end{theorem}

\begin{proof}
    Using a finite $R$ gives an upper bound for how many times an agent may replan. As mentioned in \Cref{alg:agent-class-main,alg:md-pibt-main}, each agent $a_i$ has a planning attempt counter $a_i.r$, which increments by 1 each time $a_i$ invokes \texttt{findBestPath} in \Cref{alg:md-pibt-main} line 6. Given that we have $N$ agents and we never decrement $r$, in the worst case, we invoke \texttt{findBestPath} $NR$ times. From the proof of \Cref{th1}, we know that each agent $a_i$ has $|a_i.\mathcal{P}| = K^w$ paths, where $K$ is the number of discrete actions and $w$ is the planning window. Therefore, the worst case runtime of MD-PIBT with a finite $R$ is given by the total number of paths being tried: $O(NRK^w)$.
\end{proof}

Note that $O(NRK^w) \ll O(K^{wN})$, as each agent can fail only a maximum of $R$ times before choosing the safe path, and is never replanned again afterwards.

\section{Additional Experiments} \label{appen:add-result}

\subsection{Visualization of Agent Models and Maps} \label{appen:agent-model-map}

\begin{figure*}[!t]
    \begin{subfigure}[t]{0.24\textwidth}
        \centering
        \includegraphics[width=\linewidth]{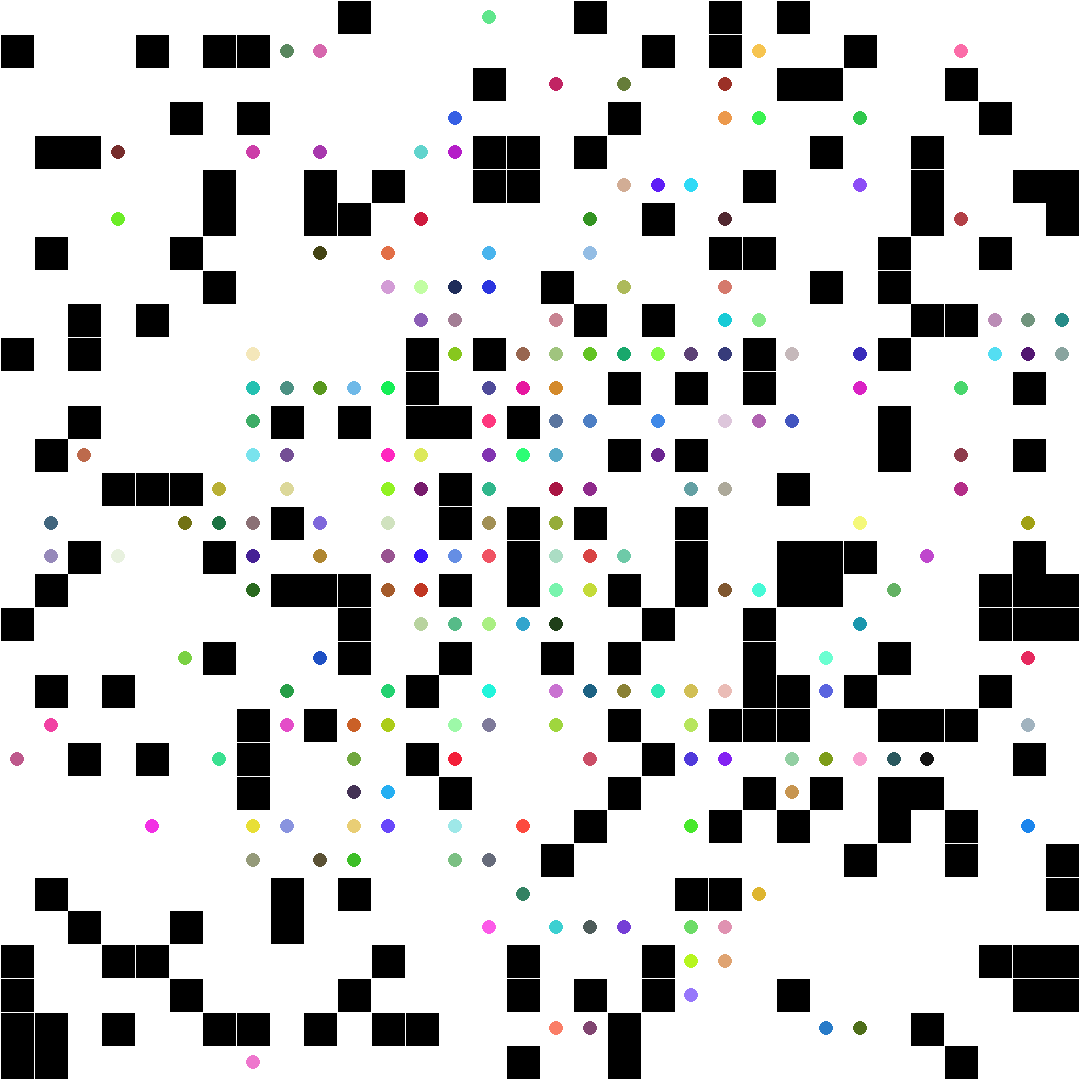}
        \caption{Pebble Motion (PM)}
    \end{subfigure}\hfill
    \begin{subfigure}[t]{0.24\textwidth}
        \centering
        \includegraphics[width=\linewidth]{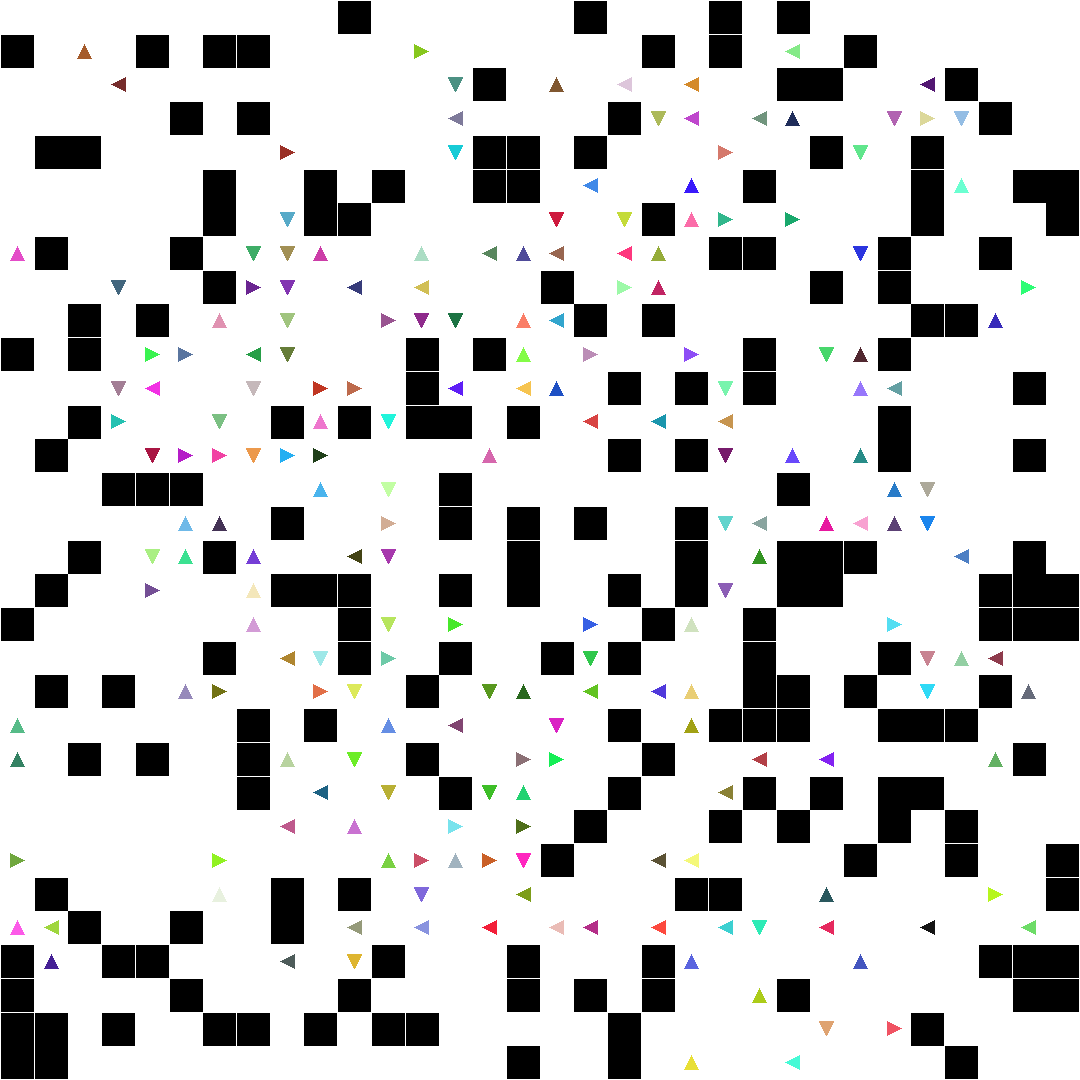}
        \caption{Rotation Motion (RM)}
    \end{subfigure}\hfill
    \begin{subfigure}[t]{0.26\textwidth}
        \centering
        \includegraphics[width=\linewidth]{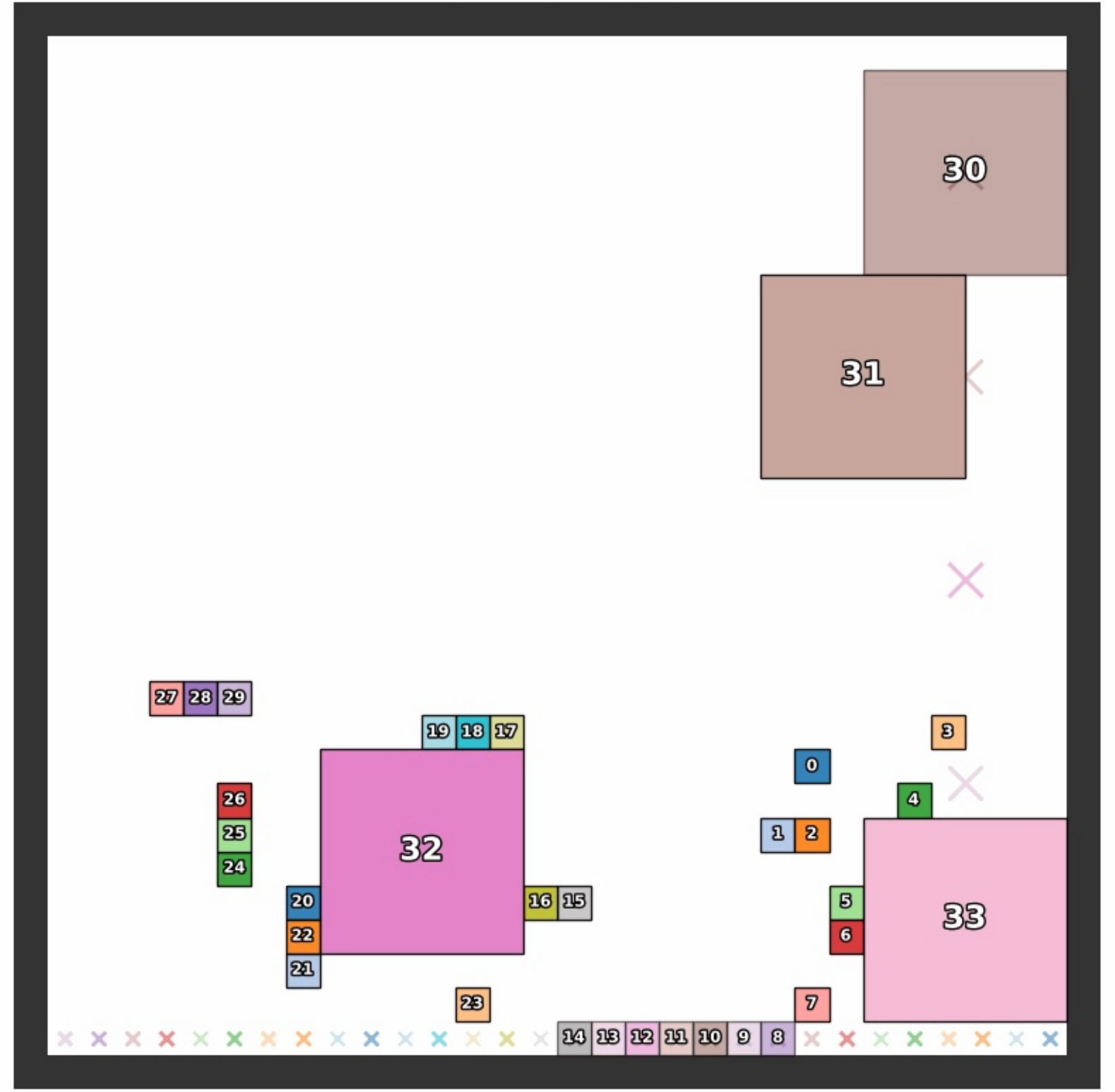}
        \caption{Pebble Motion w/ Large Agents (PMLA)}
    \end{subfigure}\hfill
    \begin{subfigure}[t]{0.24\textwidth}
        \centering
        \includegraphics[width=\linewidth]{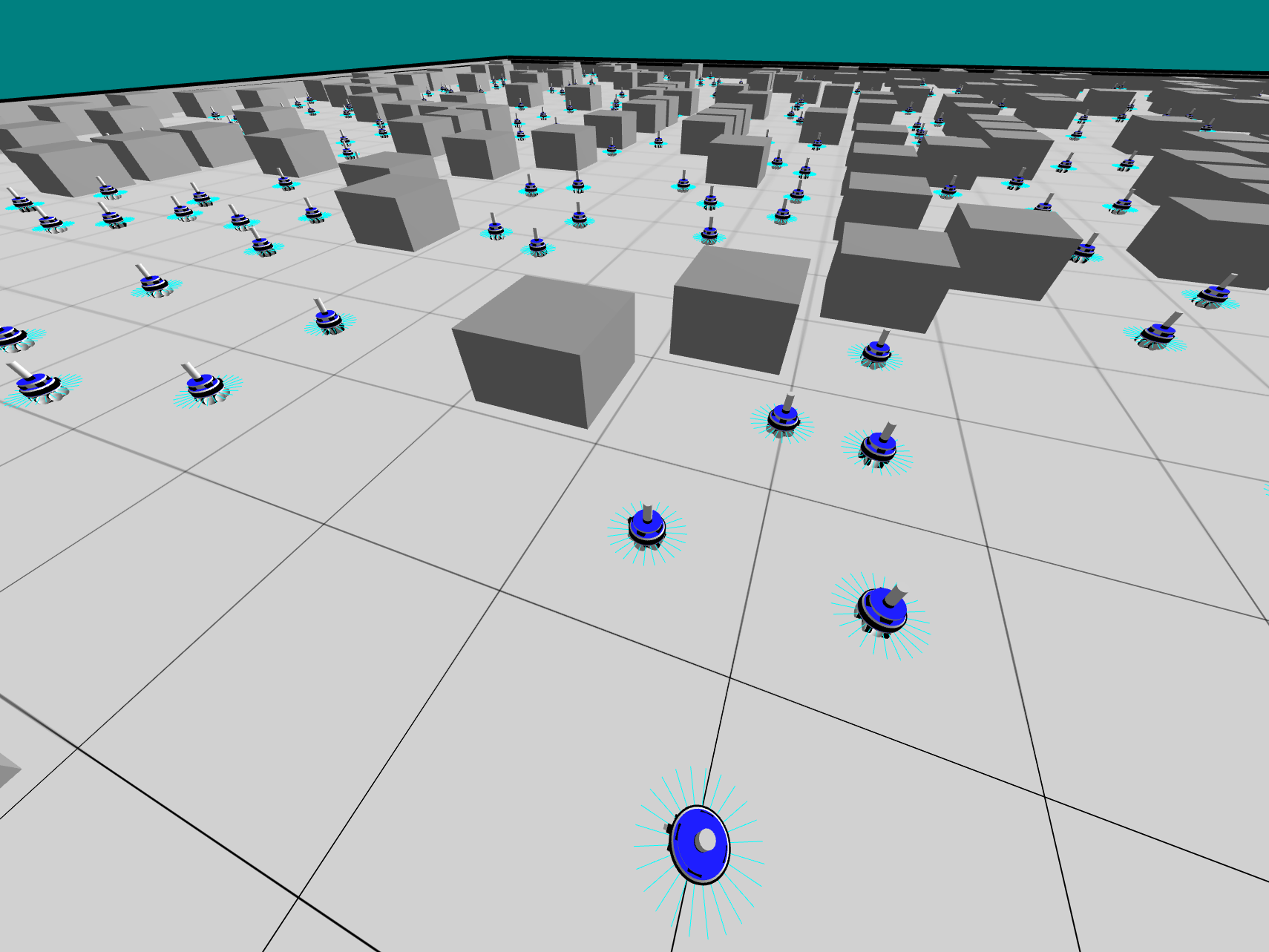}
        \caption{Differential Drive Robots (DDR)}
    \end{subfigure}\hfill
    \caption{Visualization of agent models used in our one-shot/lifelong MAPF experiments.}
    \label{fig:mapf-agents}
\end{figure*}

\begin{figure*}[!tp]
    \begin{subfigure}[t]{0.19\textwidth}
        \centering
        \includegraphics[width=\linewidth]{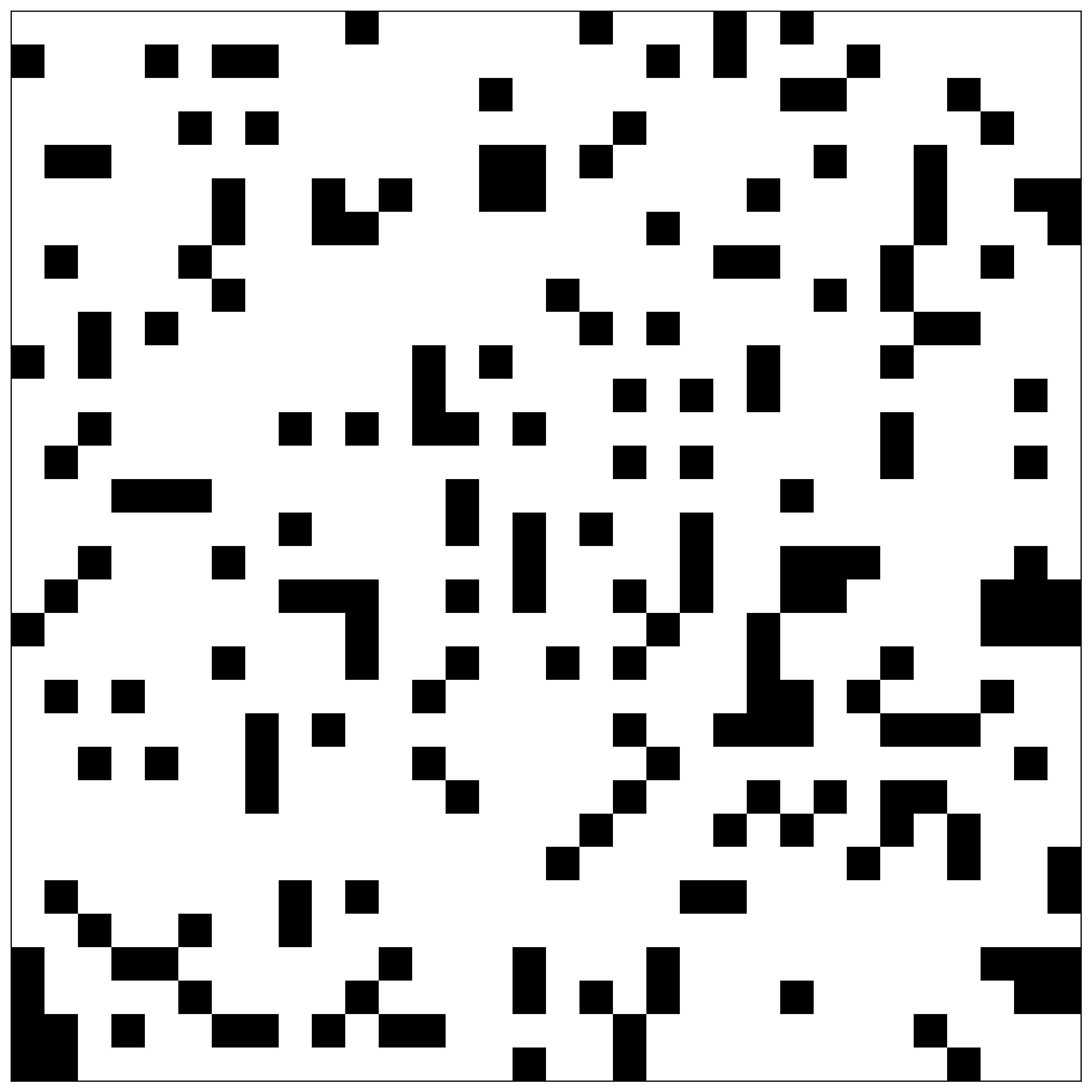}
        \caption{\randomSmall}
    \end{subfigure}\hfill
    \begin{subfigure}[t]{0.19\textwidth}
        \centering
        \includegraphics[width=\linewidth]{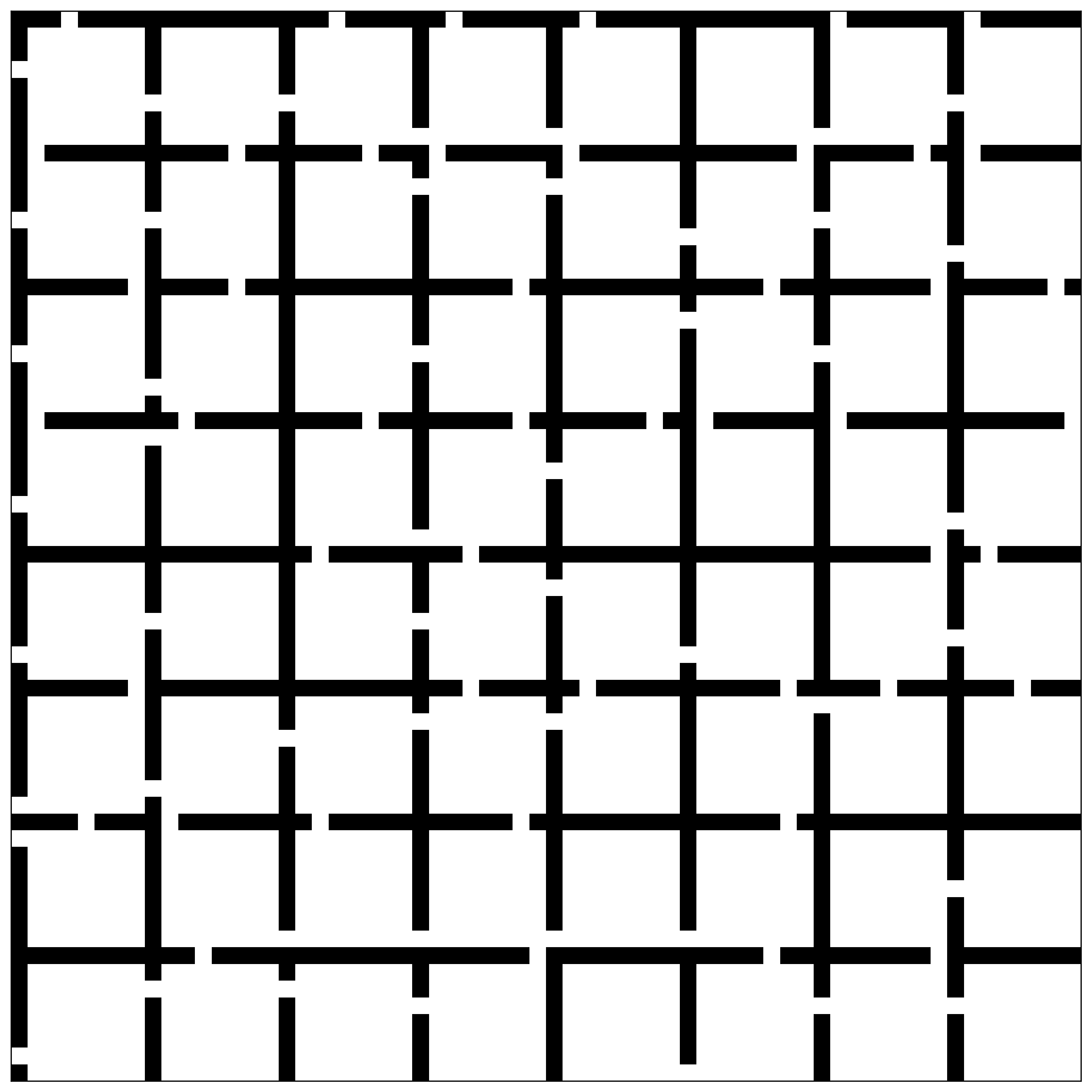}
        \caption{\roomLarge}
    \end{subfigure}\hfill
    \begin{subfigure}[t]{0.19\textwidth}
        \centering
        \includegraphics[width=\linewidth]{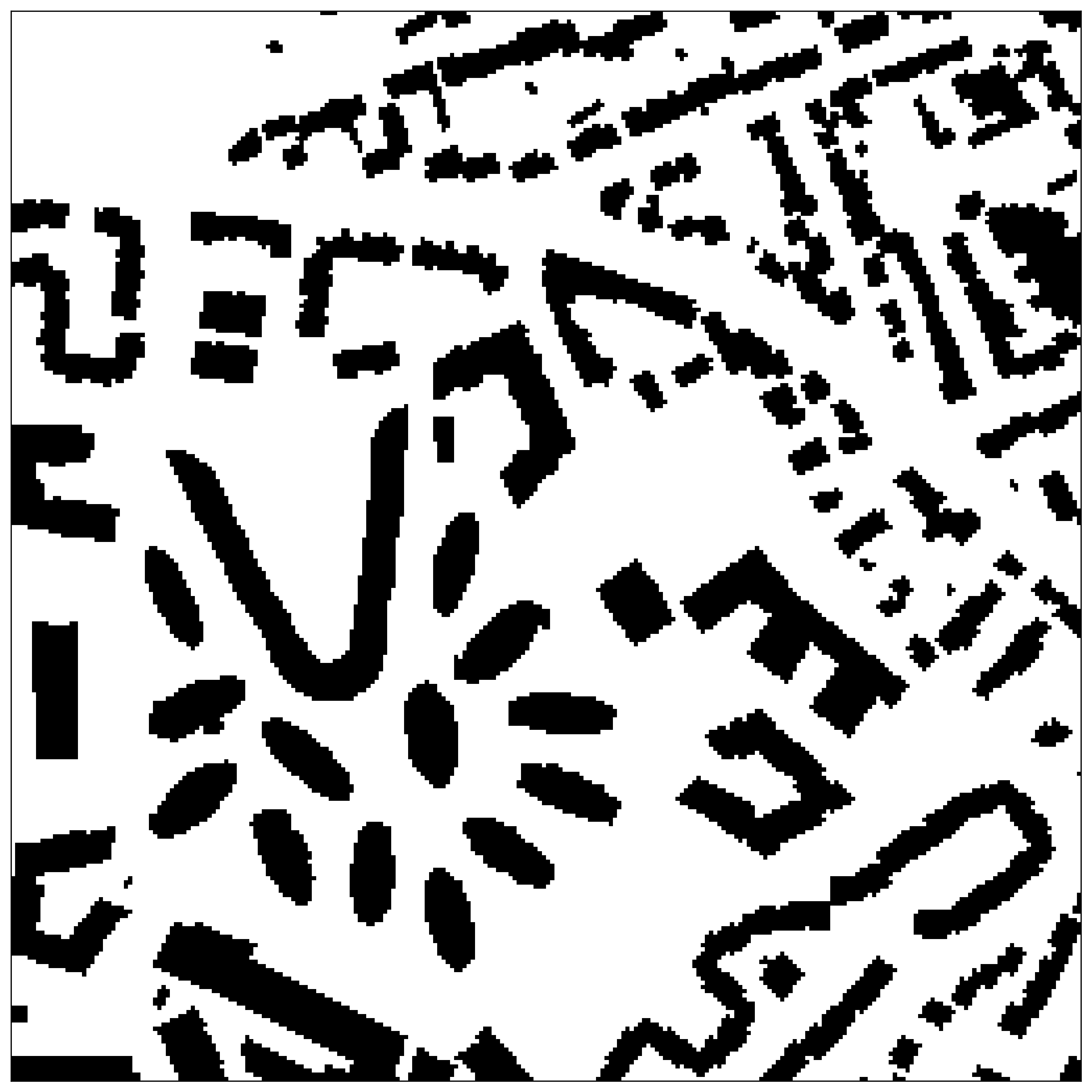}
        \caption{\paris}
    \end{subfigure}\hfill
    \begin{subfigure}[t]{0.4\textwidth}
        \centering
        \includegraphics[width=\linewidth]{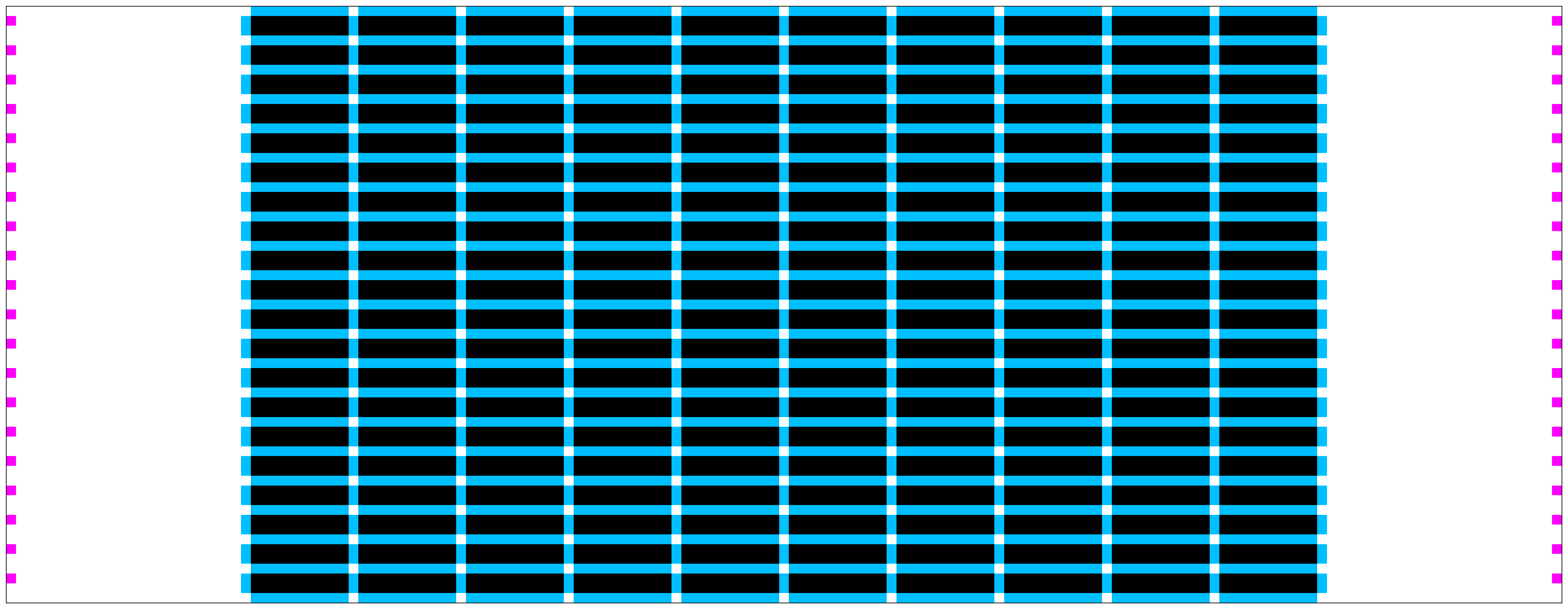}
        \caption{\warehouseXlarge}
    \end{subfigure}\hfill
    \caption{Visualization of maps used in one-shot and lifelong experiments with PM, RM, and DDR agents. In all maps, black tiles are obstacles, white tiles are empty spaces, blue tiles are endpoints, and purple tiles are workstations. In \randomSmall, \roomLarge, and \paris, agents move between randomly selected empty spaces (white). In \warehouseXlarge, agents move between randomly selected workstations (purple) and endpoints (blue).}
    \label{fig:mapf-maps}
\end{figure*}

\begin{figure*}[!tp]
    \begin{subfigure}[t]{0.24\textwidth}
        \centering
        \includegraphics[width=\linewidth]{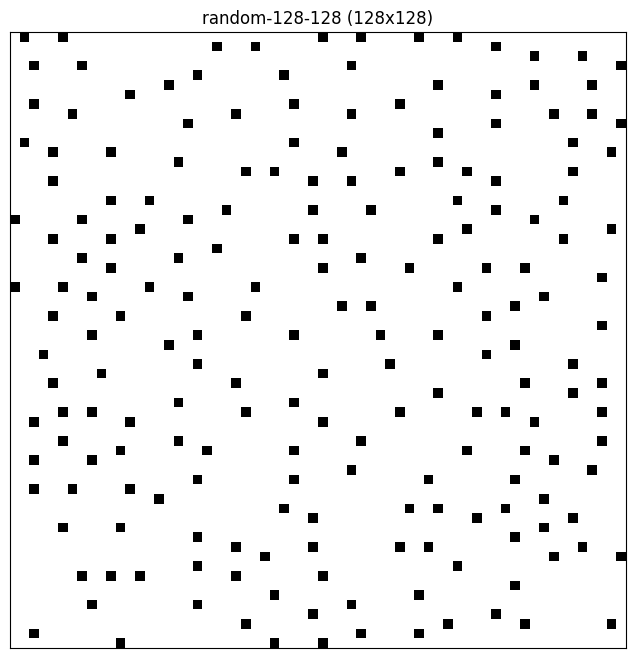}
        \caption{\randomLargeLA}
    \end{subfigure}\hfill
    \begin{subfigure}[t]{0.24\textwidth}
        \centering
        \includegraphics[width=\linewidth]{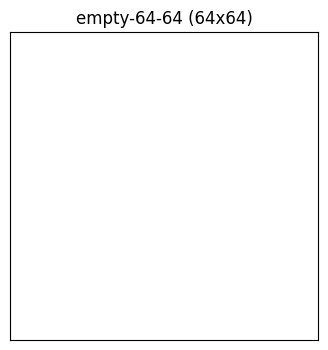}
        \caption{\emptyLargeLA}
    \end{subfigure}\hfill
    \begin{subfigure}[t]{0.24\textwidth}
        \centering
        \includegraphics[width=\linewidth]{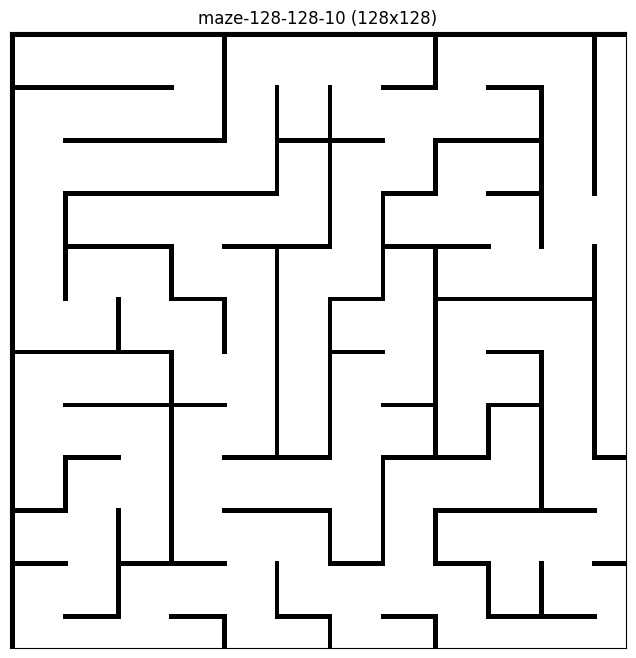}
        \caption{\mazeLargeLA}
    \end{subfigure}\hfill
    \begin{subfigure}[t]{0.24\textwidth}
        \centering
        \includegraphics[width=\linewidth]{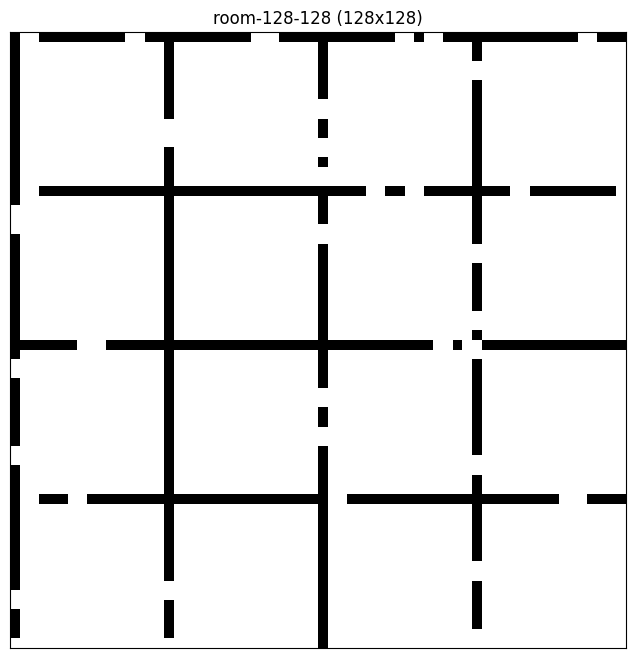}
        \caption{\roomLargeLA}
    \end{subfigure}\hfill
    \caption{Visualization of maps used in one-shot experiments with PMLA agents. In all maps, black tiles are obstacles and while tiles are empty spaces. Agents move between randomly selected empty spaces (white).}
    \label{fig:Large-mapf-maps}
\end{figure*}

\Cref{fig:mapf-agents} visualizes the agent models used in our experiments. \Cref{fig:mapf-maps} visualizes the maps used in one-shot and lifelong experiments with PM, RM, and DDR agents. For \paris, the original map is not connected. Therefore, we use the largest connected component in the original map to conduct the experiments. \Cref{fig:Large-mapf-maps} visualizes the maps used in one-shot experiments with PMLA agents. 
All maps are selected from the MAPF benchmark~\cite{Stern2019benchmark}.

\begin{table*}[!t]
    \centering
    \small
    \resizebox{\linewidth}{!}{
    \begin{tabular}{cccccccccc}
    \toprule
    MAPF & Map & $|V|$ & Agent Model & $w$ & $h$ & $R$ & $C$ & $m$ & $p$ \\
    \midrule
    \multirow{2}{*}{One-shot} & \paris       & 47,240 & PM & \multirow{2}{*}{$\{1,2,3,4\}$} & \multirow{2}{*}{$\{1,w\}$} & \multirow{2}{*}{100} & \multirow{2}{*}{$\{1,\infty\}$} & \multirow{2}{*}{$\{\text{PIBT},\text{EPIBT}\}$} & \multirow{2}{*}{LET}              \\

                              & \roomLarge         & 3,232 & PM                     \\
    \midrule
    \multirow{2}{*}{Lifelong} & \paris       & 47,240 &  PM, RM & \multirow{2}{4cm}{PM:$\{1,2,3\}$\ RM:$\{3,4,5\}$}  & \multirow{2}{*}{$\{1,w\}$} & \multirow{2}{3cm}{PM/RM:$\{10,100,1000\}$} & \multirow{2}{*}{$\{1,\infty\}$} & \multirow{2}{*}{$\{\text{PIBT},\text{EPIBT}\}$} & \multirow{2}{2cm}{PM/RM:\{LET,SD\}}   \\

                              & \roomLarge        & 3,232 & PM, RM                   \\
    \bottomrule
    \end{tabular}
    }
    \caption{
    Summary of MAPF model, maps, and agent models used in the additional experiments. $|V|$ is the number of vertices in the maps. For agent priority strategies ($p$), we consider Longer Elapsed Time (LTE) and Shorter Dist (SD).
    }
    \label{tab:exp-setup-add}
\end{table*}

\begin{figure*}[!tp]
    \centering
    \includegraphics[width=1\linewidth]{figs/major_result/oneshotmapf/small_agents/Universal_Algorithm_Legend.pdf}\par\medskip
    \begin{subfigure}{\textwidth}
        \centering
        \includegraphics[width=0.24\textwidth]{figs/major_result/oneshotmapf/small_agents/random-32-32-20_success.pdf}%
        \includegraphics[width=0.24\textwidth]{figs/major_result/oneshotmapf/small_agents/random-32-32-20_runtime.pdf}
        \includegraphics[width=0.24\textwidth]{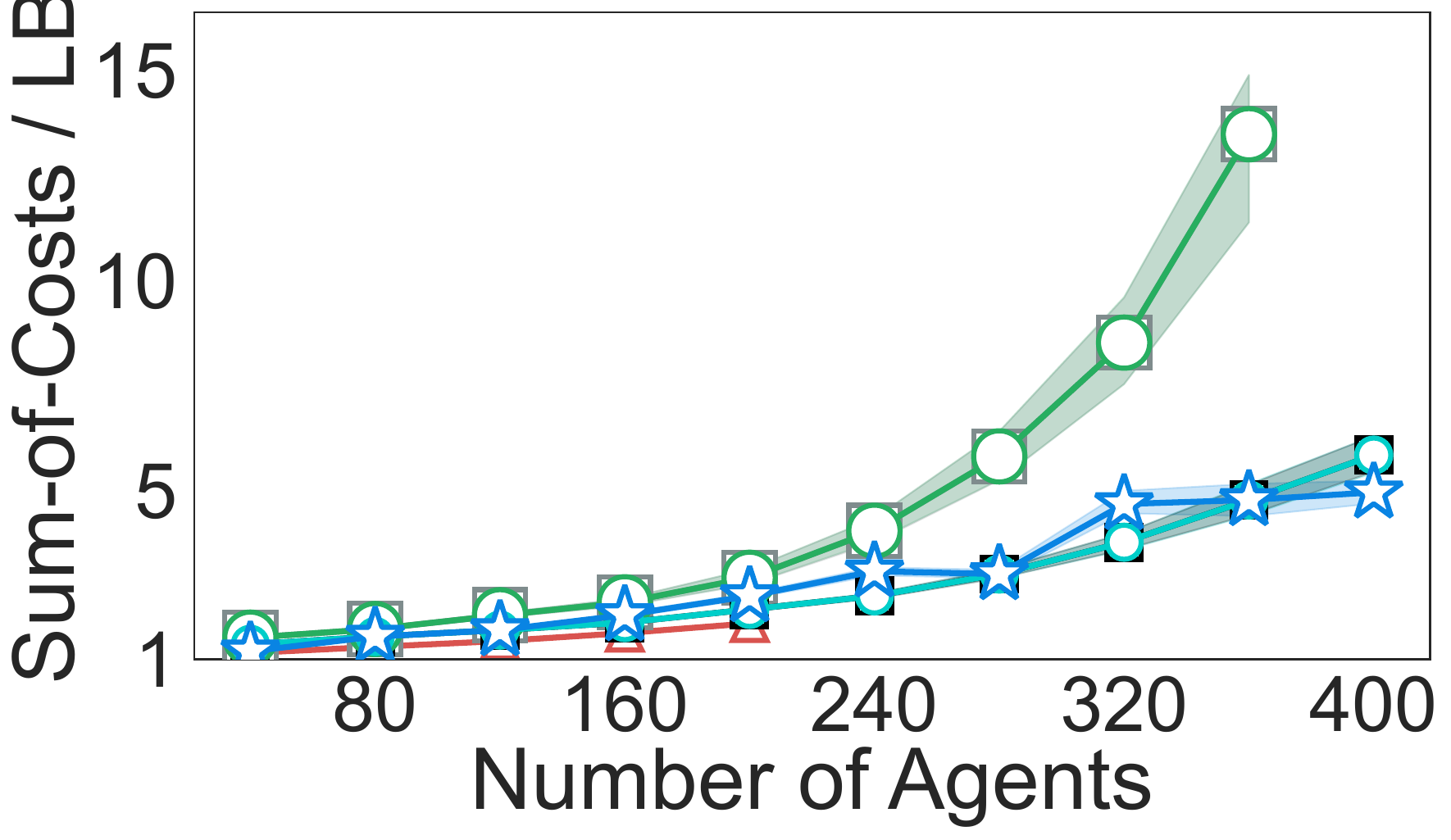}%
        \includegraphics[width=0.24\textwidth]{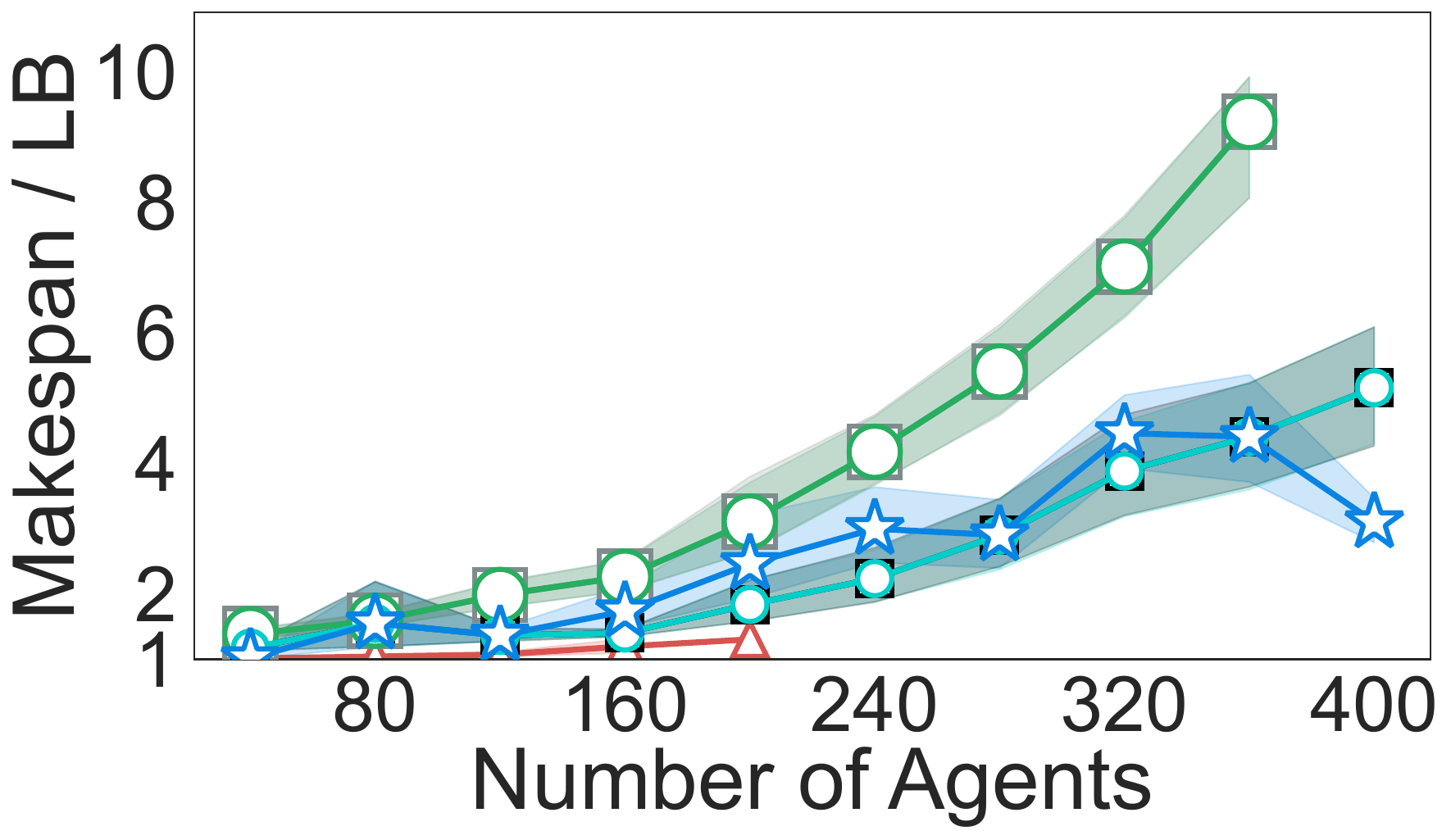}
        \caption{\randomSmall}
        \label{fig:oneshotmapf-small_agents-random_32_32_20-add}
    \end{subfigure}
    \begin{subfigure}{\textwidth}
        \centering
        \includegraphics[width=0.24\textwidth]{figs/major_result/oneshotmapf/small_agents/warehouse-10-20-10-2-1_success.pdf}%
        \includegraphics[width=0.24\textwidth]{figs/major_result/oneshotmapf/small_agents/warehouse-10-20-10-2-1_runtime.pdf}
        \includegraphics[width=0.24\textwidth]{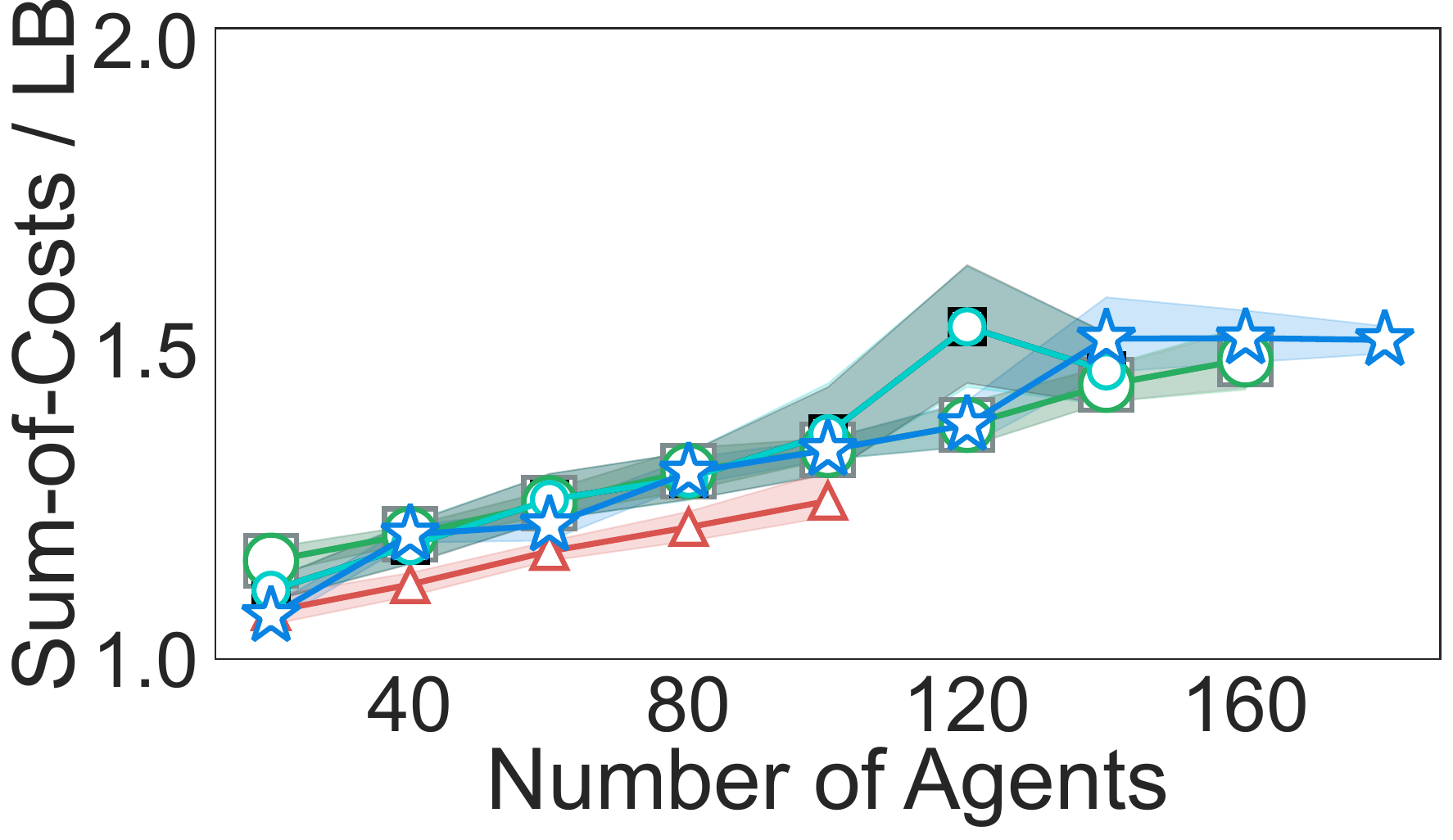}%
        \includegraphics[width=0.24\textwidth]{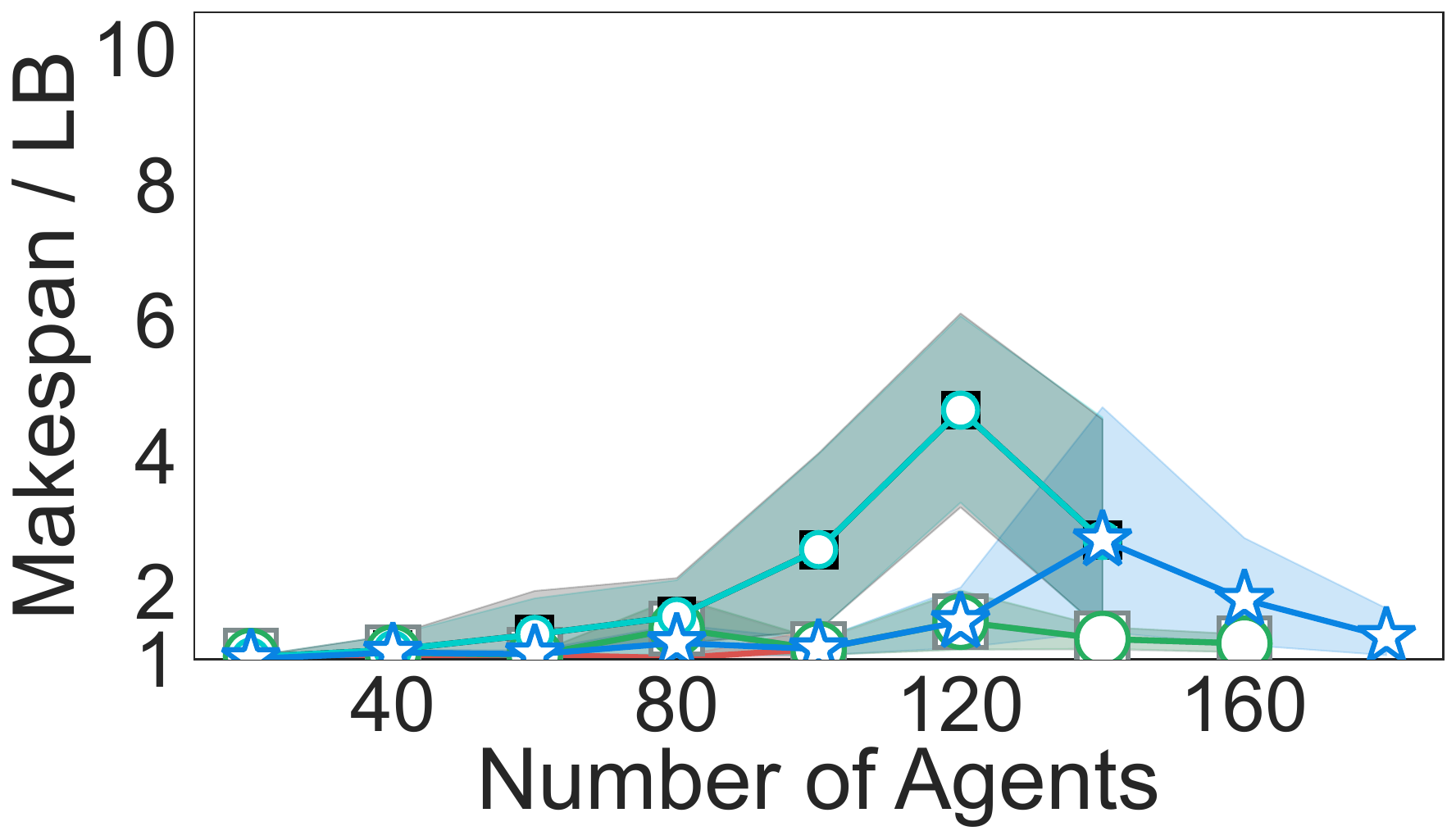}
        \caption{\warehouseXlarge}
        \label{fig:oneshotmapf-small_agents-warehouse_10_20_10_2_1-add}
    \end{subfigure}
    \begin{subfigure}{\textwidth}
        \centering
        \includegraphics[width=0.24\textwidth]{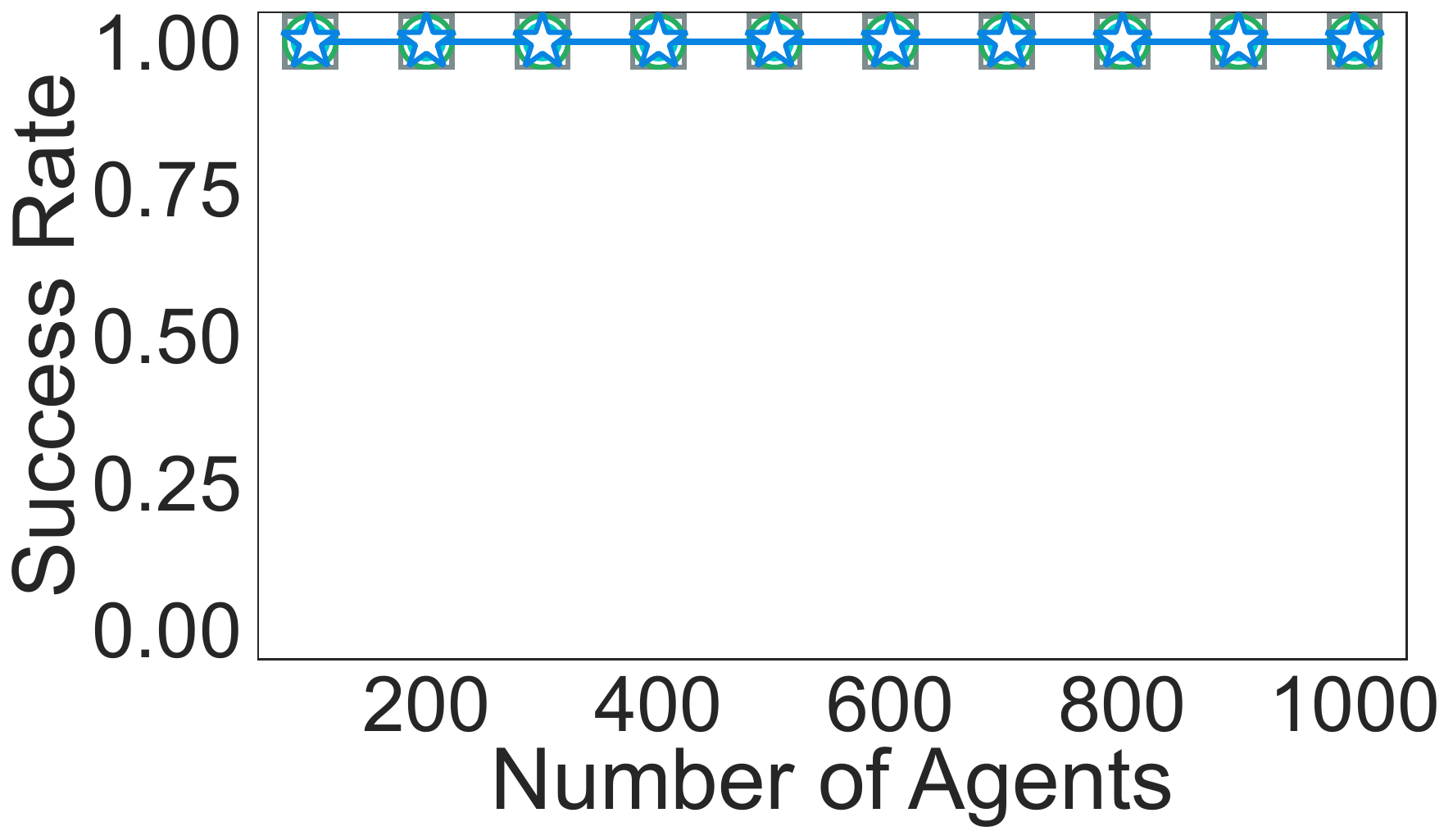}%
        \includegraphics[width=0.24\textwidth]{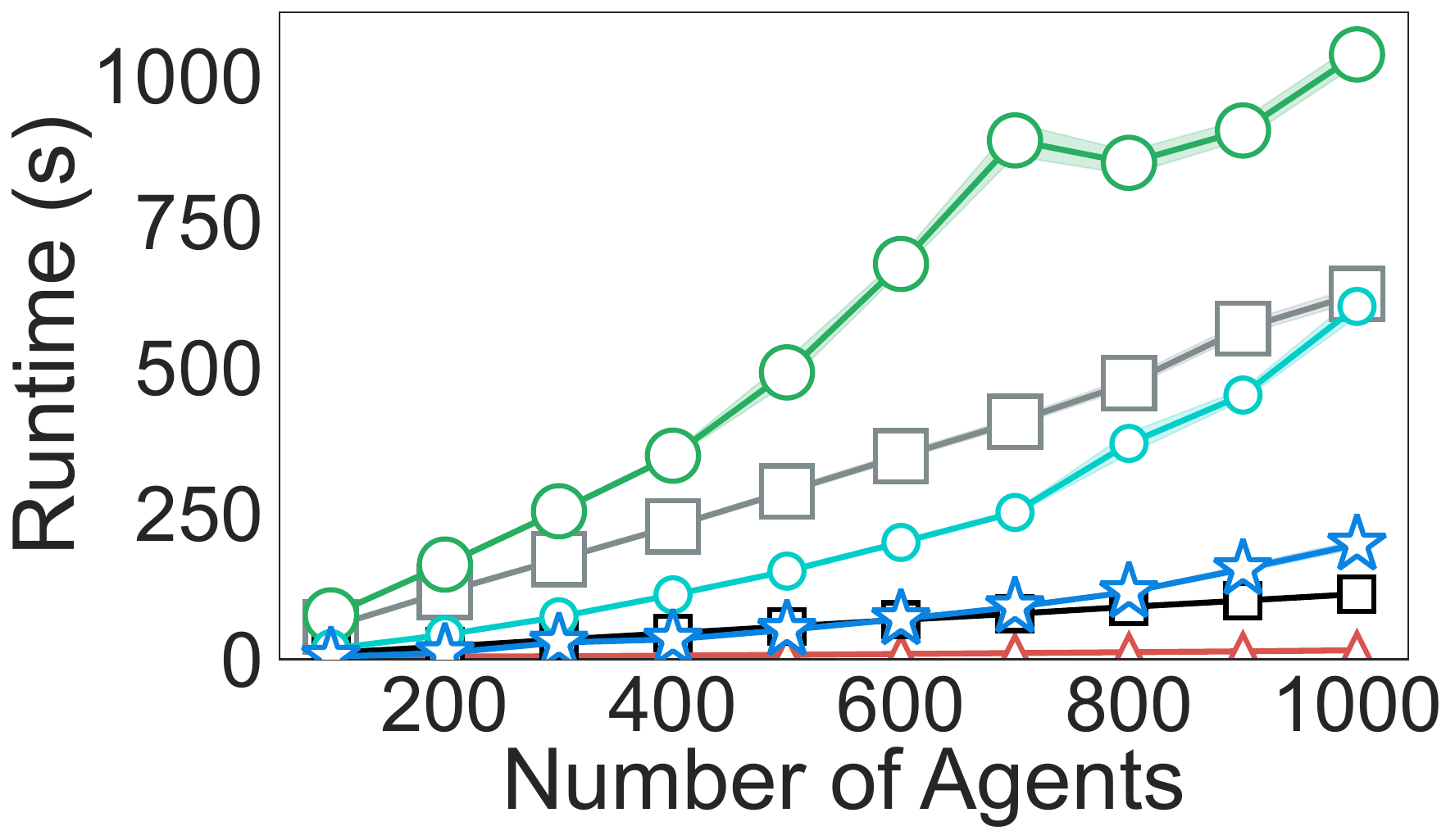}
        \includegraphics[width=0.24\textwidth]{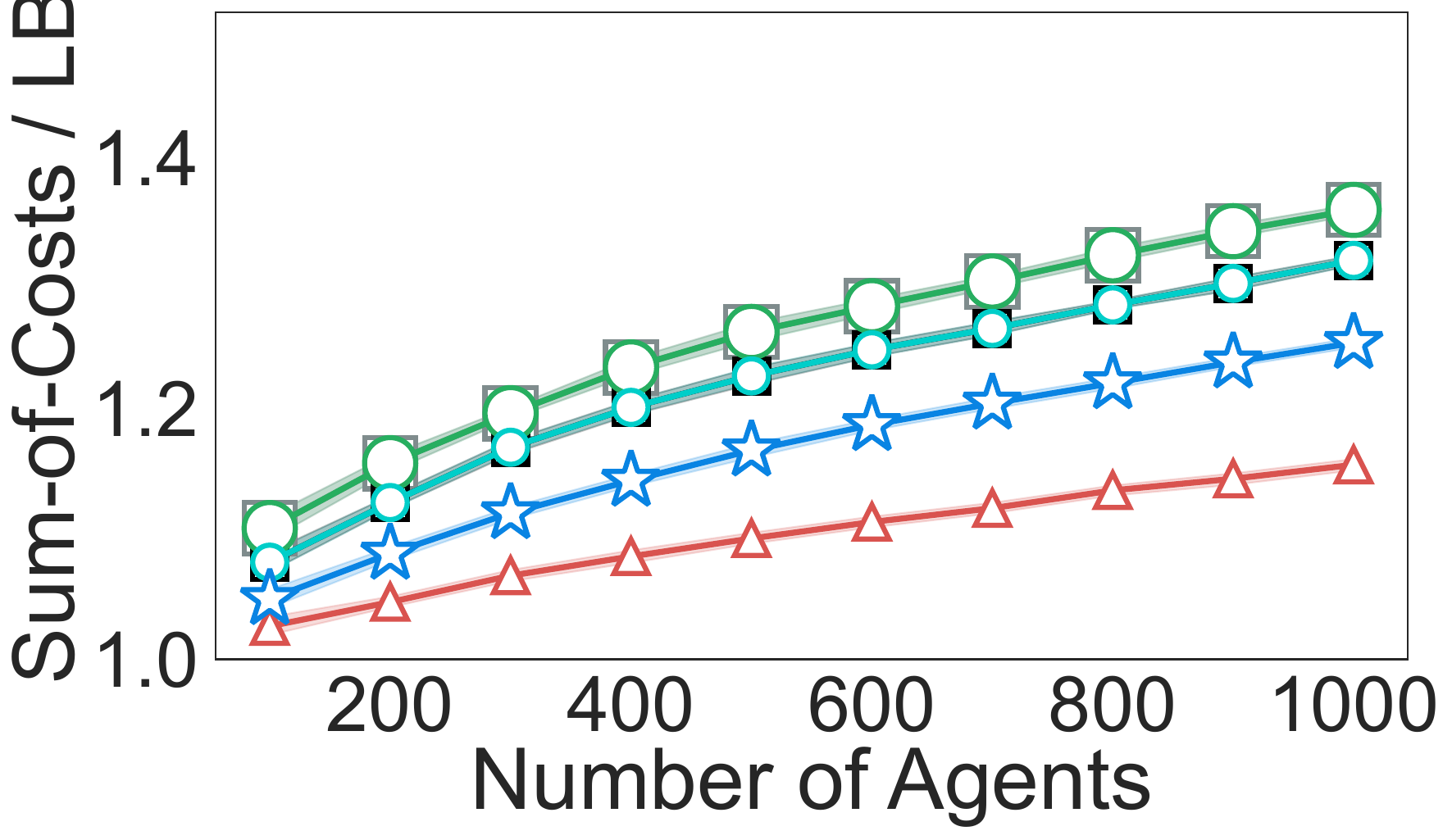}%
        \includegraphics[width=0.24\textwidth]{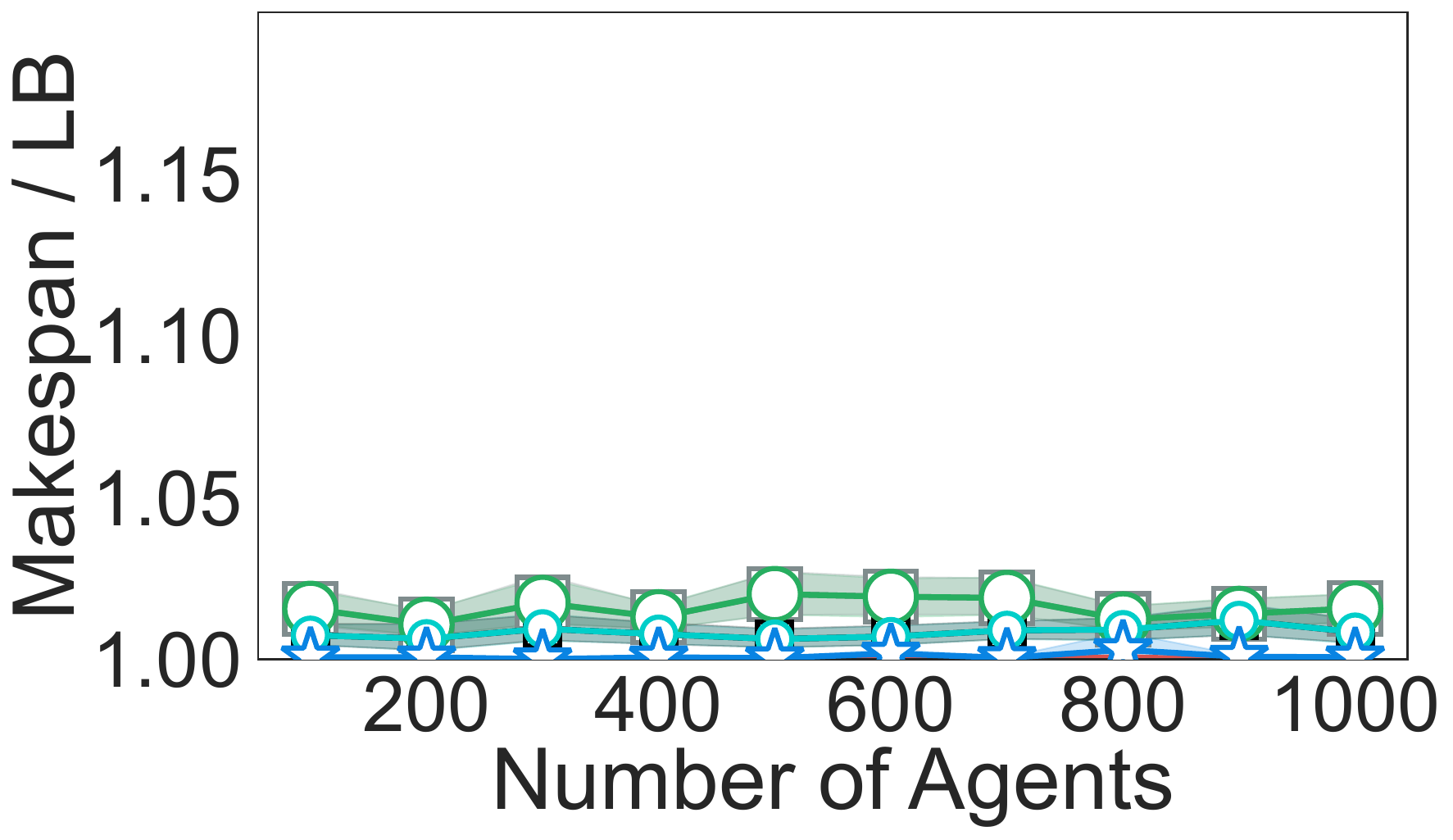}
        \caption{\paris}
        \label{fig:oneshotmapf-small_agents-Paris_1_256-add}
    \end{subfigure}
    \begin{subfigure}{\textwidth}
        \centering
        \includegraphics[width=0.24\textwidth]{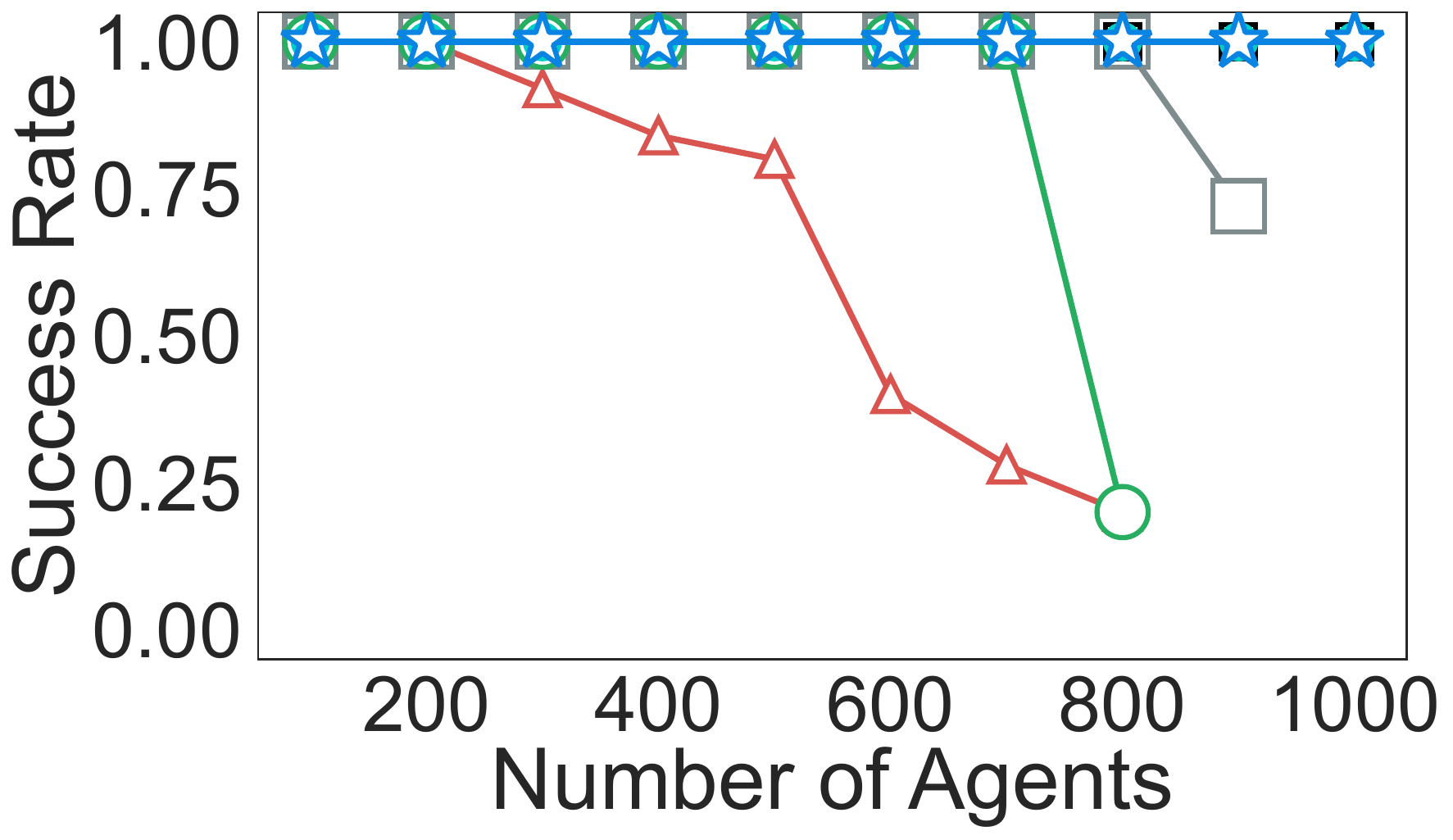}%
        \includegraphics[width=0.24\textwidth]{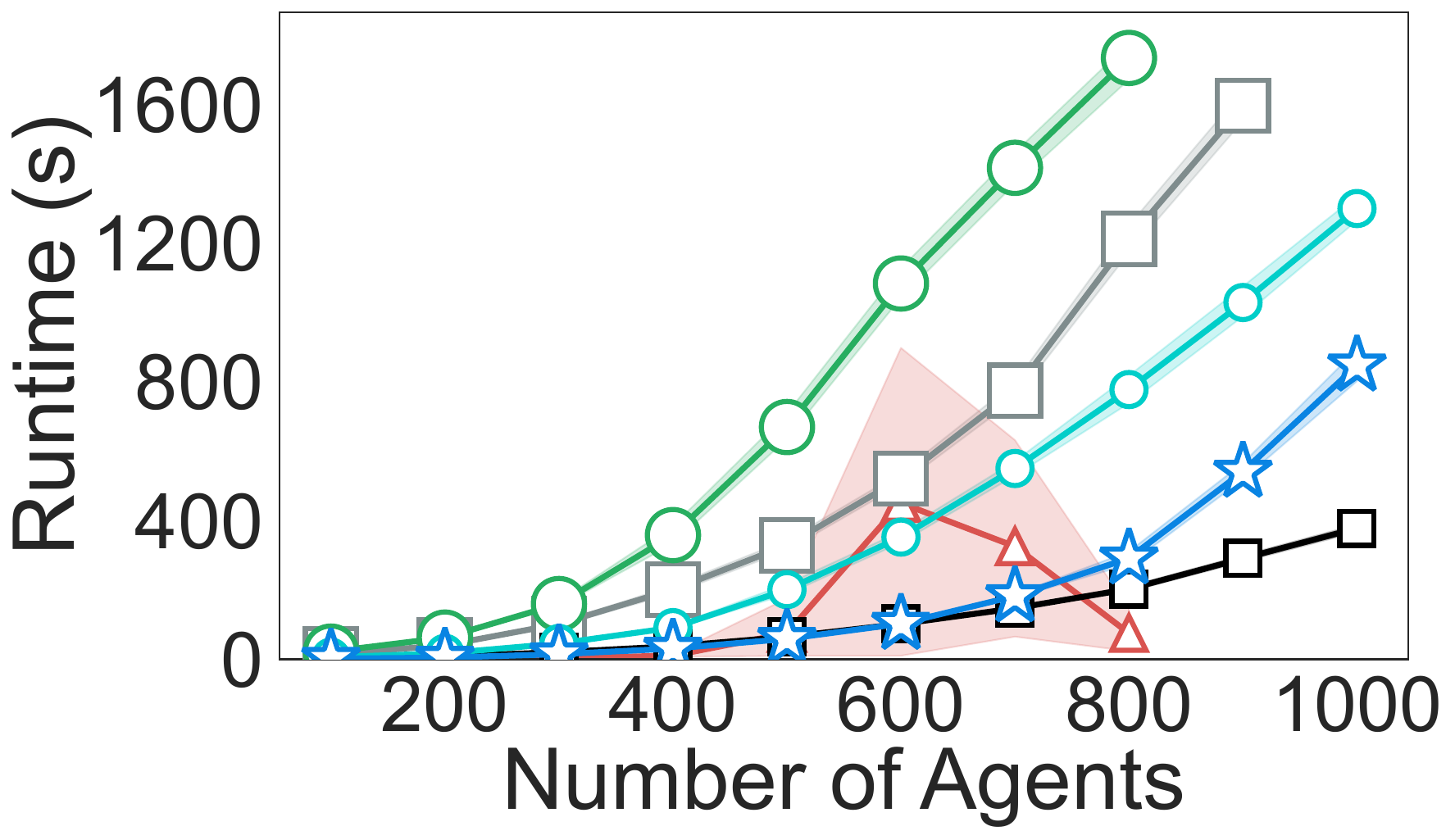}
        \includegraphics[width=0.24\textwidth]{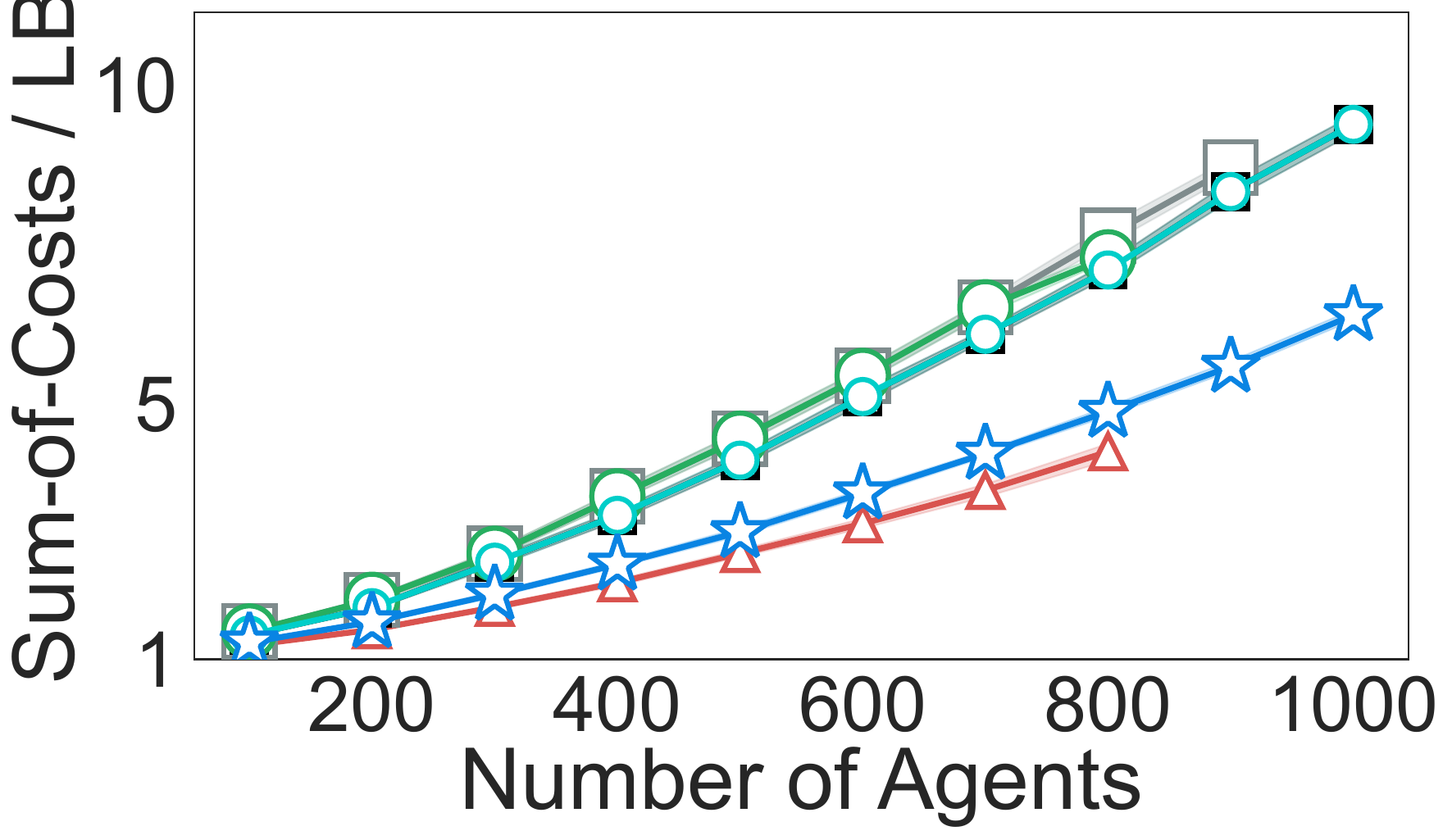}%
        \includegraphics[width=0.24\textwidth]{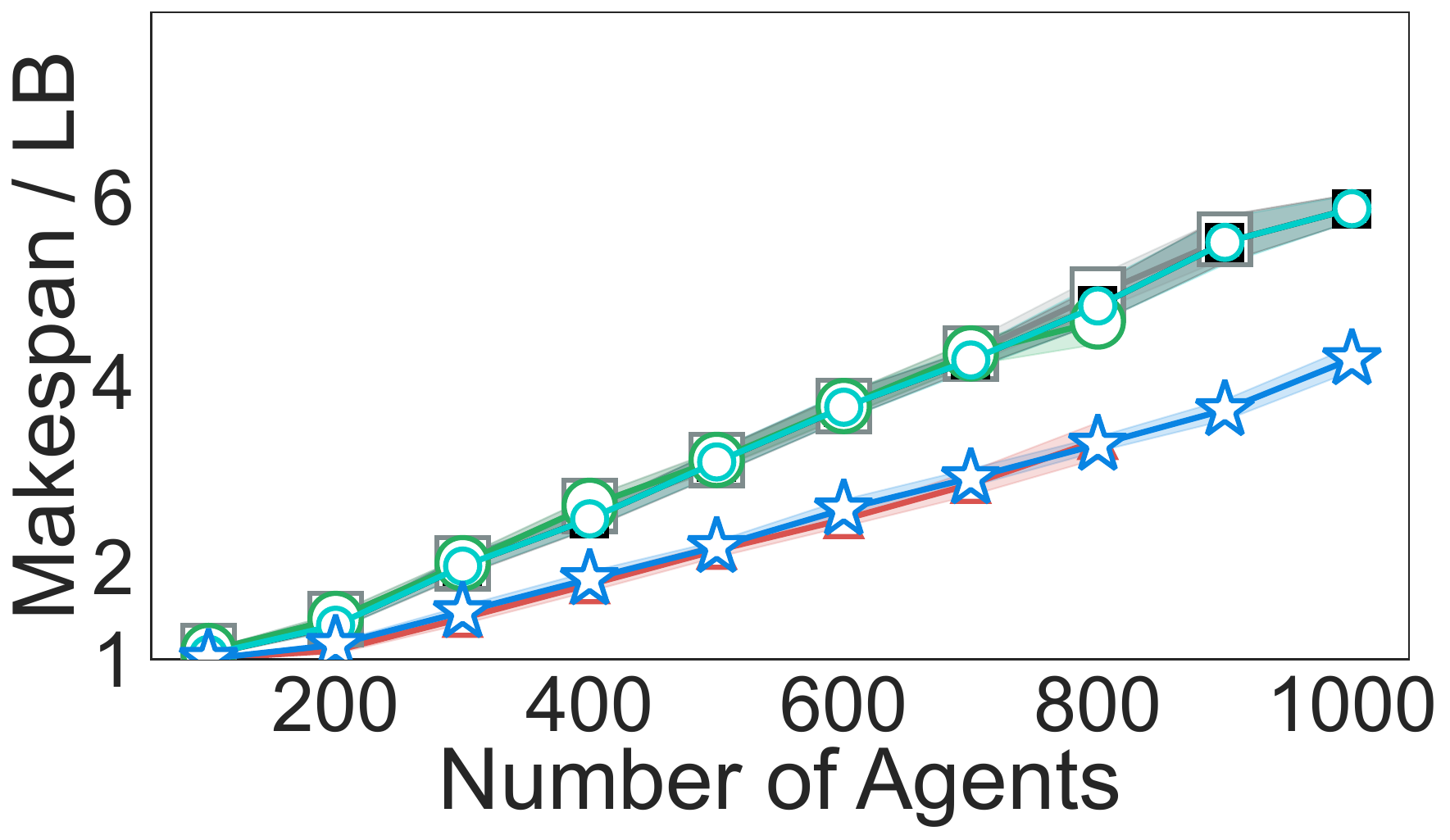}
        \caption{\roomLarge}
        \label{fig:oneshotmapf-small_agents-room_64_64_8-add}
    \end{subfigure}
    \hfill

    \caption{Results of one-shot MAPF with PM agents.}
    \label{fig:major-result-oneshotmapf-smallagents-add}
    \par\vspace{-\abovecaptionskip}
\end{figure*}

\begin{figure*}[!t]
    \centering
    \includegraphics[width=1\linewidth]{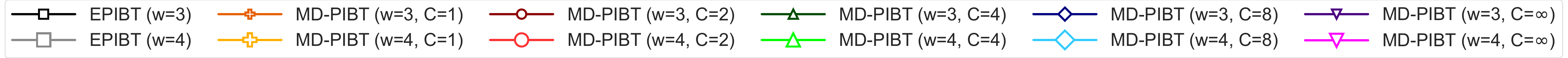}\par\medskip
    \begin{subfigure}[t]{0.24\textwidth}
        \centering
        \includegraphics[width=\linewidth]{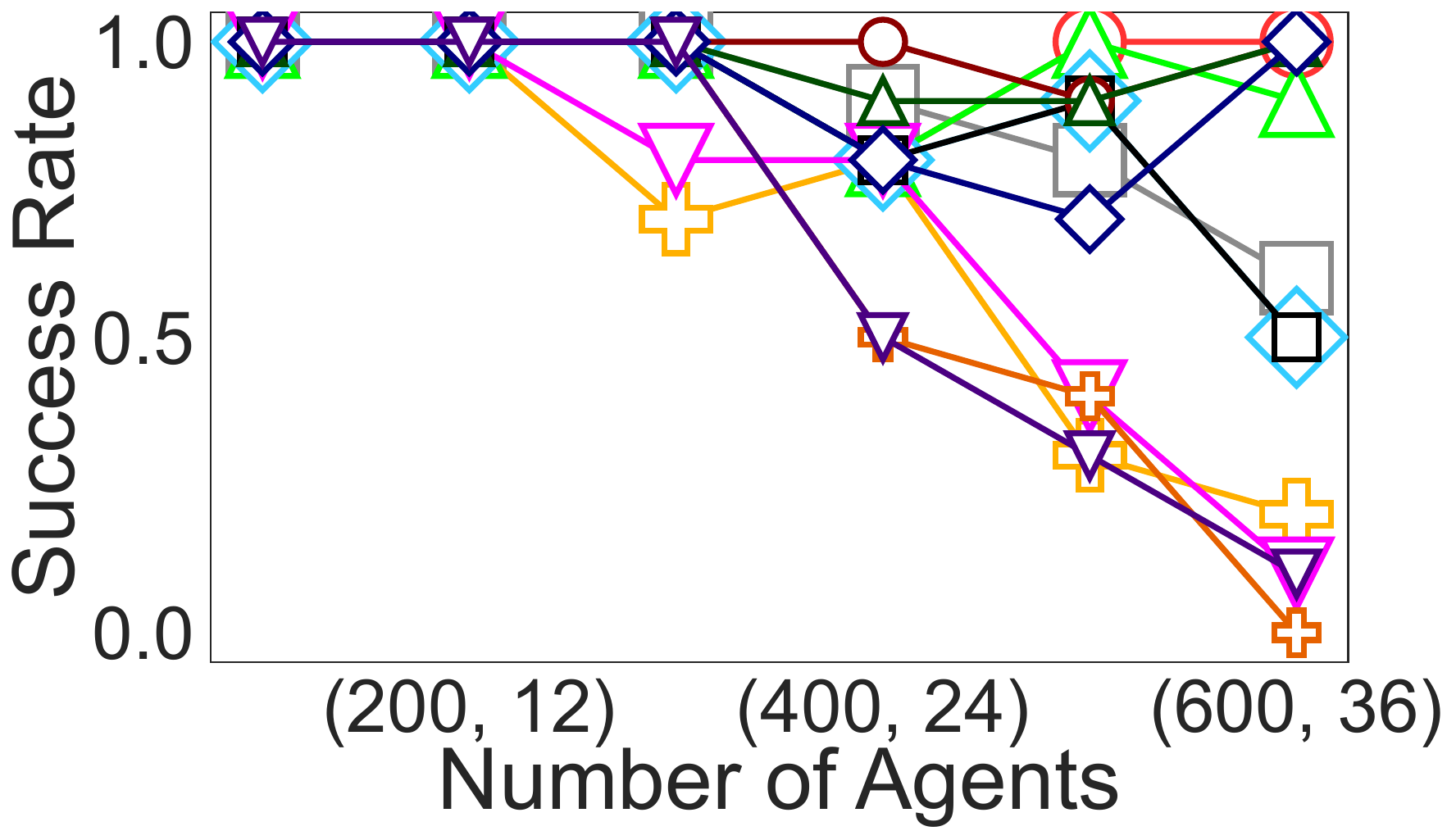}
    \end{subfigure}\hfill
    \begin{subfigure}[t]{0.24\textwidth}
        \centering
        \includegraphics[width=\linewidth]{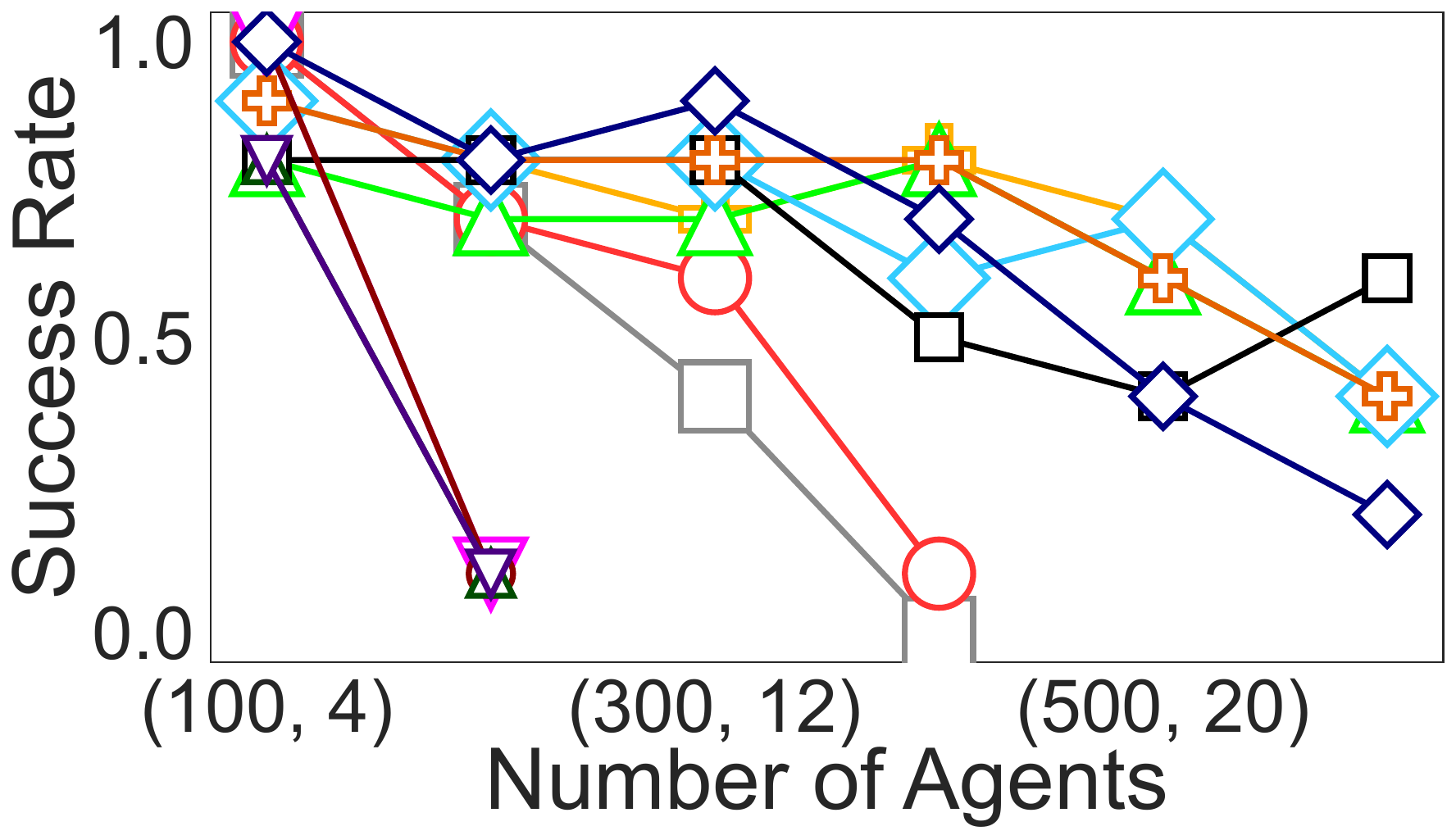}
    \end{subfigure}\hfill
    \begin{subfigure}[t]{0.24\textwidth}
        \centering
        \includegraphics[width=\linewidth]{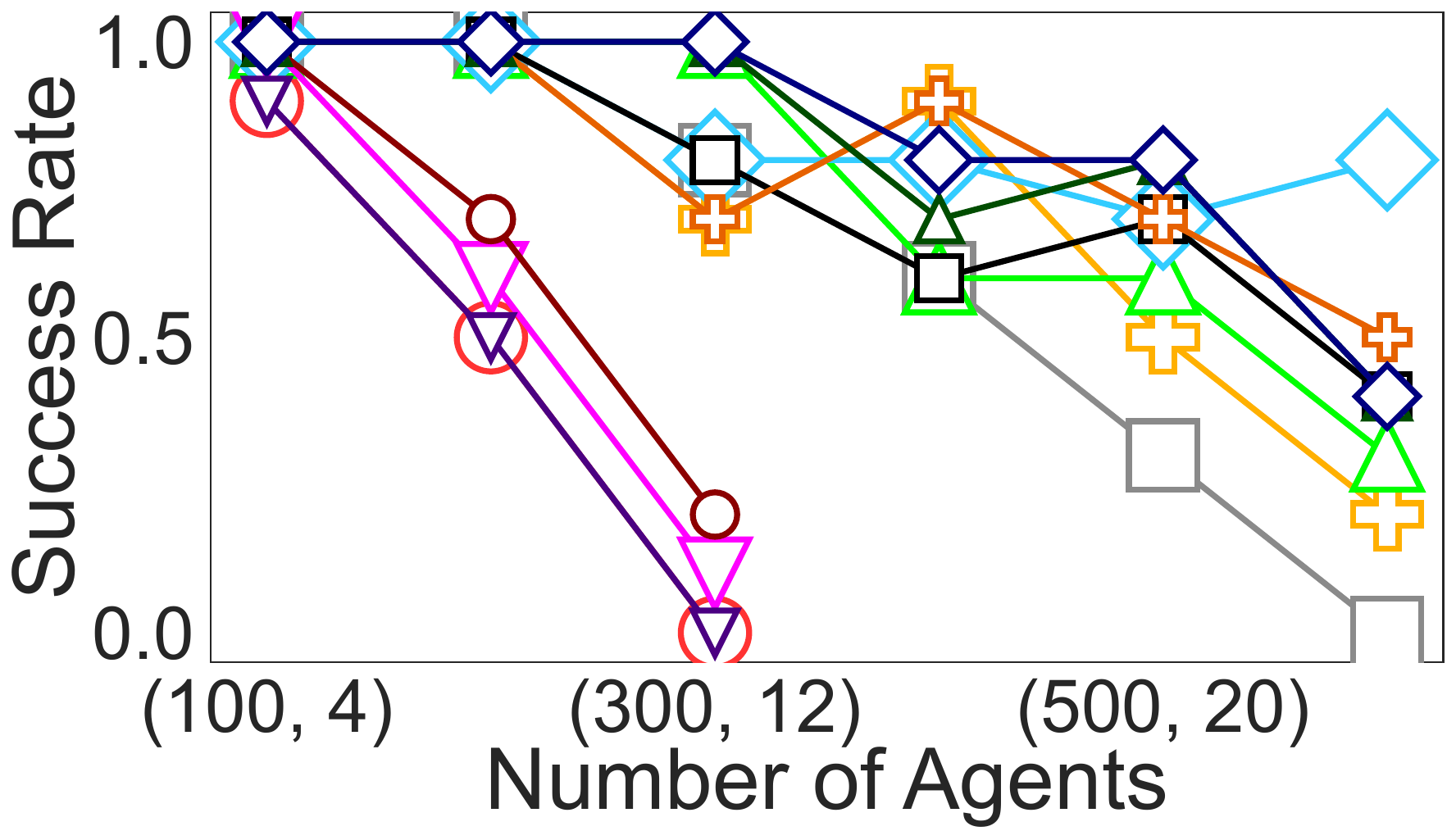}
    \end{subfigure}\hfill
    \begin{subfigure}[t]{0.24\textwidth}
        \centering
        \includegraphics[width=\linewidth]{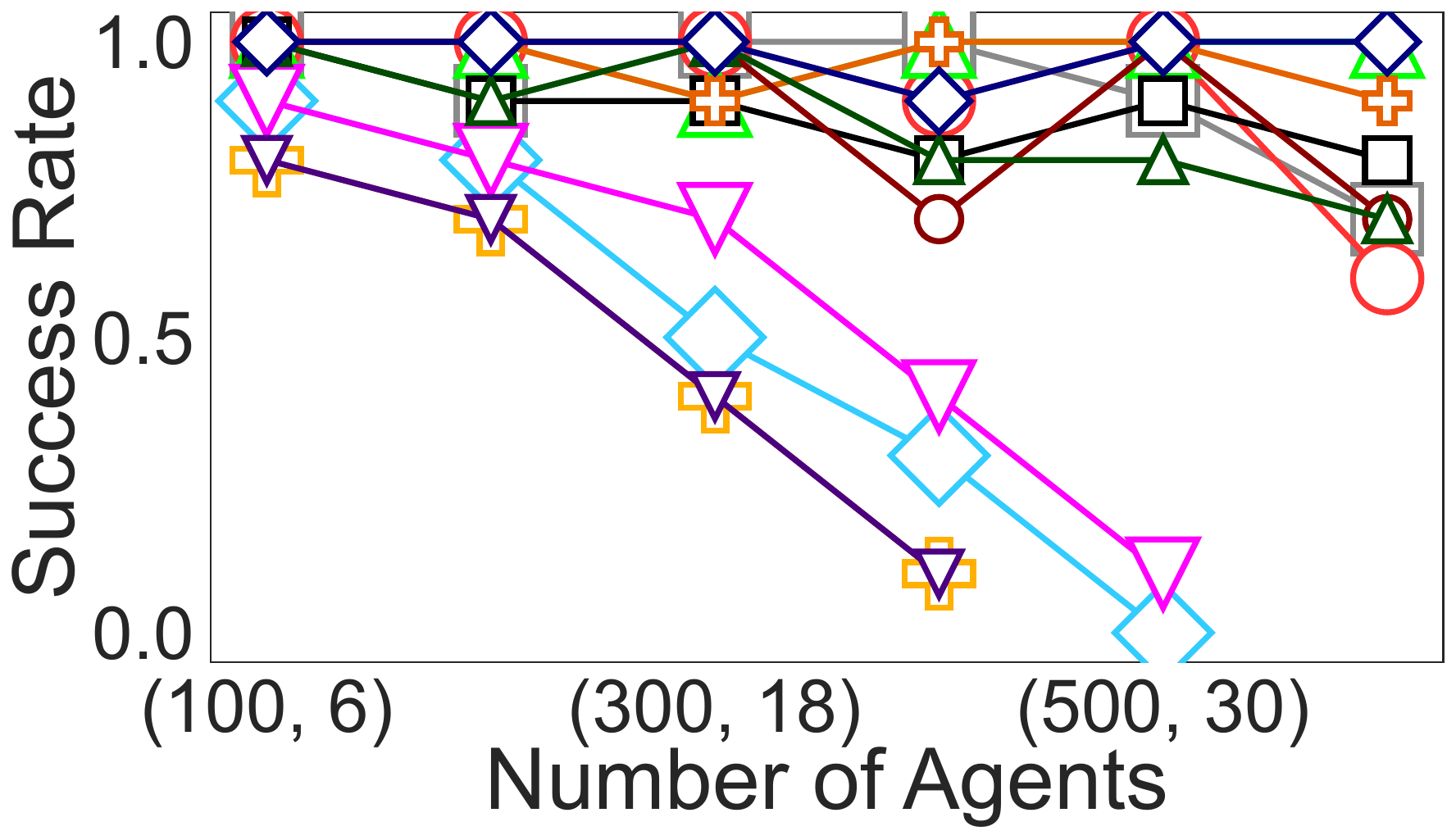}
    \end{subfigure}

    \begin{subfigure}[t]{0.24\textwidth}
        \centering
        \includegraphics[width=\linewidth]{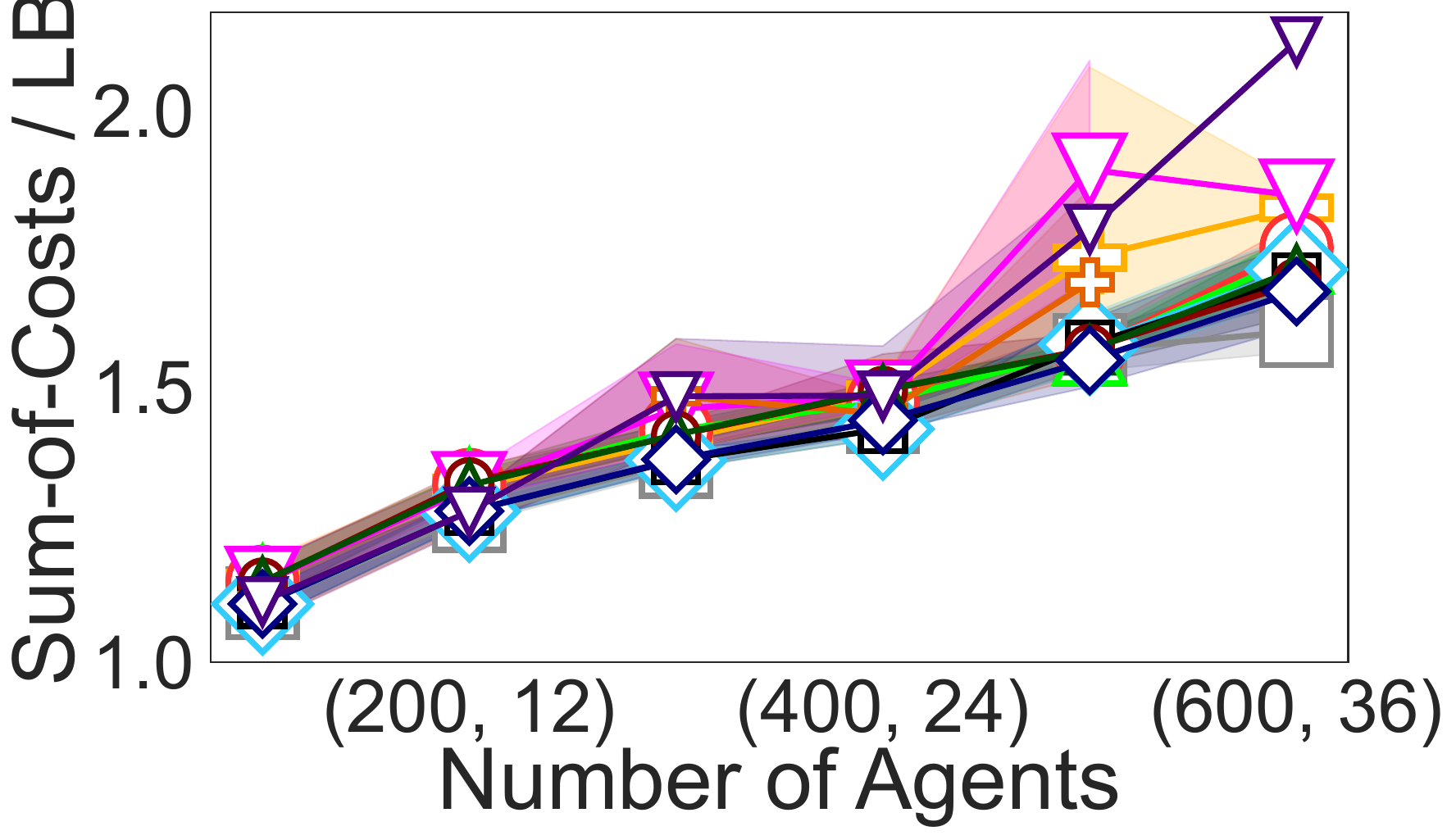}
    \end{subfigure}\hfill
    \begin{subfigure}[t]{0.24\textwidth}
        \centering
        \includegraphics[width=\linewidth]{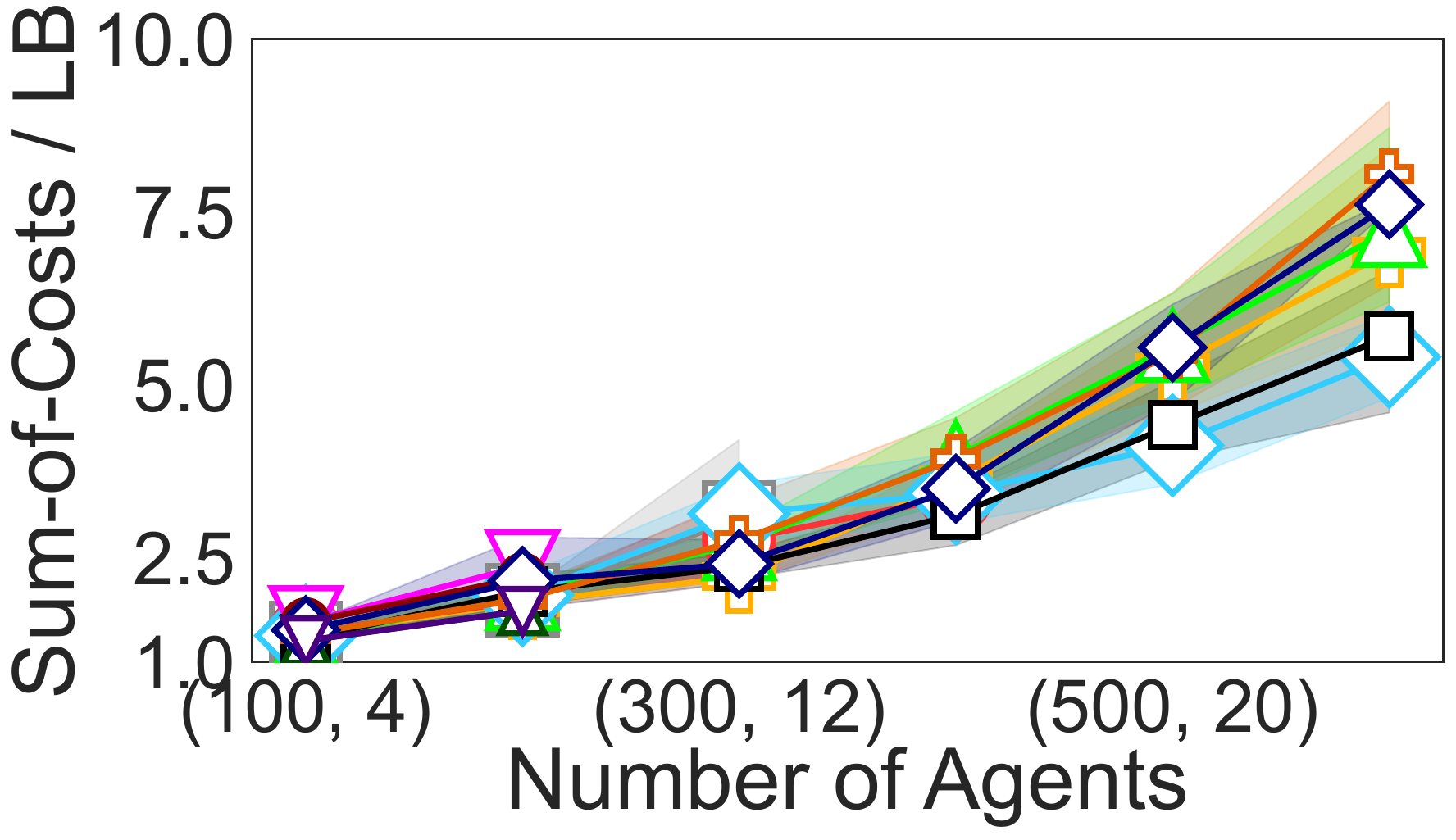}
    \end{subfigure}\hfill
    \begin{subfigure}[t]{0.24\textwidth}
        \centering
        \includegraphics[width=\linewidth]{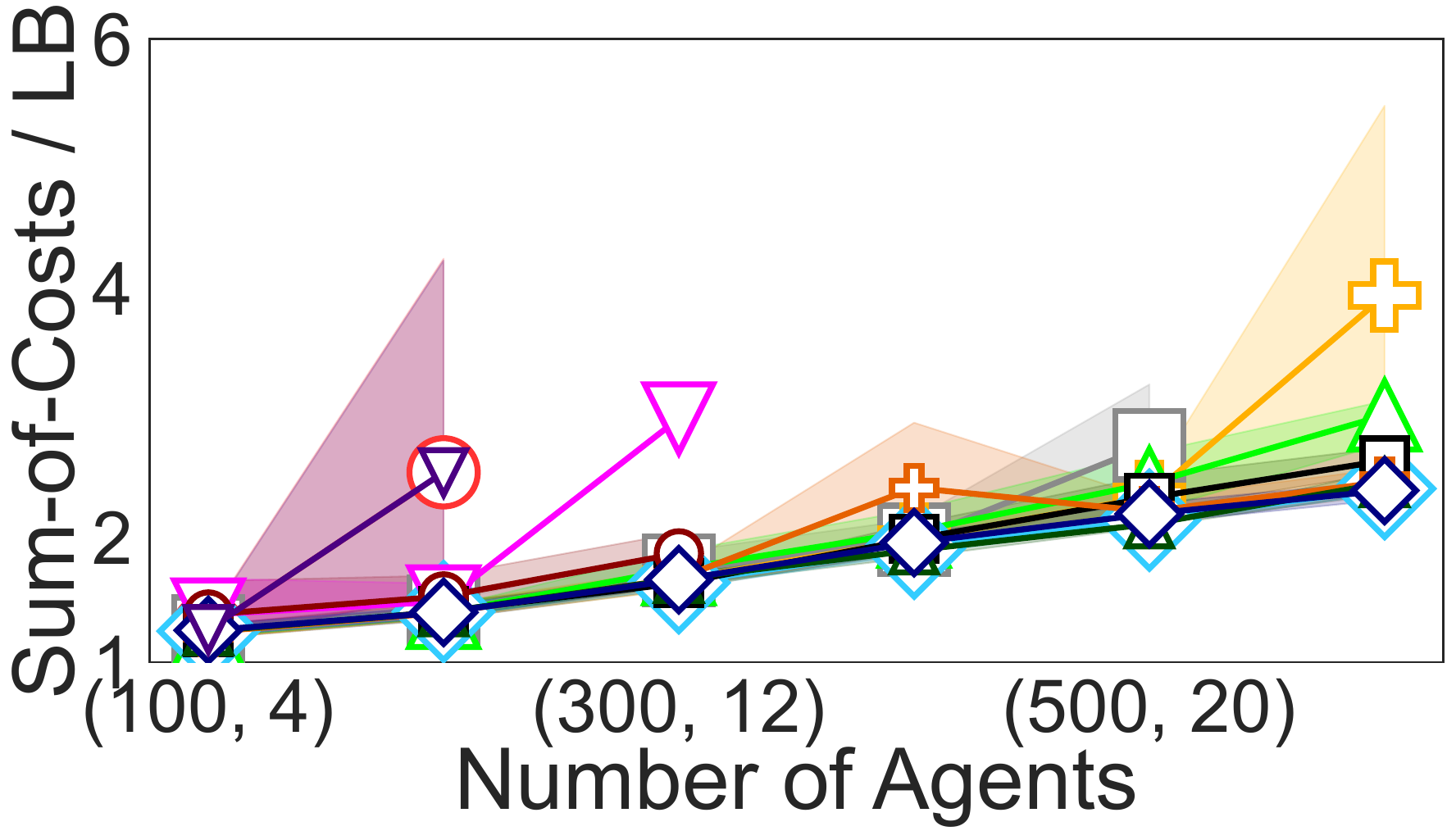}
    \end{subfigure}\hfill
    \begin{subfigure}[t]{0.24\textwidth}
        \centering
        \includegraphics[width=\linewidth]{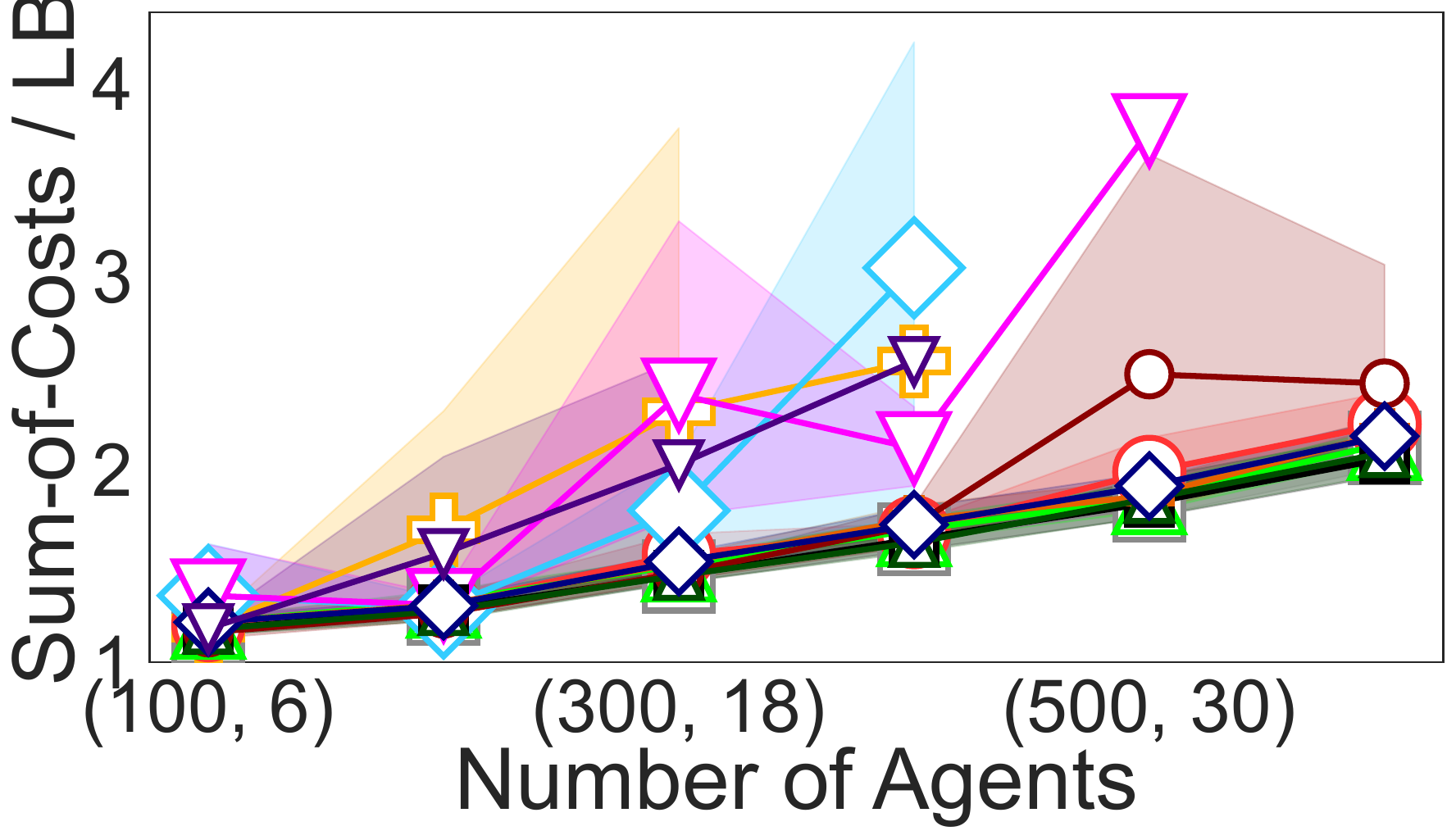}
    \end{subfigure}
    
    \begin{subfigure}[t]{0.24\textwidth}
        \centering
        \includegraphics[width=\linewidth]{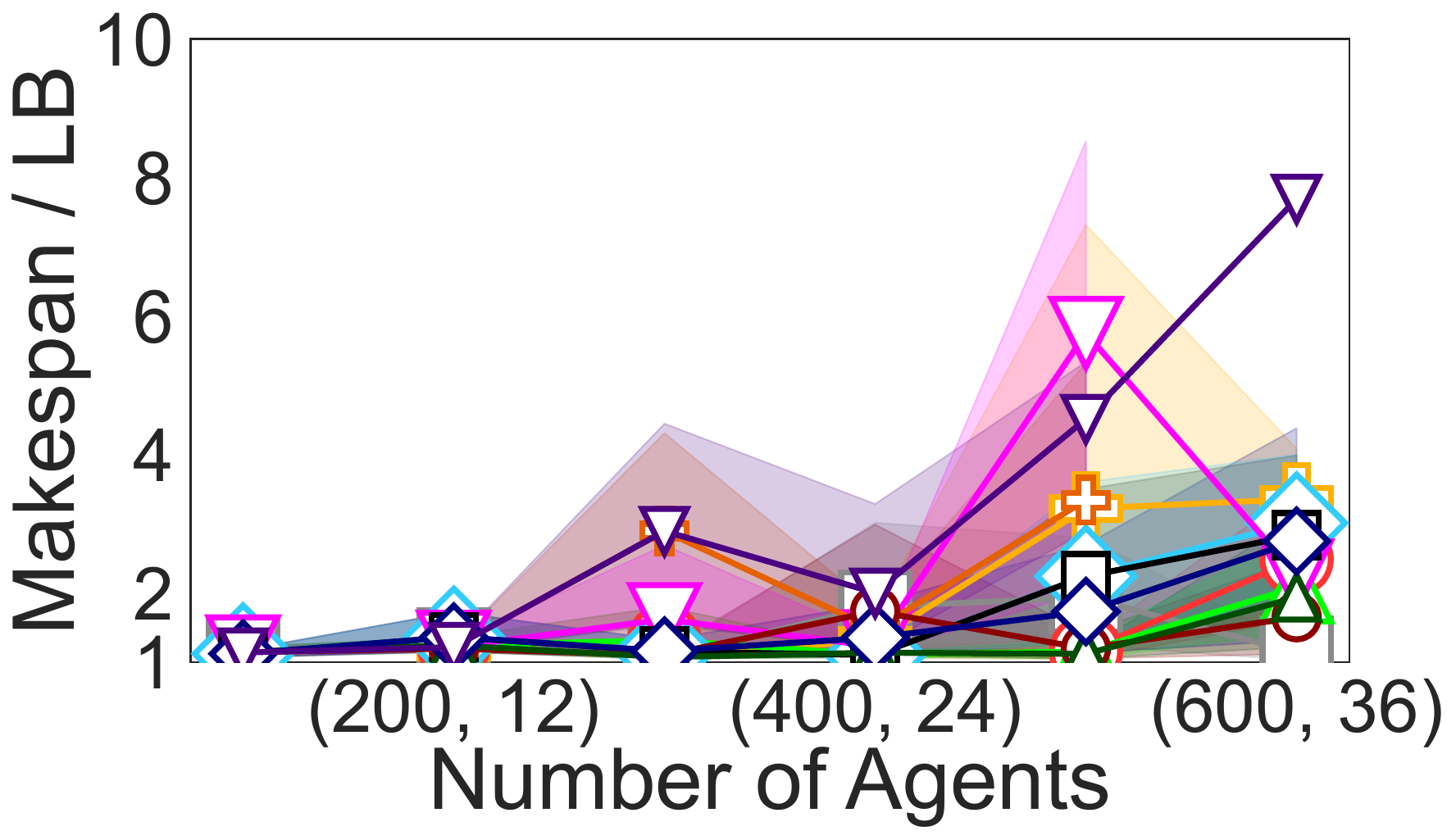}
    \end{subfigure}\hfill
    \begin{subfigure}[t]{0.24\textwidth}
        \centering
        \includegraphics[width=\linewidth]{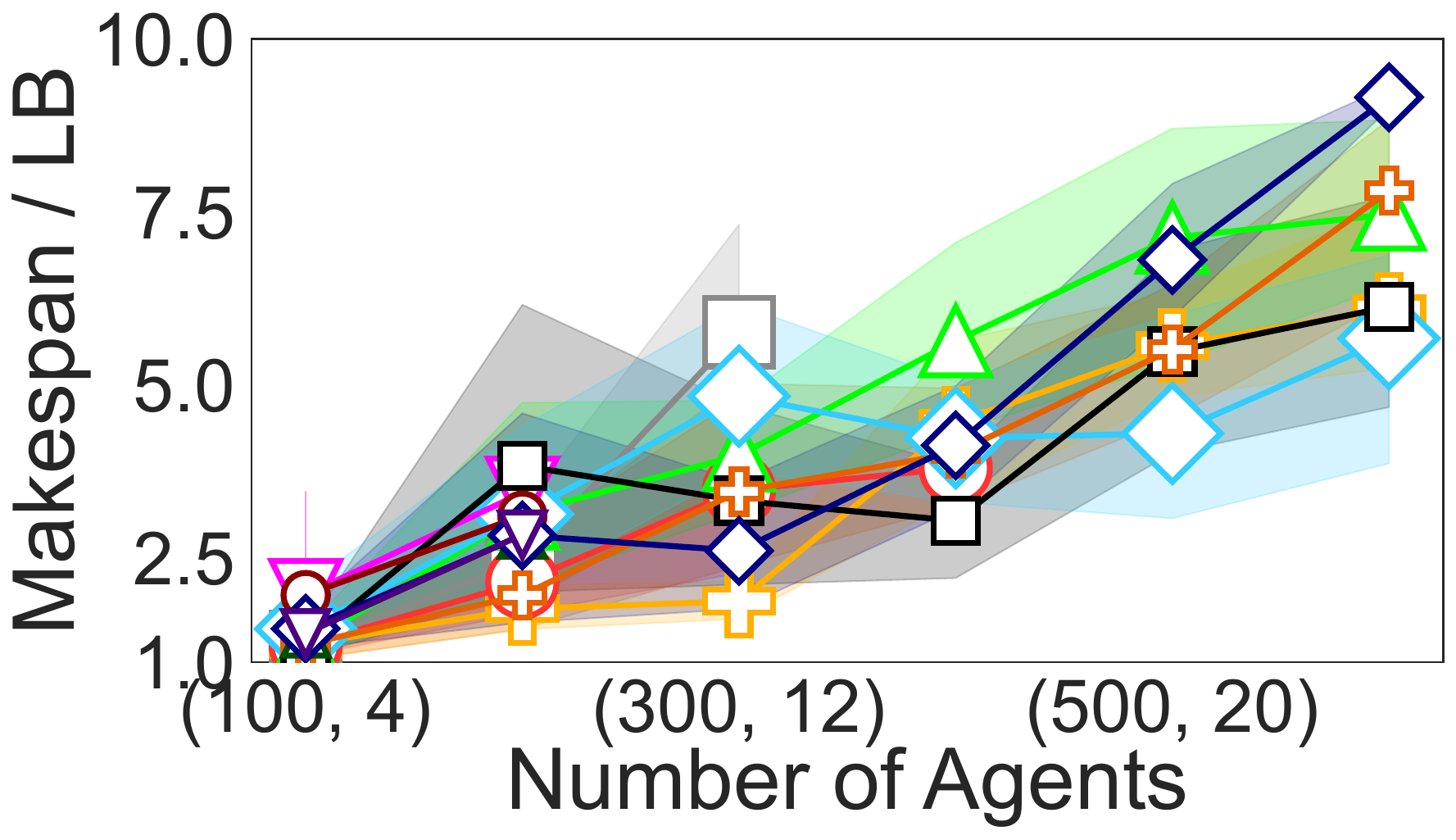}
    \end{subfigure}\hfill
    \begin{subfigure}[t]{0.24\textwidth}
        \centering
        \includegraphics[width=\linewidth]{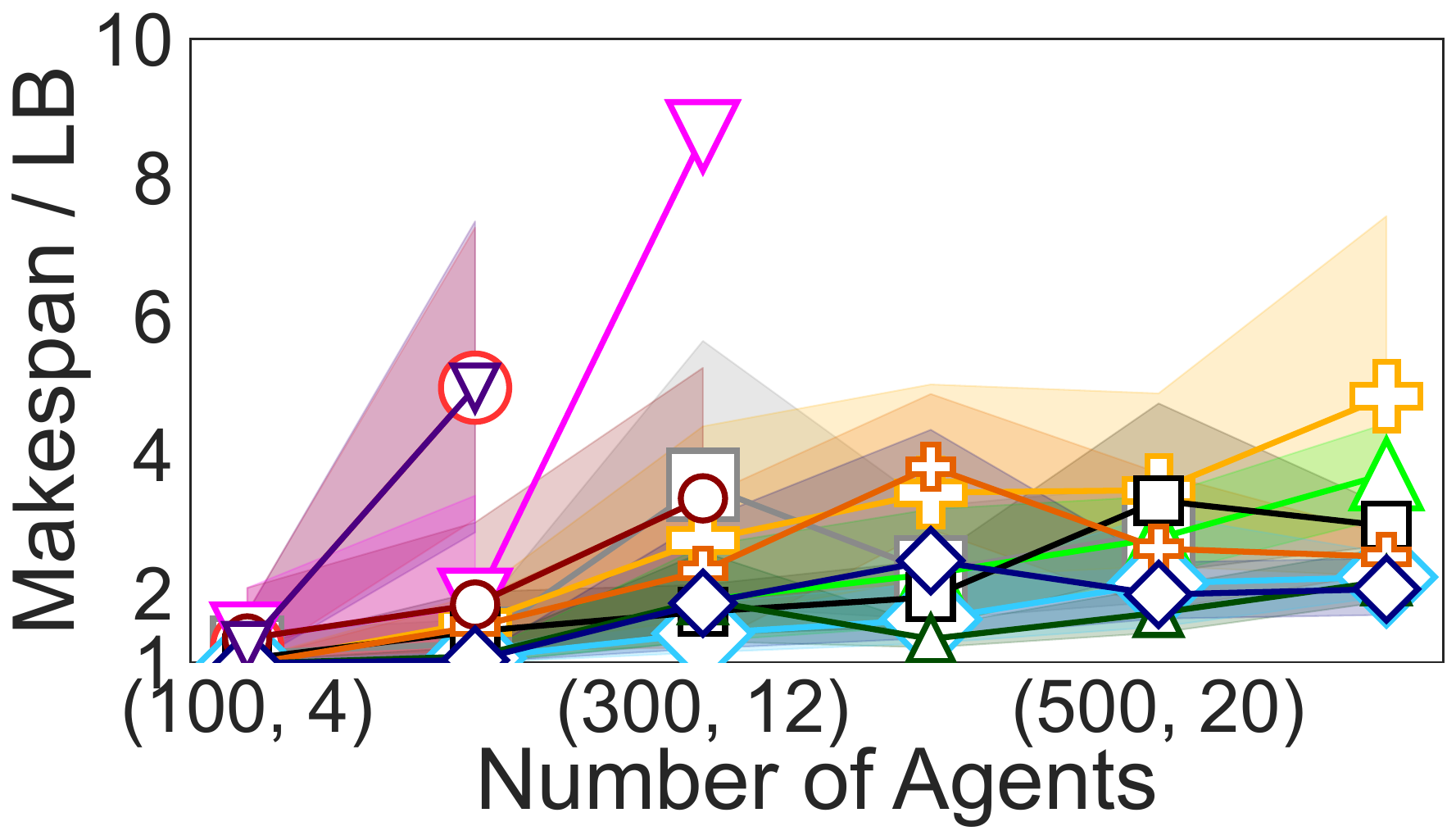}
    \end{subfigure}\hfill
    \begin{subfigure}[t]{0.24\textwidth}
        \centering
        \includegraphics[width=\linewidth]{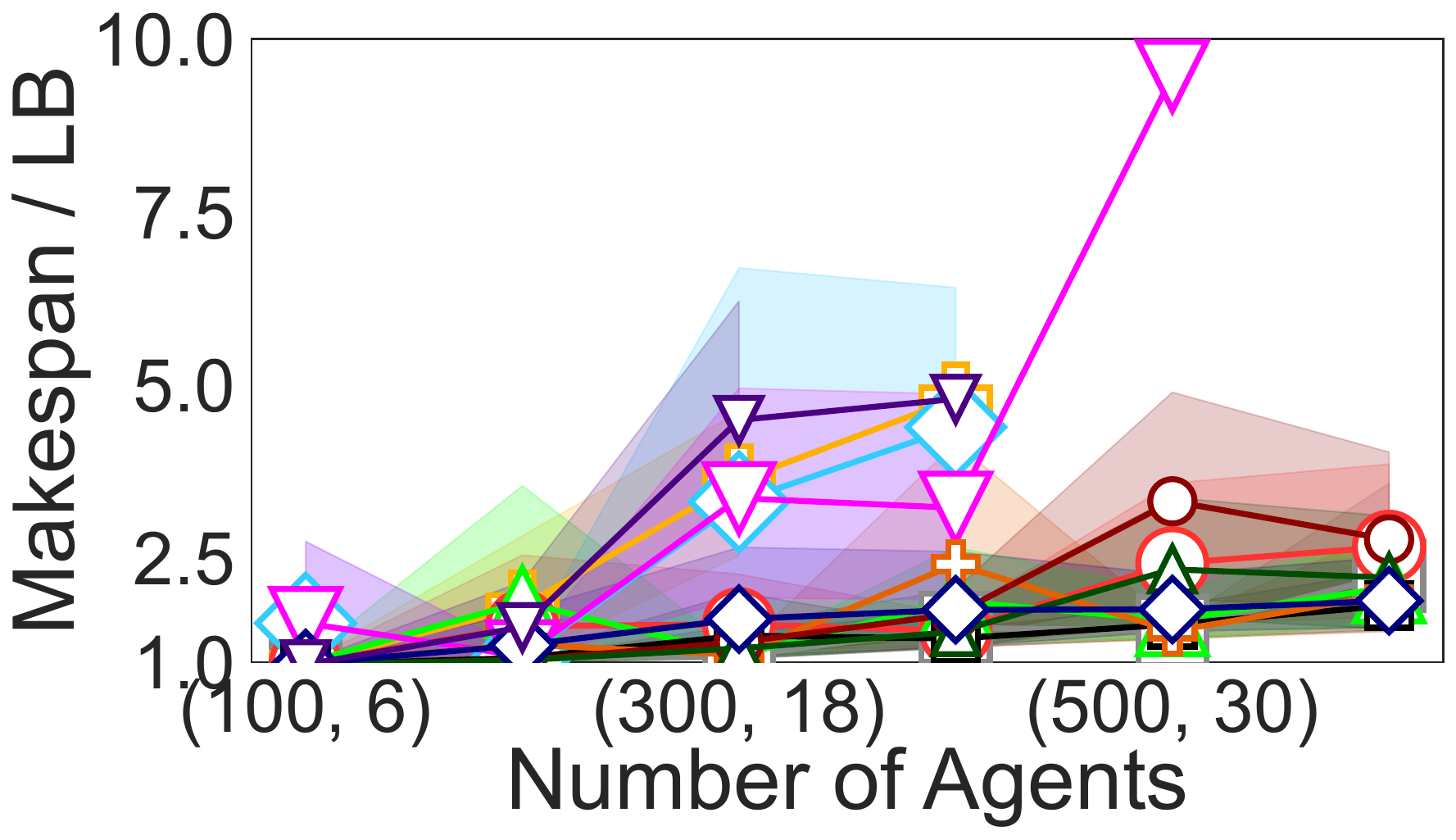}
    \end{subfigure}

    \begin{subfigure}[t]{0.24\textwidth}
        \centering
        \includegraphics[width=\linewidth]{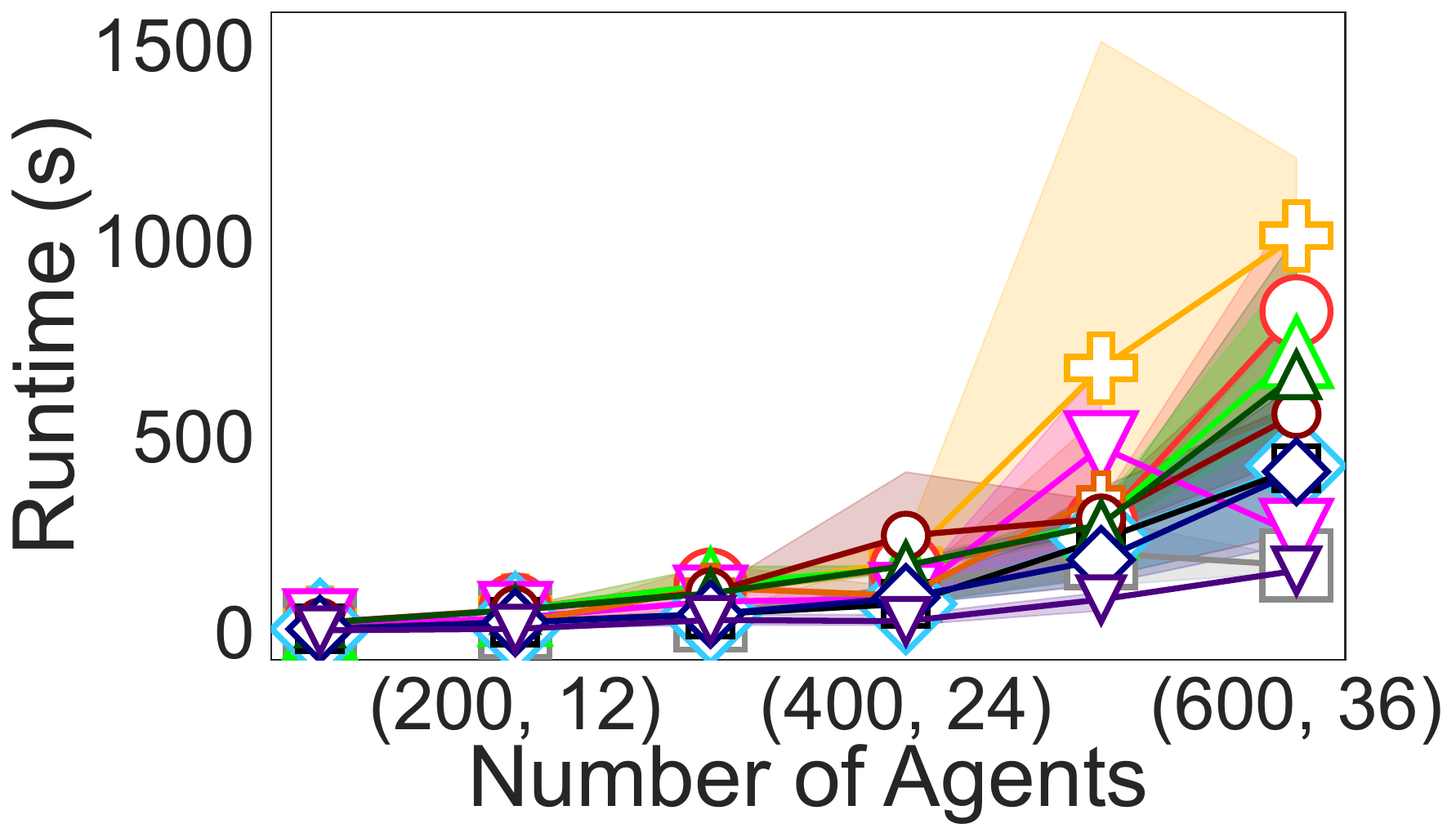}
    \end{subfigure}\hfill
    \begin{subfigure}[t]{0.24\textwidth}
        \centering
        \includegraphics[width=\linewidth]{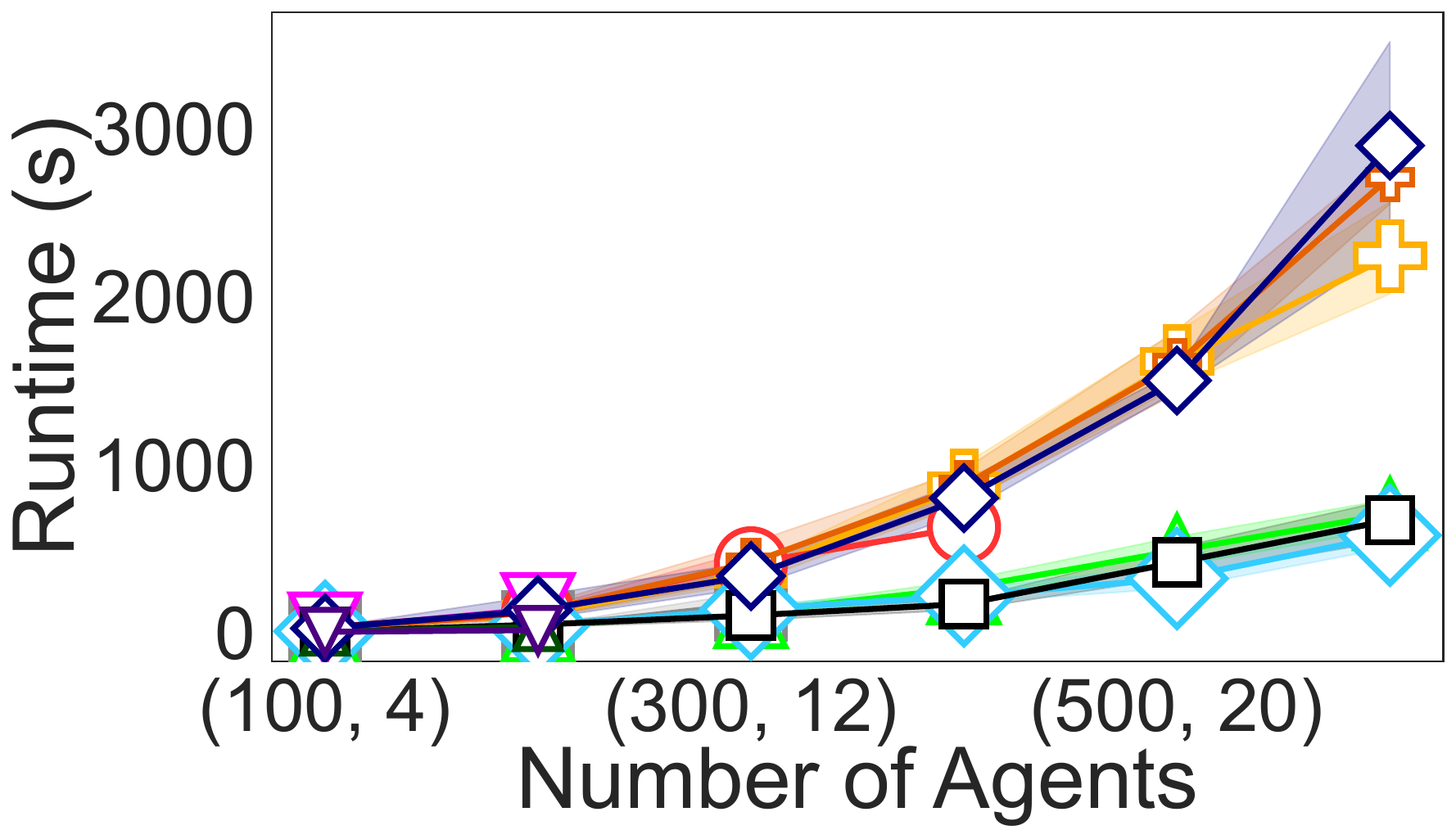}
    \end{subfigure}\hfill
    \begin{subfigure}[t]{0.24\textwidth}
        \centering
        \includegraphics[width=\linewidth]{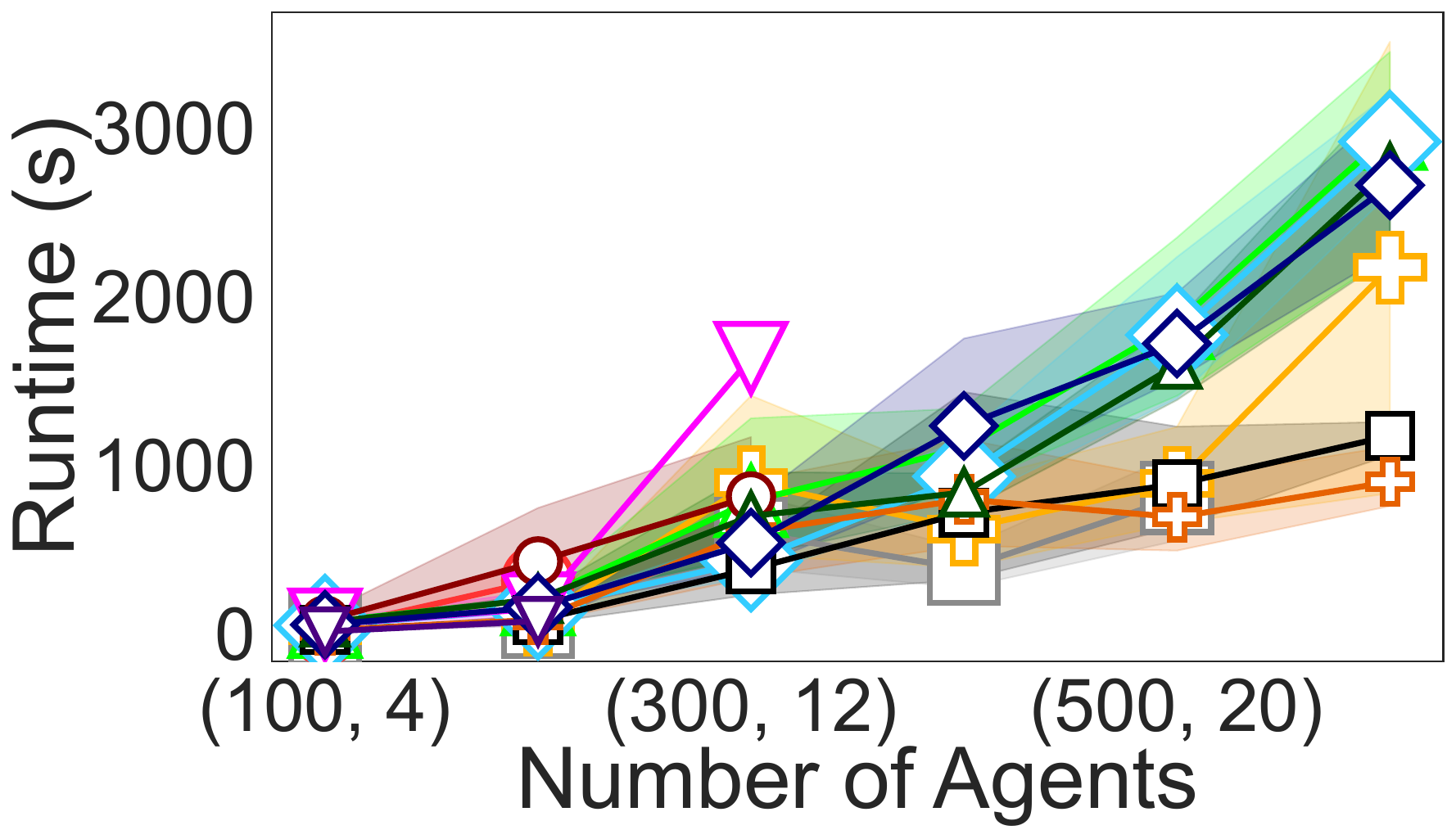}
    \end{subfigure}\hfill
    \begin{subfigure}[t]{0.24\textwidth}
        \centering
        \includegraphics[width=\linewidth]{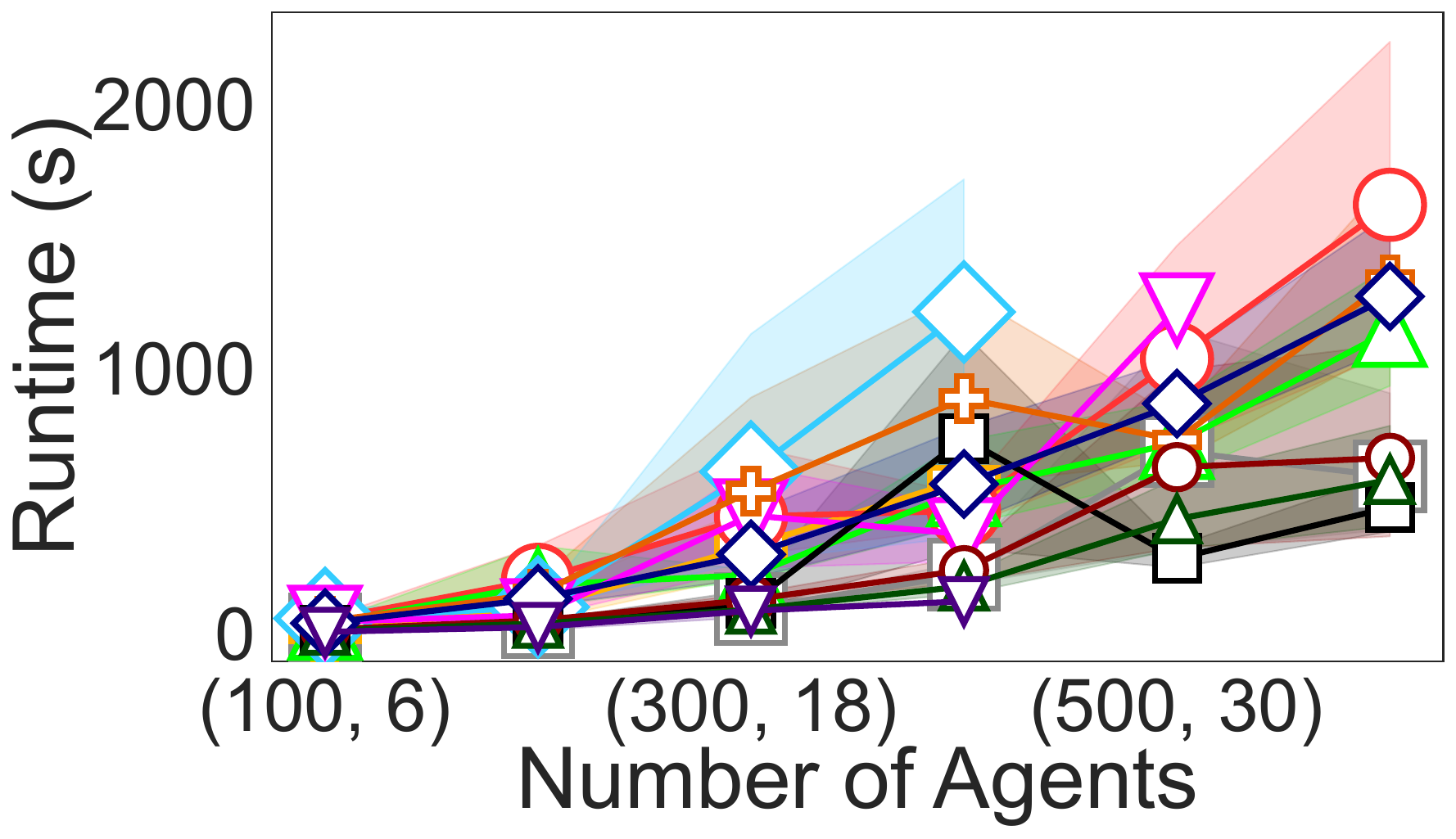}
    \end{subfigure}
    \par\vspace{-\abovecaptionskip}
    \subcaptionbox{\randomLargeLA\label{fig:oneshotmapf-large_agents-random_128_128-add}}
        {\phantom{\rule{0.24\textwidth}{0pt}}}\hfill
    \subcaptionbox{\emptyLargeLA\label{fig:oneshotmapf-large_agents-empty_64_64-add}}
        {\phantom{\rule{0.24\textwidth}{0pt}}}\hfill
    \subcaptionbox{\mazeLargeLA\label{fig:oneshotmapf-large_agents-maze_128_128-add}}
        {\phantom{\rule{0.24\textwidth}{0pt}}}\hfill
    \subcaptionbox{\roomLargeLA\label{fig:oneshotmapf-large_agents-room_128_128-add}}
        {\phantom{\rule{0.24\textwidth}{0pt}}}\hfill
    
    \caption{Success rate, suboptimality of sum-of-cost, makespan and runtime for different numbers of agents for one-shot MAPF with PMLA agents.}
    \label{fig:major-result-oneshotmapf-largeagents-add}
\end{figure*}

\subsection{Compute Resources} \label{appen:compute}
For one-shot experiments, we use a local machine with a 64 core AMD Ryzen 3990X CPU and 192 GB of RAM. For lifelong experiments, we use three machines: (1) a local machine with a 64 core AMD Ryzen 3990X CPU and 128 GB of RAM, (2) a local machine with a 64 core AMD Ryzen 7980X CPU and 256 GB of RAM, and (3) a High Performing Cluster with AMD EPYC 7742 CPUs~\cite{PSCBridgeTwo2021}.

\begin{figure*}[!t]
    \centering
    \includegraphics[width=1\linewidth]{figs/major_result/oneshotmapf/large_agents/bump_distributions_legend.pdf}\par\medskip

    \begin{subfigure}[t]{0.24\textwidth}
        \centering
        \includegraphics[width=\linewidth]{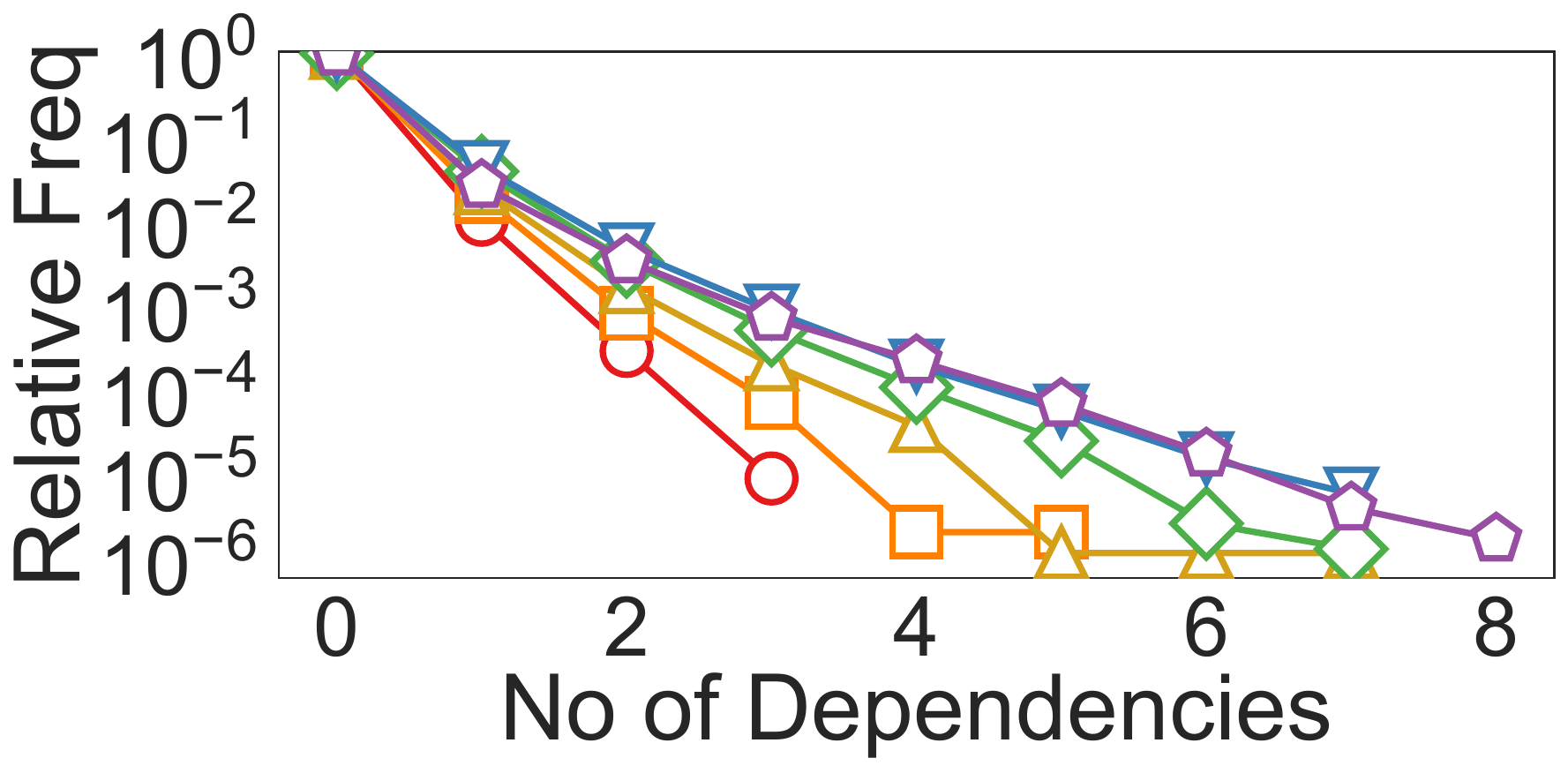}
    \end{subfigure}\hfill
    \begin{subfigure}[t]{0.24\textwidth}
        \centering
        \includegraphics[width=\linewidth]{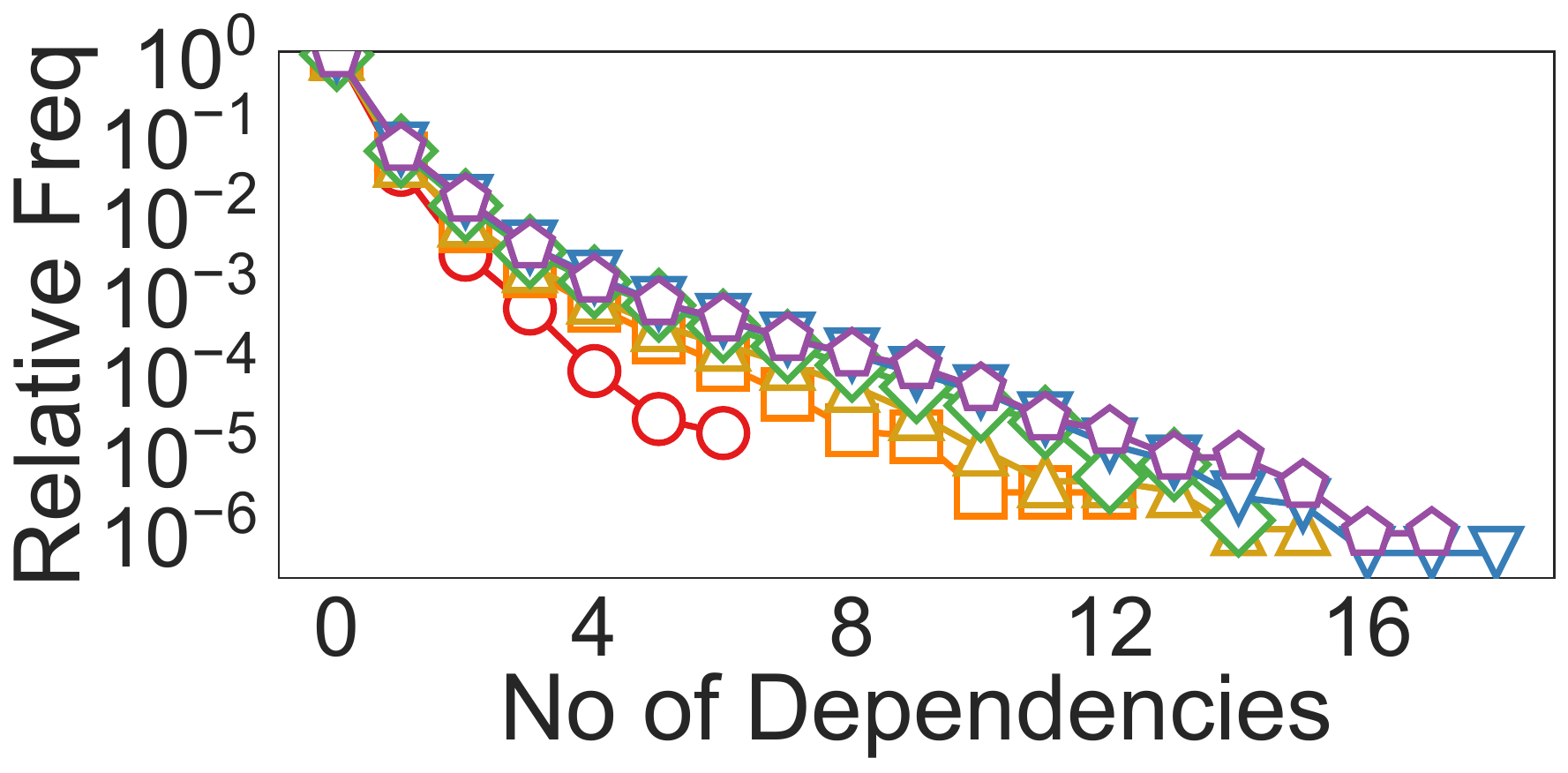}
    \end{subfigure}\hfill
    \begin{subfigure}[t]{0.24\textwidth}
        \centering
        \includegraphics[width=\linewidth]{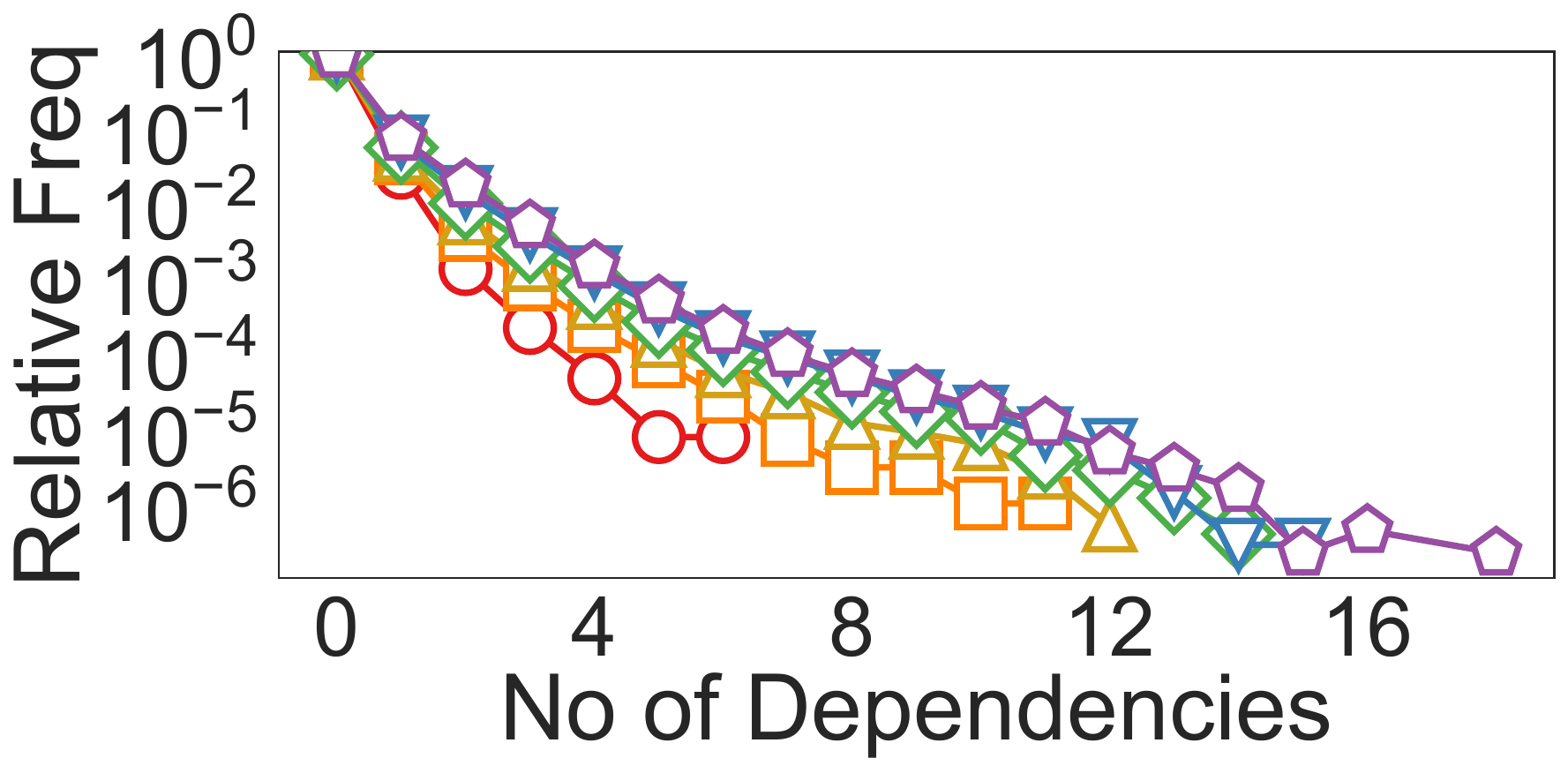}
    \end{subfigure}\hfill
    \begin{subfigure}[t]{0.24\textwidth}
        \centering
        \includegraphics[width=\linewidth]{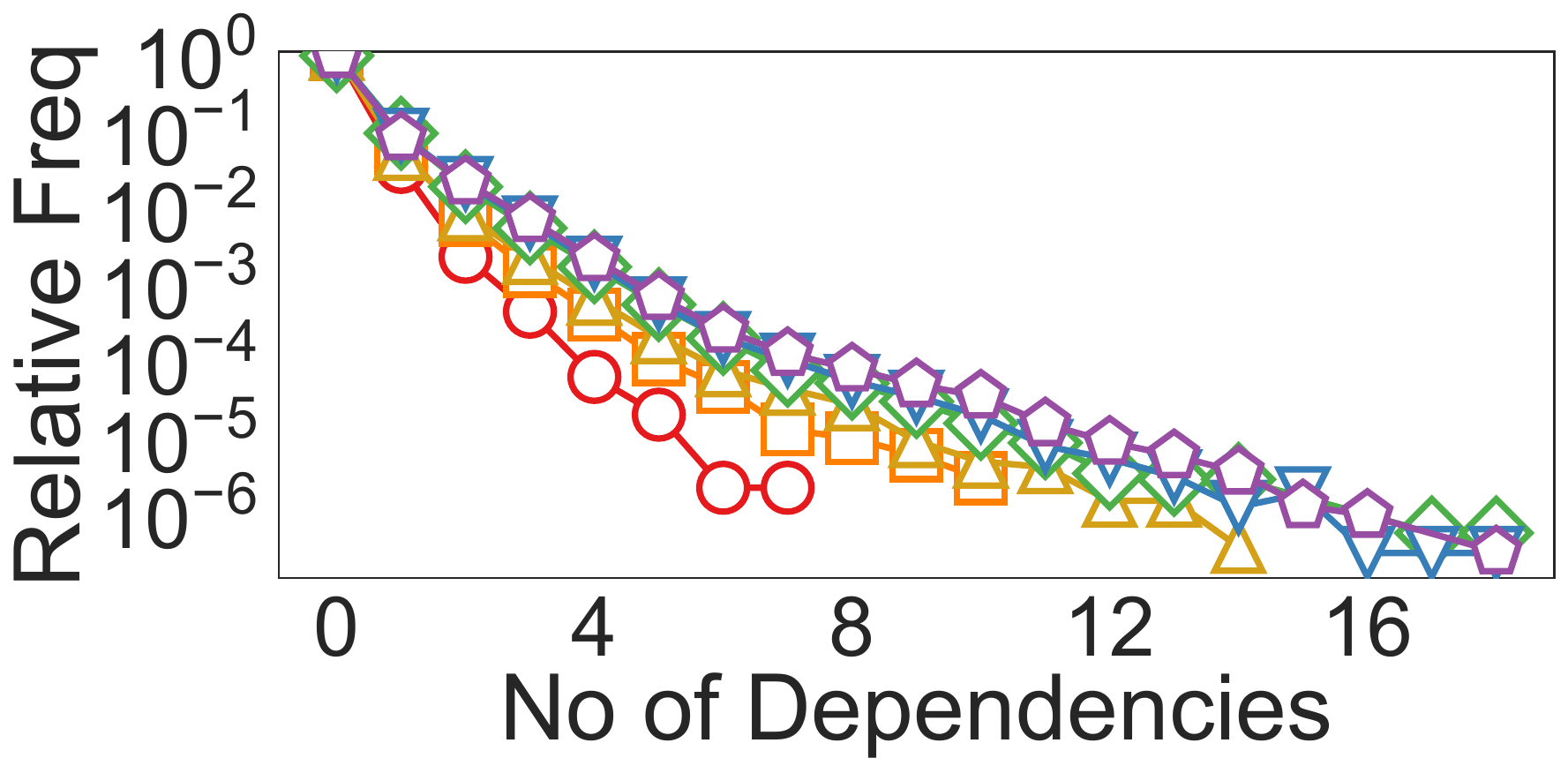}
    \end{subfigure}
    \par\vspace{-\abovecaptionskip}
    \subcaptionbox{\randomLargeLA\label{fig:oneshotmapf-large_agents-random_128_128-bump-add}}
        {\phantom{\rule{0.24\textwidth}{0pt}}}\hfill
    \subcaptionbox{\emptyLargeLA\label{fig:oneshotmapf-large_agents-empty_64_64-bump-add}}
        {\phantom{\rule{0.24\textwidth}{0pt}}}\hfill
    \subcaptionbox{\mazeLargeLA\label{fig:oneshotmapf-large_agents-maze_128_128-bump-add}}
        {\phantom{\rule{0.24\textwidth}{0pt}}}\hfill
    \subcaptionbox{\roomLargeLA\label{fig:oneshotmapf-large_agents-room_128_128-bump-add}}
        {\phantom{\rule{0.24\textwidth}{0pt}}}\hfill    
    \caption{Distribution of the number of agent dependencies generated per agent per successful planning attempt, under different agent densities, with $w = 4$. $S$ and $L$ refers to the number of small and large agents, respectively.}
    \label{fig:major-result-oneshotmapf-largeagents-bump-add}
\end{figure*}

\subsection{Additional Results} \label{appen:add-result-result}

\subsubsection{One-shot MAPF with PM Agents} \label{appen:one-shot-exp-pm}

Rows 1 and 2 of \Cref{tab:exp-setup-add} summarize the additional experiments we conduct for MD-PIBT in one-shot MAPF. We use these configurations to perform a hyper-parameter search in the \paris and \roomLarge maps of the MAPF benchmark~\cite{Stern2019benchmark} with PM agents. In addition to runtime and success rate, we show the solution qualities as \emph{Sum-of-cost / LB} and \emph{Makespan / LB}, which are the (over-estimated) suboptimality of sum-of-cost and makespan, respectively. The sum-of-costs is the sum of each agent's path length while the makespan is the number of timesteps it takes for all agents to reach their goals. LB refers to the lower bound. The lower bound of sum-of-cost is obtained by summing the shortest possible path of each agent, and the lower bound of makespan is obtained by taking the longest shortest path of all agents. If the success rate is less than 1, only the successful runs are used to calculate the suboptimality of sum-of-cost and makespan.

In \Cref{fig:major-result-oneshotmapf-smallagents-add}, we show the full results of the experiment discussed in \Cref{fig:major-result-oneshotmapf-smallagents} of \Cref{sec:exp-oneshot}. In addition to success rates and runtime, we show the suboptimality with respect to both sum-of-cost and makespan.
MD-PIBT achieves equivalent or better success rates compared to EPIBT and PIBT. Meanwhile, MD-PIBT finds solutions equivalent to or higher quality than EPIBT in both sum-of-cost and makespan. The solution quality of PIBT appears to be higher than others because it is averaged only on successful runs, which is a smaller and easier set of MAPF problems than the full set.

To better understand which configuration of MD-PIBT is superior, we show the best sets of hyper-parameters of MD-PIBT in each map with different numbers of agents in \Cref{tab:best-config-oneshot}. For the mode of finding path, $m=\text{EPIBT}$ outperforms $m=\text{PIBT}$ most of the time.
Other parameters in the best configurations vary across maps. In \warehouseXlarge, for example, the long corridors allow MD-PIBT to prefer a longer planning window of $w=4$. In \randomSmall and \roomLarge, on the other hand, using a shorter planning window of $w=3$ is sufficient. For \paris, the map is too large, making a shorter planning window of $w\leq2$ the best option. We do not observe particular patterns in the best value of $C$. We conjecture that allowing one agent to collide with multiple agents during planning is not always helpful in all scenarios.

\subsubsection{One-shot MAPF with PMLA Agents} \label{appen:one-shot-exp-pmla}

In \Cref{fig:major-result-oneshotmapf-largeagents-add}, we show the full results of the experiment discussed in \Cref{fig:major-result-oneshotmapf-largeagents} of \Cref{sec:exp-oneshot}. In addition to success rates and runtime, we show the suboptimality with respect to both sum-of-cost and makespan. We show the average as solid lines and the 95\% confidence interval as shaded areas. Clearly, MD-PIBT variants find paths of higher quality compared to EPIBT. Among the MD-PIBT variants, a higher value of $C$ is generally helpful in terms of searching for better solutions.

In \Cref{fig:major-result-oneshotmapf-largeagents-bump-add}, we perform the same analysis as in \Cref{fig:major-result-oneshotmapf-largeagents-bump} in \Cref{sec:exp-oneshot}, with a different size of the planning window of $w=4$. The trend is similar to the result in \Cref{fig:major-result-oneshotmapf-largeagents-bump}, where more agents correspond to more agent dependencies on all maps.

\begin{figure*}[!tp]
    \centering
    \includegraphics[width=1\linewidth]{figs/major_result/lmapf/pm/legend-pm.pdf}\par\medskip
    \vspace{-0.5em}
    \begin{subfigure}{0.49\textwidth}
        \centering
        \includegraphics[width=0.5\textwidth]{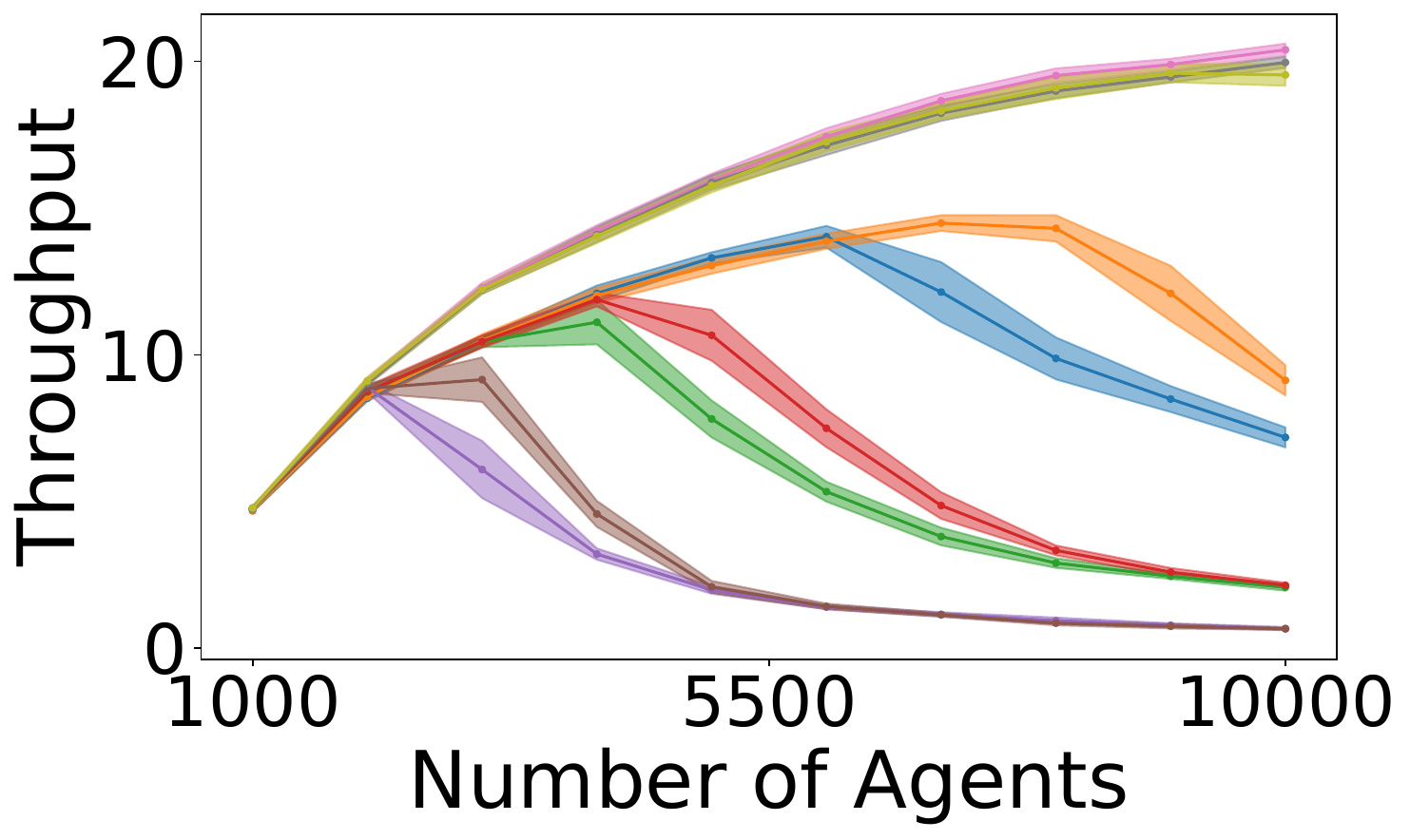}%
        \includegraphics[width=0.5\textwidth]{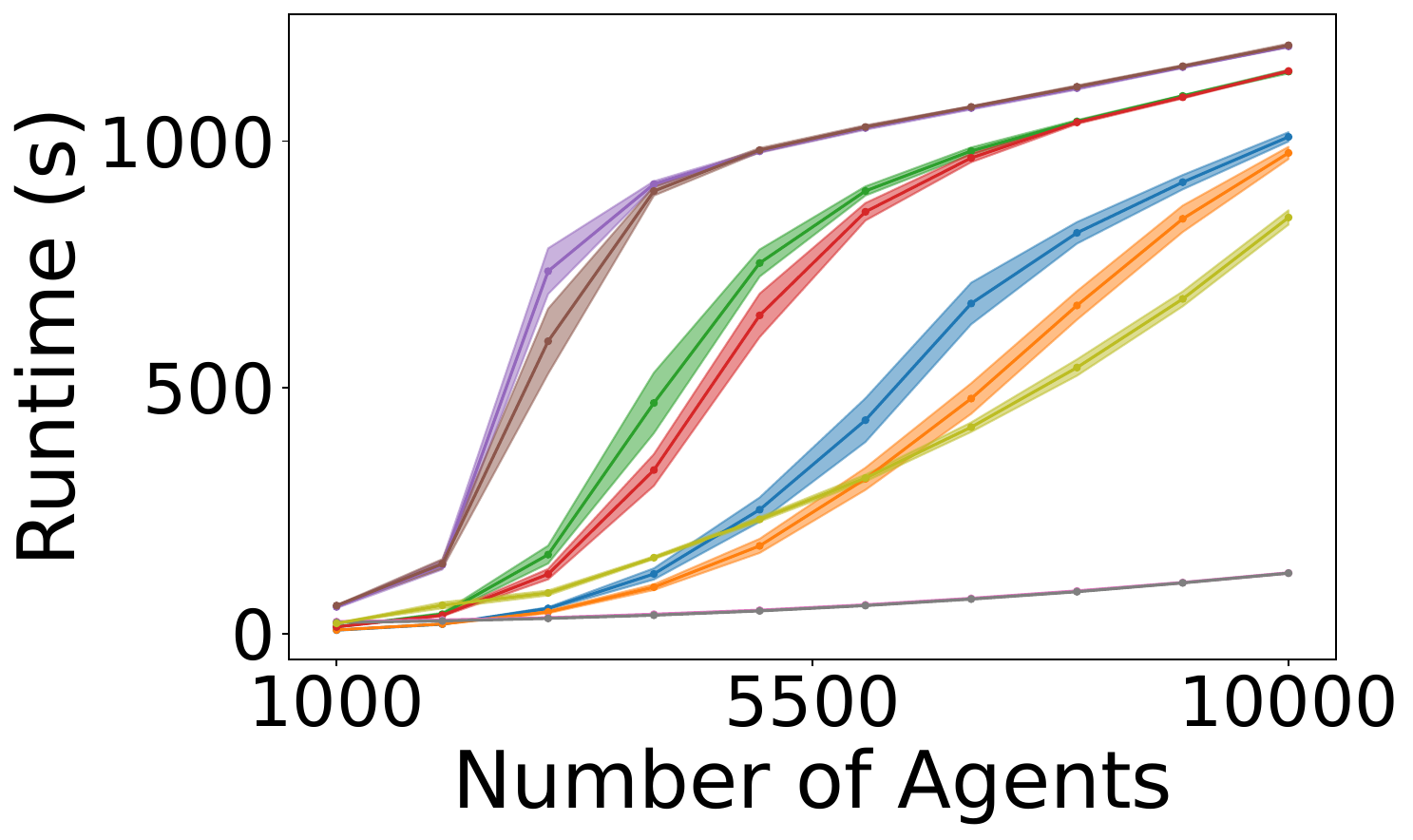}
        \vspace{-1.7em}
        \caption{\paris}
        \label{fig:lmapf-pm-paris-1-256}
    \end{subfigure}%
    \hfill
    \begin{subfigure}{0.49\textwidth}
        \centering
        \includegraphics[width=0.5\textwidth]{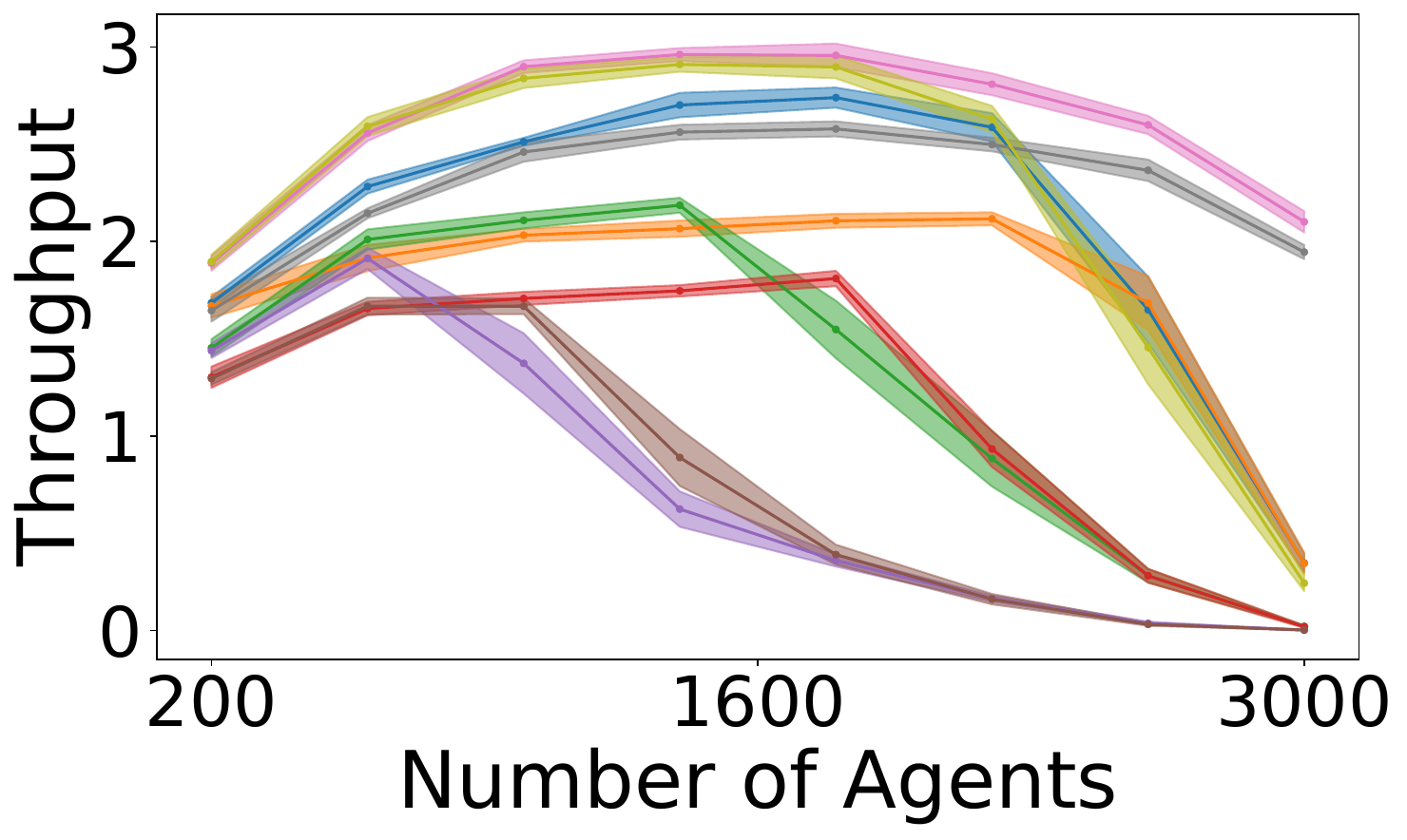}%
        \includegraphics[width=0.5\textwidth]{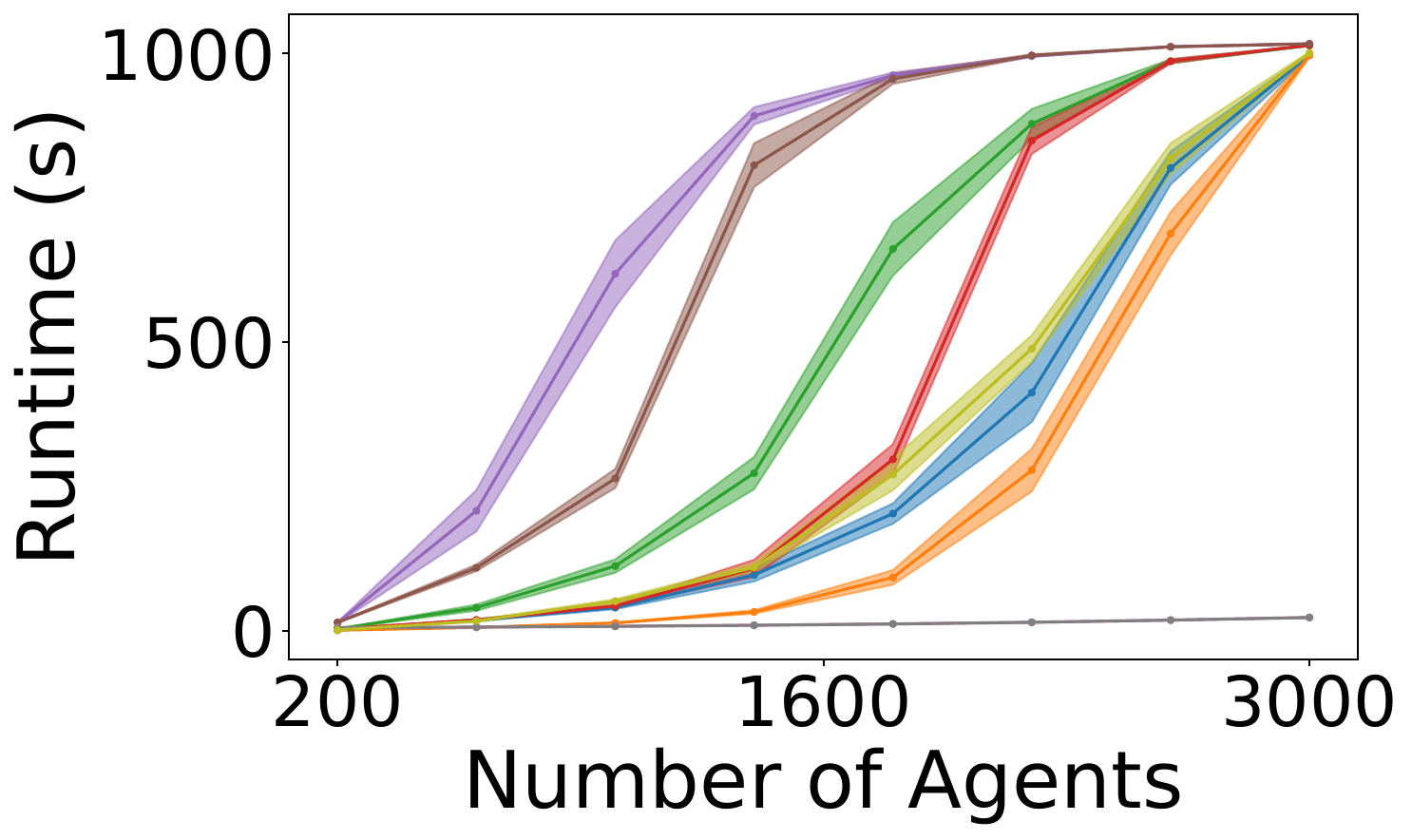}
        \vspace{-1.7em}
        \caption{\roomLarge}
        \label{fig:lmapf-pm-room-64-64-8}
    \end{subfigure}
    \hfill
    \caption{Throughput and runtime with different numbers of agents for LMAPF with PM agents. }
    \label{fig:major-result-lmapf-pm-add}
    \par\vspace{-\abovecaptionskip}
\end{figure*}

\begin{figure*}[!tp]
    \centering
    \includegraphics[width=1\linewidth]{figs/major_result/lmapf/rot/legend-rot.pdf}\par\medskip
    \vspace{-0.5em}
    \begin{subfigure}{0.49\textwidth}
        \centering
        \includegraphics[width=0.5\textwidth]{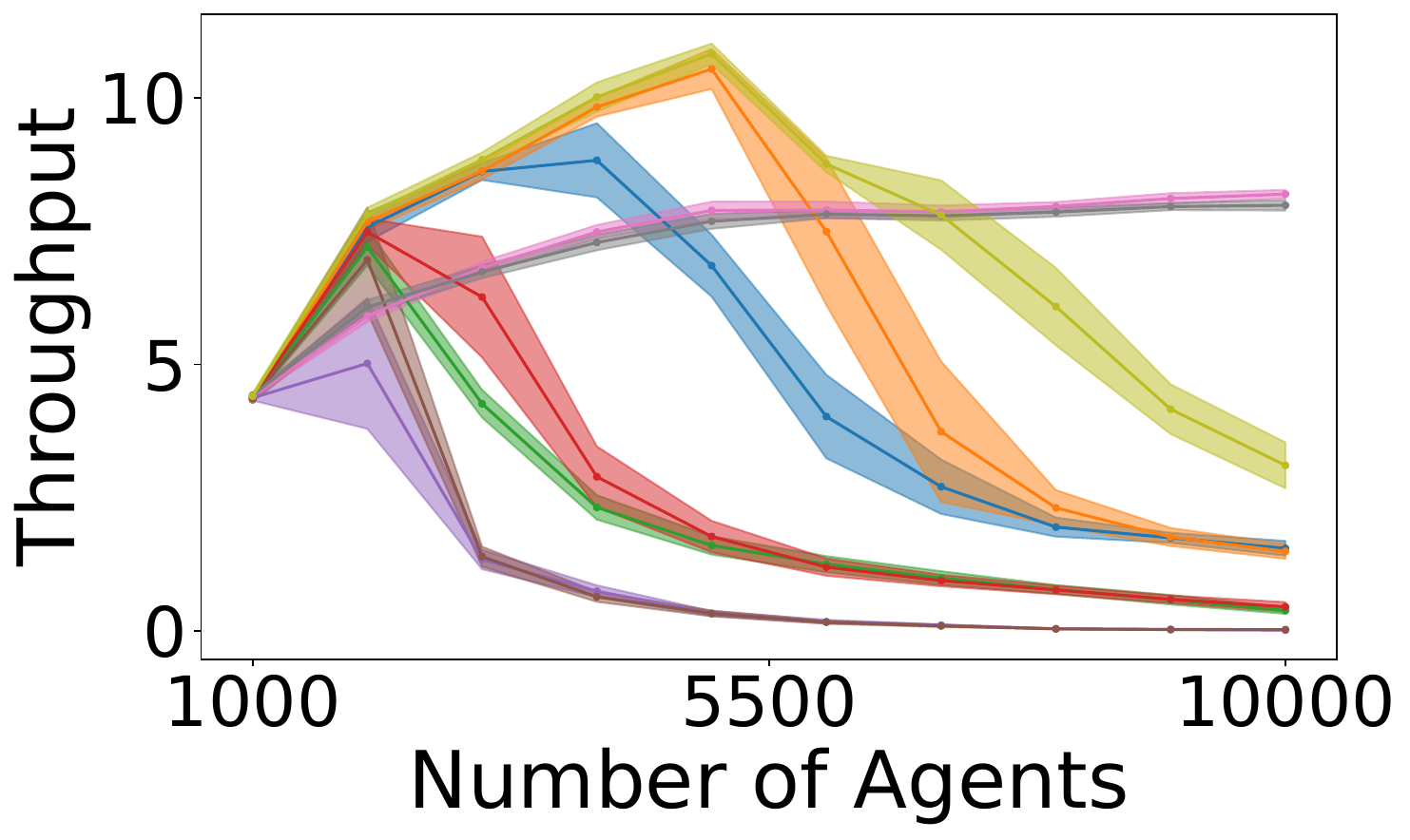}%
        \includegraphics[width=0.5\textwidth]{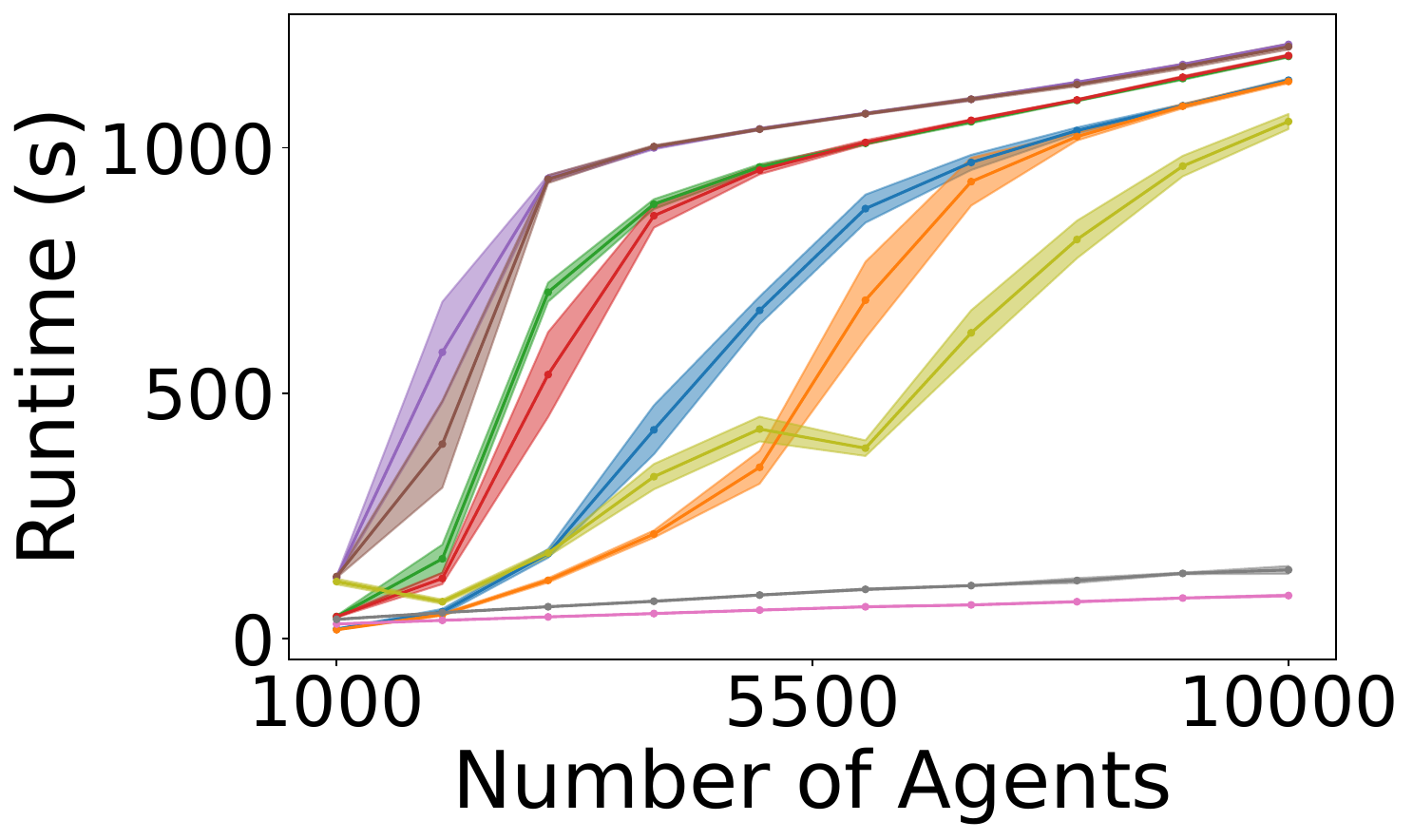}
        \vspace{-1.7em}
        \caption{\paris}
        \label{fig:lmapf-rot-Paris-1-256}
    \end{subfigure}%
    \hfill
    \begin{subfigure}{0.49\textwidth}
        \centering
        \includegraphics[width=0.5\textwidth]{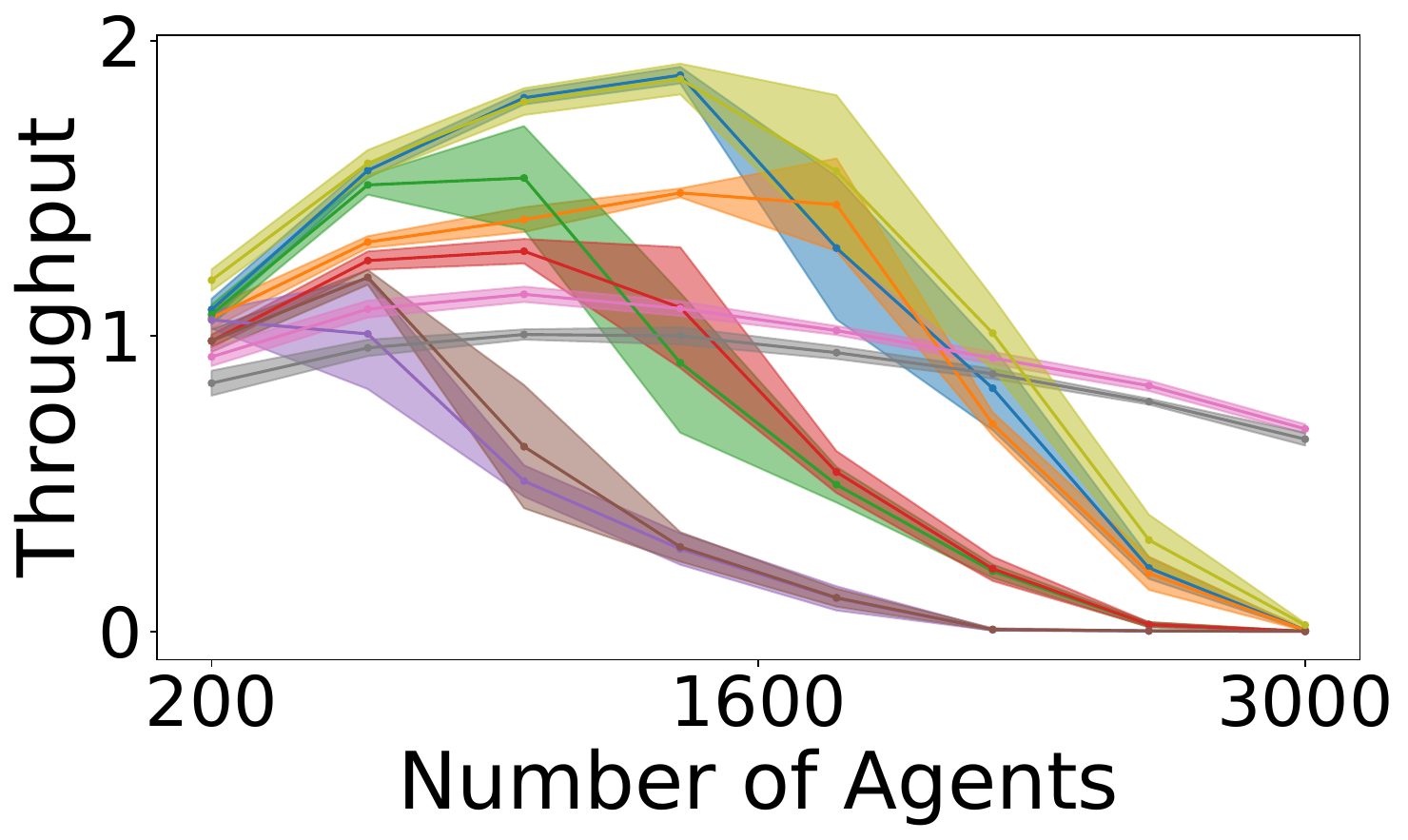}%
        \includegraphics[width=0.5\textwidth]{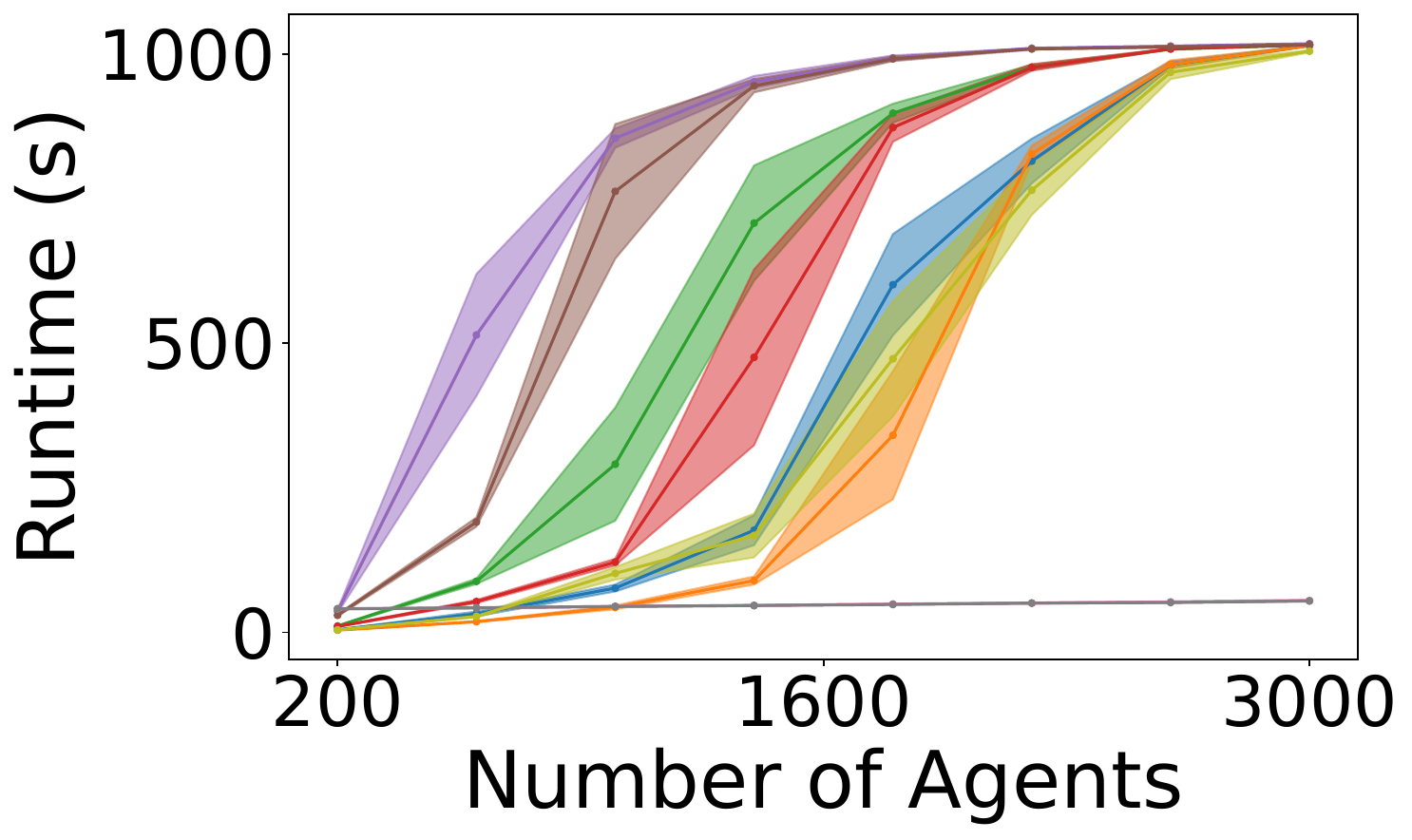}
        \vspace{-1.7em}
        \caption{\roomLarge}
        \label{fig:lmapf-rot-room-64-64-8}
    \end{subfigure}
    \hfill
    \caption{Throughput and runtime with different numbers of agents for LMAPF with RM agents. }
    \label{fig:major-result-lmapf-rot-add}
    \par\vspace{-\abovecaptionskip}
\end{figure*}

\begin{table*}[!tp]
    \centering
    \small
    \resizebox{\linewidth}{!}{
    \begin{tabular}{ccrcc cccc}
    \toprule
    $W$ & $h$ & $C$ & $m$ & $\rho$ & \randomSmall & \roomLarge & \warehouseXlarge & \paris \\
    \midrule
    $1$ & $1$ & $1$      & EPIBT & LET   & & & & 300 \\
    $2$ & $2$ & $1$      & PIBT  & LET   & & & 20 & \\
    $2$ & $2$ & $1$      & EPIBT & LET   & & & & 600 \\
    $2$ & $2$ & $\infty$ & EPIBT & LET   & 40 & 100, 200 & & 100, 200, 400, 500, 700, 800, 900, 1000 \\
    $3$ & $1$ & $1$      & EPIBT & LET   & 80, 120, 280, 360 & & & \\
    $3$ & $1$ & $\infty$ & EPIBT & LET   & 160, 200, 240, 320 & & & \\
    $3$ & $3$ & $\infty$ & EPIBT & LET   & & 300, 400, 500, 600, 700, 800, 900, 1000 & & \\
    $4$ & $1$ & $1$      & EPIBT & LET   & & & 40, 100, 120 & \\
    $4$ & $1$ & $\infty$ & EPIBT & LET   & & & 80, 140, 160, 180 & \\
    $4$ & $4$ & $1$      & EPIBT & LET   & & & 60 & \\
    $4$ & $4$ & $\infty$ & PIBT  & LET   & 400 & & & \\
    \bottomrule
    \end{tabular}
    }
    \caption{
    Best configurations of MD-PIBT for each map for different numbers of agents in one-shot experiments ($R=100$). The best configuration is sorted by success rate, and sum-of-cost and runtime are used as tie-breakers.
    }
    \label{tab:best-config-oneshot}
    \par\vspace{-\abovecaptionskip}
\end{table*}

\begin{table*}[!tp]
    \centering
    \small
    \resizebox{\linewidth}{!}{
    \begin{tabular}{ccrcc cccc}
    \toprule
    $W$ & $h$ & $C$ & $m$ & $\rho$ & \randomSmall & \roomLarge & \warehouseXlarge & \paris \\
    \midrule
    $1$ & $1$ & $1$      & PIBT  & SD    & 800 & 600, 1000, 1400, 2200 & & 4000, 5000, 6000, 7000, 9000, 10000 \\
    $1$ & $1$ & $1$      & PIBT  & LET   & & 2600 & 1800, 2200, 2600 & \\
    $1$ & $1$ & $\infty$ & PIBT  & SD    & 400, 500, 600, 700 & 200, 1800 & & 8000 \\
    $1$ & $1$ & $\infty$ & PIBT  & LET   & & 3000 & 3000 & \\
    $2$ & $2$ & $\infty$ & PIBT  & LET   & & & & 3000 \\
    $3$ & $1$ & $1$      & PIBT  & SD    & & & 600 & \\
    $3$ & $1$ & $1$      & PIBT  & LET   & & & 1400 & \\
    $3$ & $1$ & $1$      & EPIBT & SD    & & & 1000 & \\
    $3$ & $1$ & $\infty$ & EPIBT & SD    & & & 200 & \\
    $3$ & $3$ & $\infty$ & PIBT  & SD    & 200 & & & \\
    $3$ & $3$ & $\infty$ & PIBT  & LET   & 300 & & & 2000 \\
    $3$ & $3$ & $\infty$ & EPIBT & LET   & 100 & & & 1000 \\
    \bottomrule
    \end{tabular}
    }
    \caption{
    Best configurations of MD-PIBT for each map for different numbers of agents in LMAPF experiments with PM agents. The best configuration is sorted by throughput.
    }
    \label{tab:best-config-pm}
    \par\vspace{-\abovecaptionskip}
\end{table*}

\begin{table*}[!tp]
    \centering
    \small
    \resizebox{\linewidth}{!}{
    \begin{tabular}{ccrcc cccc}
    \toprule
    $W$ & $h$ & $C$ & $m$ & $\rho$ & \randomSmall & \roomLarge & \warehouseXlarge & \paris \\
    \midrule
    $3$ & $1$ & $1$      & PIBT  & SD    & 300 & 200 & & \\
    $3$ & $1$ & $1$      & PIBT  & LET   & & & 3000 & 2000, 3000, 4000 \\
    $3$ & $1$ & $1$      & EPIBT & SD    & 400, 500, 600, 700 & 600, 1000, 1400, 1800, 2200, 2600 & & \\
    $3$ & $1$ & $1$      & EPIBT & LET   & & & 1800, 2200, 2600 & 5000 \\
    $3$ & $3$ & $1$      & PIBT  & SD    & 800 & 3000 & & \\
    $3$ & $3$ & $1$      & PIBT  & LET   & & & & 6000, 7000, 8000, 9000, 10000 \\
    $5$ & $1$ & $1$      & PIBT  & SD    & & & & 1000 \\
    $5$ & $1$ & $1$      & PIBT  & SD    & 100 & & 200, 600 & \\
    $5$ & $1$ & $1$      & EPIBT & SD    & 200 & & 1000 & \\
    $5$ & $1$ & $1$      & EPIBT & LET   & & & 1400 & \\
    \bottomrule
    \end{tabular}
    }
    \caption{
    Best configurations of MD-PIBT for each map for different numbers of agents in LMAPF experiments with RM agents. The best configuration is sorted by throughput.
    }
    \label{tab:best-config-rm}
    \par\vspace{-\abovecaptionskip}
\end{table*}

\subsubsection{Lifelong MAPF} \label{appen:lifelong-exp}

Rows 3 and 4 of \Cref{tab:exp-setup-add} summarize the additional experiments we conduct for MD-PIBT in lifelong MAPF. We use these configurations to perform a hyper-parameter search in the \paris and \roomLarge maps of the MAPF benchmark~\cite{Stern2019benchmark} with PM and RM agents.

\Cref{fig:major-result-lmapf-pm-add,fig:major-result-lmapf-rot-add} show the results of LMAPF with PM and RM agents, respectively. With PM agents, PIBT remains surprisingly powerful in terms of throughput and runtime, and our MD-PIBT replication of PIBT closely follows. With RM agents, EPIBT and our MD-PIBT replication of EPIBT generally achieve the highest throughput. \Cref{tab:best-config-pm,tab:best-config-rm} summarizes the best MD-PIBT configurations with PM and RM agents, respectively. For PM agents, interestingly, the MD-PIBT variant with $w=h=1$, $m=\text{PIBT}$, and $C=1$, which replicates PIBT, is quite powerful, especially in larger maps with many agents. The benefit of $C=\infty$ and the higher values of $w$ and $h$ can be seen with a moderate number of agents on different maps (e.g, 100-300 in \randomSmall, 1000-1400 in \warehouseXlarge and 1000-2000 in \paris). For RM agents, the MD-PIBT variant with $w\geq1$, $h=1$, $C=1$, and $m=\text{PIBT}$, which replicates EPIBT, achieves the highest throughput most of the time. Surprisingly, $C=\infty$ never outperforms $C=1$. We conjecture that resolving multiple collisions per agent is more challenging with RM agents than with PM agents. 

\end{document}